\documentclass{article}

\usepackage{comment}
\usepackage{amsfonts}
\usepackage{graphics}
\usepackage{xcolor}
\usepackage{amssymb}
\usepackage{verbatim}
\usepackage{amsmath}
\usepackage{esint}
\usepackage{amssymb}
\usepackage{amsthm}
\usepackage{amscd}
\usepackage{euscript}
\usepackage{enumerate}
\usepackage{mathtools}

\usepackage[all]{xy}

\usepackage{hyperref}

\newcommand{\hol}{{\rm Hol}}  % holonomy hook
\def \s {\sigma}
\def \t {\tau}

\def \M {{\cal M}}
\usepackage{rlepsf}
\usepackage{float}
\usepackage[all]{xy}

\usepackage{amsfonts}
\usepackage{amsmath}

\newcommand{\hlgt}{higher lattice gauge theory }

\makeatletter
\newcommand*{\da@rightarrow}{\mathchar"0\hexnumber@\symAMSa 4B }
\newcommand*{\da@leftarrow}{\mathchar"0\hexnumber@\symAMSa 4C }
\newcommand*{\xdashrightarrow}[2][]{%
  \mathrel{%
    \mathpalette{\da@xarrow{#1}{#2}{}\da@rightarrow{\,}{}}{}%
  }%
}
\newcommand{\xdashleftarrow}[2][]{%
  \mathrel{%
    \mathpalette{\da@xarrow{#1}{#2}\da@leftarrow{}{}{\,}}{}%
  }%
}
\newcommand*{\da@xarrow}[7]{%
  \sbox0{$\ifx#7\scriptstyle\scriptscriptstyle\else\scriptstyle\fi#5#1#6\m@th$}%
  \sbox2{$\ifx#7\scriptstyle\scriptscriptstyle\else\scriptstyle\fi#5#2#6\m@th$}%
  \sbox4{$#7\dabar@\m@th$}%
  \dimen@=\wd0 %
  \ifdim\wd2 >\dimen@
    \dimen@=\wd2 %
  \fi
  \count@=2 %
  \def\da@bars{\dabar@\dabar@}%
  \@whiledim\count@\wd4<\dimen@\do{%
    \advance\count@\@ne
    \expandafter\def\expandafter\da@bars\expandafter{%
      \da@bars
      \dabar@
    }%
  }%
  \mathrel{#3}%
  \mathrel{%
    \mathop{\da@bars}\limits
    \ifx\\#1\\%
    \else
      _{\copy0}%
    \fi
    \ifx\\#2\\%
    \else
      ^{\copy2}%
    \fi
  }%
  \mathrel{#4}%
}
\makeatother
\def \N {\mathbb{N}}
\usepackage[utf8]{inputenc}
\usepackage{dashbox}

\def \sgn {\mathrm{sgn}}

\def \trrr {\bullet}  %% We agree that ordinary group action needs no symbol.
\newcommand{\dvQ}{\d_v^Q}  %% oriented boudanry operator
\newcommand{\dv}{\d_v}  %% oriented boudanry operator
\newcommand{\iv}{\iota_v}  %% oriented boudanry operator
\newcommand{\ivb}{\overline{\iota_v}}  %% oriented boudanry operator
\newcommand{\ivx}{\overline{\iota_x}}
\newcommand{\ivbp}{\overline{\iota_{v'}}}  %% oriented boudanry operator
\newcommand{\dvp}{\d_{v'}}  %% oriented boudanry operator
\newcommand{\ivp}{\iota_{v'}}  %% oriented boudanry operator

\newcommand{\htwo}{ {g_{\gamma_1} \trr e_{P_1}^{\theta_1}\,\,g_{\gamma_2} \trr  e_{P_2}^{\theta_2}\,\,
                       \dots\,\, g_{\gamma_N} \trr  e_{P_N}^{\theta_N} }}

\newcommand{\be}{\begin{equation}}
\newcommand{\ee}{\end{equation}}
\newcommand{\bq}{\begin{eqnarray}}
\newcommand{\eq}{\end{eqnarray}}

\newcommand{\trr}{\triangleright}

\newcommand{\Z}{{\mathbb Z}}

\newtheorem{Theorem}{Theorem}
\newtheorem{lemma}[Theorem]{Lemma}
\newtheorem{proposition}[Theorem]{Proposition}

\theoremstyle{definition}  %% non-italic block type for subsequent environments

\newtheorem{definition}[Theorem]{Definition}
\newtheorem{example}[Theorem]{Example}

\newtheorem{remark}[Theorem]{Remark}

\newcommand{\bluex}[1]{}

\newcommand{\redx}[1]{}

\newcommand{\dd}{d}  %% boundary map hook for double groupoids

\def \d {\partial}
\def \bound {\mathrm{bd}}
\newcommand{\dL}{\d_L}  %% disambiguate \partial???
\newcommand{\iL}{\iota_L}
\def \ra {\xrightarrow}
\def \Dna {{\rm int}({D^n})}
\def \D {{\cal D}}
\def \A {{\cal A}}
\def \B {{\cal B}}
\def \g {\gamma}
\def \Gc {{\cal G}}
\def \C {{\cal C}}
\def \H {{\cal H}}

\def \fPhi {\Theta_{\rm 2flat}}
\def \U {{\cal U}}

\def \F {{\cal F}}
\def \id {{\mathrm{id}}}
\def \l {\lambda}

\def \aut {{\rm Aut}}
\usepackage{amscd}
\usepackage{authblk}
\usepackage{soul}

\newcommand{\freegl}{\fg{ L^0,L^1 }}   %% free groupoid on L
\newcommand{\catfun}{(\freegl,G)}  %% cat of functors
\newcommand{\CC}{{\mathbb C}}
\newcommand{\RR}{{\mathbb R}}
\newcommand{\vac}{\Omega_1}   %% naive vacuum state
\newcommand{\dG}{\partial_{\Gc}}  %% partial for crossed modules
\newcommand{\PhiM}{\Theta(M,L,\Gc ) }  % set of ffgc's
\newcommand{\ignorex}[1]{}
\newcommand{\E}{{\mathsf E}}  %% a set of edges incident at vrtex

\usepackage{amsmath}
{\catcode`\@=11 \uccode`9=`\l \uccode`8=`\o %
 \uppercase{\gdef\striplong@#1#2#3#4\relax{%
  \ifx9#2\ifx8#3\@xp\@xp\@xp\@xp\@xp\@xp\@xp\zap@to@space\fi\fi}}}

\begin{document}
\title{Higher lattices, discrete two-dimensional holonomy and topological phases in (3+1)D with higher gauge symmetry}
\author[1,2]{Alex Bullivant\footnote{E-mail address: 
py09alb@leeds.ac.uk}}
\affil[1]{School of Physics and Astronomy, University of Leeds, Leeds, LS2 9JT, United Kingdom}
\author[2,3]{Marcos Cal\c cada\footnote{E-mail address: mcalcada@uepg.br}}
%\affil{School of Physics and Astronomy, University of Leeds, Leeds, LS2 9JT, United Kingdom}
\author[2]{Zolt\'an K\'ad\'ar\footnote{E-mail address: Zoltan.Kadar@morganstanley.com}}
\affil{Department of Pure Mathematics, University of Leeds, Leeds, LS2 9JT, United Kingdom}
\author[2]{Jo\~ao Faria Martins\footnote{E-mail address: j.fariamartins@leeds.ac.uk}}
\author[2]{Paul Martin\footnote{E-mail address: 
p.p.martin@leeds.ac.uk}}
%\af fil{School of Physics and Astronomy, University of Leeds, Leeds, LS2 9JT, United Kingdom}
\affil{Departamento de Matem\'atica e Estat\'{\i}stica,
  Universidade Estadual de Ponta Grossa, Ponta Grossa, PR,
  Brazil}
\date{\today}
\maketitle
\begin{abstract}Higher gauge theory is a higher order version of gauge theory that makes possible the definition of 2-dimensional holonomy along surfaces embedded in a manifold where a gauge 2-connection is present. In this paper, we will continue the study of Hamiltonian models for discrete higher gauge theory on a lattice decomposition of a manifold.
 In particular, we show that a previously proposed construction for higher lattice gauge theory is well-defined, including in particular a Hamiltonian for topological phases of matter in 3+1 dimensions. Our construction builds upon the Kitaev quantum double model, replacing the finite gauge connection with a finite gauge 2-group 2-connection. Our Hamiltonian higher lattice gauge theory model is defined on spatial manifolds of arbitrary dimension presented by slightly {\it combinatorialised} CW-decompositions (2-lattice decompositions), whose 1-cells and 2-cells carry discrete 1-dimensional and 2-dimensional holonomy data.  We prove that the ground-state degeneracy of Hamiltonian higher lattice gauge theory is a topological invariant of manifolds, coinciding with the number of homotopy classes of maps from the manifold to the classifying space of the underlying gauge 2-group.

The operators of our Hamiltonian model 
are closely related to discrete 2-dimensional holonomy
operators for discretised 2-connections on manifolds with a 2-lattice
decomposition.  
We therefore address the
definition of discrete 2-dimensional holonomy for surfaces embedded
 in 2-lattices.
 Several 
 results concerning the
 well-definedness of discrete 2-dimensional holonomy, and its
 construction in a combinatorial and algebraic topological setting
 are presented.
 \end{abstract}


%#################################################################################################################################
\section{{Introduction}}

In the absence of external symmetries, a topological phase of matter 
is characterised by a local, gapped, quantum many-body Hamiltonian
whose effective (infra-red) field theory is described by a topological
quantum field theory
(TQFT) \cite{Fradkin,Nayak,williamson2016hamiltonian,Wangtqc}.
A {topological phase} is therefore  diffeomorphism invariant,
and thus insensitive to local perturbations in the sense
that the amplitudes of physical processes are global topological
invariants.
It is this latter property which makes
{topological phases of matter}   
candidates for the implementation of fault
tolerant quantum
computing \cite{Wangtqc,pachos2012introduction,kitaev2003fault,Nayak}. 

Due to a lack of local observables, experimentally distinguishing
different topological phases can be a difficult task from a
microscopic point of view. Instead the characterising properties of
topological phases are most efficiently described by their emergent
behaviours. Signatures for the presence of topological order include
ground state degeneracies which depend on the spatial topology of the
material in question \cite{wen1990ground,kapustin2014ground},
universal negative corrections to the entanglement
entropy \cite{KitaevPreskill,HammaZanardi,LevinWen} and the presence
of stable topological excitations which provide non-trivial
representations of their respective motion
groups \cite{kong2014braided,kadar2016local}.  %ie.
This means the braid group for
point particles (anyons) in 2+1D  \cite{Nayak} and the  loop braid
group for loop excitations in 3+1D  \cite{WalkerWang}. 

In 2+1D there exist 
several constructions for TQFTs
(see for instance  \cite{Turaev}).
Path-integral models % for such theories
arise
from Chern-Simons-Witten theory  \cite{witten1989quantum} and from BF
theory  \cite{baezBF}, {while the} discrete realisation of BF-theory
coincides with the
Turaev-Viro \cite{Turaev}/Barrett-Westbury \cite{BW1996} state-sum
(see  \cite{Baez_BF2}).
{In contrast, in 3+1D 
a framework  general enough to capture all features of 4D
topology is still lacking.}
Nevertheless, we have the Crane-Yetter
TQFT \cite{CraneYetterTQFT,Turaev} and its
generalisations  \cite{williamson2016hamiltonian,cui2016higher}; and
the Yetter homotopy 2-type
TQFT  \cite{yetter_tqft,porter_tqft,martins_porter},
derived from a strict finite 2-group  \cite{baez_lauda}.
All of these 3+1D TQFTs give rise
to topological invariants which  at most depend on the homotopy
2-type, signature and spin-structure of
space-time  \cite{Turaev,martins_porter}, or are conjectured to do
so. 

{One successful approach to understanding candidate models for
$(d+1)$D topological phases has been to define Hamiltonian
realisations of $(d+1)$D TQFTs
\cite{Rowell_Wang,williamson2016hamiltonian}. This  means that a
finite dimensional Hilbert space $V(M,L)$, and an exactly solvable
(the sum of mutually commuting projectors)  Hamiltonian $H_L\colon
V(M,L) \to V(M,L)$ is assigned to each $d$-manifold $M$, with a given
{\em lattice decomposition} $L$  (e.g. $L$ can be a triangulation or a
CW-decomposition of $M$).
The constructions of both $V(M,L)$ and
$H_L$ should be local on 
$L$.
To say that
such a Hamitonian schema \cite{Rowell_Wang} is a realisation of the
TQFT ${\cal Z}$ roughly means that given a $d$-manifold $M$ each
ground state vector space $GS(M,L)$ of
$H_L$
is canonically isomorphic to ${\cal Z}(M)$.
(In particular this implies that the ground state degeneracy
$\dim(GS(M,L))$ does not depend on  %the lattice decomposition
$L$  %of $M$
and it is a topological invariant of $M$.) 
In 2+1D, this {Hamiltonian realisation} approach has been
successfully achieved in the case of Dijkgraaf-Witten topological gauge
theories \cite{dijkgraaf1990topological,kitaev2003fault,hu2013twisted,wan2015twisted}
and the Turaev-Viro TQFT, giving rise to the so-called string-net
models \cite{Stringnet}.
The % celebrated
Kitaev quantum-double
model  \cite{kitaev2003fault} can be seen as a Hamiltonian realisation
for the  Dijkgraaf-Witten TQFT   \cite{dijkgraaf1990topological} with
trivial cocycle and thus also for finite-group BF-theory. 
Similar ideas were applied to the 3+1D  Crane-Yetter TQFT  \cite{WalkerWang,Simon}, giving rise to the Walker-Wang model.}

A Hamiltonian realisation of  Yetter's homotopy 2-type TQFT was
constructed in  \cite{williamson2016hamiltonian,BCKMM}. This is a
higher gauge theory version of Kitaev quantum-double
model \cite{kitaev2003fault}.
It is this that we continue to  develop in this paper. {We note that topological phases protected by higher gauge symmetry are also proposed in \cite{C}.}

Higher gauge theory \cite{BaezHuerta11,baez_schreiber} is a
generalisation of ordinary gauge theory with further levels of
structure and symmetry. A key feature of higher gauge theory is
parallel transport along surfaces embedded in a manifold where a gauge
2-connection is
present \cite{BaezHuerta11,baez_schreiber,martins_picken,Zu1}. In
higher gauge theory, instead of local gauge symmetry groups we have
local gauge symmetry 2-groups.
{These, recall, are equivalent to crossed modules of groups $\Gc=(\d\colon E  \to
G,\trr)$ \cite{brown_higgins_sivera,Barrett_Mackaay,baez_lauda}. (Recall
$\d\colon E \to G$ is a map of groups and $\trr$ is a left action of
$G$ on $E$ by automorphisms, satisfying some compatibility relations:
the 1st and 2nd Peiffer relations.).}

{
In this paper,
completing the programme initiated in \cite{BCKMM},
we define an exactly solvable Hamiltonian model for \hlgt  on
manifolds $M$ of any dimension, {here called the ``higher Kitaev model''}. This lifts Kitaev's quantum-double
model  \cite{kitaev2003fault} for topological phases from finite
group topological gauge theory to finite 2-group topological higher
gauge theory \cite{BFCG1,BFCG2}. Typically $M$ will be a 3-dimensional manifold, and the higher Kitaev model is proposed to be a model for (3+1)-dimensional topological phases \cite{WalkerWang,williamson2016hamiltonian,wan2015twisted,BCKMM,Simon,kong2014braided,SimonPhase,C,D}.}
We  prove that the ground state
degeneracy of the {higher Kitaev model} is a topological (in fact homotopy) invariant
of manifolds. Specifically, we show that the ground state degeneracy
is given by the number of homotopy classes of maps from $M$ to the
classifying space of the underlying gauge symmetry
2-group \cite{brown_higgins_sivera,martins_porter}; hence the ground
state degeneracy is closely related \cite{martins_porter} to Yetter
homotopy 2-type TQFT.
(The precise relation appears in \cite{BCKMM}, where a {proof} 
that {the higher Kitaev model} is a Hamiltonian realisation of Yetter
homotopy 2-type TQFT is given.)

\subsection{{Higher lattices and higher lattice gauge theory}}

Our model utilises ideas from
higher lattice gauge theory  {\cite{Pfeiffer,C}.} Similarly
to   \cite{kitaev2003fault}, we take lattice gauge theory as the
starting point (with its good connection to physical
observation \cite{Wilson74,Kogut79}) and  %mathematically
lift the
structure through the process of categorification. We thereby replace
the gauge group with a gauge 2-group and a gauge connection
discretised on a lattice with a discretised higher gauge connection,
here called a {\em fake-flat 2-gauge configuration}. {Therefore we
enrich the local variables of lattice gauge theory (holonomies along
edges) to include non-abelian 2-dimensional holonomies along the faces
of the lattice;  recall again that 2-dimensional holonomies feature
prominently in higher gauge
theory; see  \cite{baez_schreiber,martins_picken,BaezHuerta11,Zu1}.} 

%\medskip

{The model constructed in this paper extends {and formalises} the
proposal of  \cite{BCKMM} for a Hamiltonian model for the Yetter
homotopy 2-type TQFT \cite{yetter_tqft}, from  triangulated manifolds
to manifolds with a slightly combinatorialized version of
CW-decompositions (here called {\em 2-lattice decompositions} --- see
Definition \ref{2-lattice}). {Hence a 2-lattice decomposition  $L$ represents a manifold $M$ as a disjoint union of $i$-cells, here $i$ is an arbitrary non-negative integer, where each $i$-cell homeomorphic to the interior of the $i$-disk $[0,1]^i$. As costumary, $0$-cells are called vertices, 1-cells are called edges, 2-cells are called faces (or plaquettes), and 3-cells are called blobs. These 2-lattice decomposition are considerably less rigid than triangulations. Therefore}  
using 2-lattice decompositions of manifolds, as opposed to triangulations, 
to decompose a manifold into smaller pieces has the advantage that
fewer cells are needed to decompose a manifold, leading to microscopic
Hilbert spaces of much smaller rank. We illustrate this fact
by  {describing two small models for discrete higher gauge theory
in the 3-sphere}.} 

{Many constructions in this paper would still work if we use
CW-complex decompositions of manifolds rather than 2-lattice
decomposition; however a lot of the combinatorial flavour presented in
the final construction of the Hamiltonian model would be lost.} 
{By using 2-lattice decompositions instead of triangulations some
combinatorics is taken away; therefore, despite the fact that our model is
fully combinatorial, some  %basic
algebraic topology will be required
in proving that it is well-defined.}

{By passing from triangulations to 2-lattices, we hence
demonstrate the internal consistency of the model in  \cite{BCKMM},
which tacitly assumed that discrete 2-dimensional holonomy of a
discrete higher gauge field is well-defined, for instance when proving
in {\it loc cit} that the ground state degeneracy is a topological
invariant derived from Yetter TQFT, and as such {that} our model is a
Hamiltonian realisation of Yetter TQFT.}

{Prominent in this paper is the  %already mentioned
concept of a fake-flat 2-gauge configuration in a 2-lattice, to be a
discretised model of a higher gauge field; as well as the construction
of discrete 2-dimensional holonomy operators for surfaces cellularly
embedded in a 2-lattice. (Fake-flat 2-gauge configurations are in line
with the framework for higher lattice gauge theory
of \cite{Pfeiffer,martins_porter} and also appear in formal homotopy
quantum field theory constructions; see \cite{PorterTuraev}). We
carefully construct these discrete 2-dimensional holonomy operators,
in an algebraic topological (\S \ref{atd}) and  in a combinatorial
manner (\S \ref{cd}), and, using %appealing to some basic
algebraic topology,
prove that this discrete 2-dimensional holonomy is gauge invariant and
independent of the way we combine the faces of a particular
CW-decomposition of the 2-sphere and of the  2-disk. 
These are %seemingly new
results, of  %its own independent
intrinsic interest.
They provide a combinatorial construction of the 2-dimensional holonomy of
a higher order bundle, completing its differential geometrical
construction  discussed for example in  
 \cite{baez_schreiber,schreiber_waldorf2,martins_picken,BaezHuerta11}.}

{In sections \ref{ss:hlgt}, \ref{2dholonomy}  and \ref{GTRANS} we
lift the construction of ordinary {lattice} gauge 
theory to a higher setting, as outlined in  \cite{BCKMM,Pfeiffer}.
Let us summarise the general procedure. }

{A gauge configuration of ordinary lattice gauge theory
with gauge group $G$ is
{given by a}
 map
from the set of (by definition oriented) edges of the lattice into
%the underlying gauge group
$G$.
The well-definedness of lattice gauge
theory can be expressed by 
saying that there is a lattice groupoid supporting well-defined groupoid
maps (here called {\em discrete parallel transport
functors} \S\ref{dptf}) to a gauge group anytime a gauge configuration
is given.   
The `lattice groupoid' is a groupoid version of the free category over
a graph (see for example  \cite{MacLane,Higgins}) for a suitable graph
derived from the lattice.
%{!This para is broken!}
It is the freeness that makes discrete parallel transport functors well defined.  
A `suitable graph' is (it is claimed) the 1-skeleton of a
suitable CW-complex decomposition of physical space. If we aim for {\em
	topological}
field theory then in principle any  sufficiently regular CW-complex will
do.
Normally there is a notion of local structure --- chunks of space
that are independent of each other, which collectively encode 
extended structure.
In this sense, the `big story' of lattice gauge theory is that 
the free groupoid over a suitable lattice is an adequate model of
physical space.} 

Our first task here is to construct a well-defined lift of these notions to the higher setting. The main tool is the concept of a lattice 2-groupoid (to be a model of space in lattice higher gauge theory), which in  this paper is constructed  in an algebraic topological language as the fundamental crossed module $\Pi_2(M^2,M^1,M^0)$ (see \cite{brown_hha} and \cite[Chapter 6]{brown_higgins_sivera}) of a certain filtered space associated to a 2-lattice decomposition $L$ of the manifold $M$; see \S \ref{revWhite}. We will make a very strong use of a  freeness result for the lattice 2-groupoid, which essentially is a classical freeness theorem of Whitehead \cite{Whitehead,brown_whitehead}, transported to the groupoid setting by Brown and Higgins; see  \cite[6.8]{brown_higgins_sivera} and {\cite{brown_2dvk,brown_higgins_colimits,brown_higgins_cubes}}. Whitehead's theorem  provides also an equivalent combinatorial definition of the lattice 2-groupoid {of $(M,L)$, {i.e. of a pair} consisting of a manifold with a 2-lattice decomposition.}  

Let $M$ be a manifold with a 2-lattice decomposition $L$. Given a crossed module $\Gc=(\d\colon E \to G,\trr)$, representing the underlying gauge symmetry 2-group, a 2-gauge configuration is defined as a map assigning an element of the group $E$ to each (pointed and oriented) face of  $L$ and an element of $G$ to each (oriented) edge of $L$. Physically relevant configurations  furthermore satisfy a certain compatibility condition
--- called fake-flatness. This is a discretised version of the well-established fake-flatness condition for differential geometrical 2-connections; see \cite{BaezHuerta11,baez_schreiber,schreiber_waldorf2,martins_picken}. The term fake-flatness was first used in \cite{BrM}.

In analogy to lattice gauge  theory, we prove that any fake-flat 2-gauge configuration $\F$ extends uniquely to a crossed module map  (called a {\em discrete parallel transport 2-functor}) from the lattice 2-groupoid of $(M,L)$ into the underlying gauge 2-group $\Gc$; see {\S\ref{revWhite}.} These discrete parallel transport 2-functors are a discrete version of the differential-geometrical parallel transport 2-functors of \cite{schreiber_waldorf2,martins_picken}. Given an oriented {2-disk} or 2-sphere $\Sigma$ embedded in $M$, as a subcomplex, and a vertex $v$ of $\Sigma$, we can then combine the 1-dimensional and 2-dimensional holonomies of the constituting pieces of $\Sigma$, and obtain an $E$-valued 2-dimensional holonomy ${\rm Hol}_v^2(\F,\Sigma,L)$ of the {fake-flat} 2-gauge configuration $\F$ along $\Sigma$. These are the 2-dimensional holonomy operators previously referred to. By using some basic algebraic topology, {and the fact that the oriented mapping class groups of the 2-sphere and of the 2-disk both are trivial,} we can then provide algebraic-topological and combinatorial descriptions of  ${\rm Hol}_v^2(\F,\Sigma,L)$, and also show that the discrete 2-dimensional holonomy  ${\rm Hol}_v^2(\F,\Sigma,L)$ of $\F$ along $\Sigma$ depends only on the base-point (in a way controlled by the action of $G$ on $E$) and on the surface orientation, and not on any other  data such as the order of multiplication of constituent 2-cells. This latter result does not apply (in this form) to other surfaces since the mapping class group is then more complicated: {in general an isotopy class of embeddings is needed to define the 2-dimensional holonomy} of a 2-gauge connection along an embedded surface. For discussion see \cite{martins_picken,Zu1}.

{Playing a prominent role in the construction of our model, we
introduce  gauge transformations between {fake-flat} 2-gauge
configurations. Gauge transformations  initially come in two different
types: vertex and edge types. These correspond to the thin and fat
gauge transformation of \cite{BFCG2}. Vertex and edge gauge
transformations obey a semi-direct product structure, and can be
assembled into a group of gauge operators, which acts on the set of
fake-flat 2-gauge configurations. This action is explicitly
constructed using a double category derived from the crossed
module \cite{brown_higgins_cubes,martins_picken}, {and ultimately
originates from a groupoid of fake-flat 2-gauge configurations and
`full gauge transformations between them', which we will carefully
construct {in} \S \ref{fgtb}.}
 We prove in \S \ref{ph1} that gauge transformations preserve
the 2-dimensional holonomy of {fake-flat} 2-gauge configurations along
cellularly embedded 2-spheres in $M$.}   

{{As  mentioned in the previous paragraph}, a major underpinning construction is that of a groupoid of fake-flat 2-gauge configurations and full gauge
transformation between  them  \S \ref{fgtb}.}
{The latter groupoid can be seen as a combinatorial description of a certain groupoid of crossed complex (a generalisation of crossed modules) maps and their homotopies, which  appeared in the work of Brown and Higgins on tensor products and homotopies of crossed complexes; see \cite{Brown_tensor,brown_hha}. This point of view will be essential when we discuss the ground state degeneracy of the higher Kitaev model in \S \ref{gsd}.}

{Let $\Gc=(\d\colon E \to G,\trr)$ be a crossed module, representing the underlying gauge symmetry 2-group. Let $L$ be a 2-lattice decomposition of $M$. A fake-flat 2-gauge configuration $\F$ in $(M,L)$ is said to be 2-flat along a cellularly embedded 2-sphere $\Sigma$ if the 2-dimensional holonomy ${\rm Hol}_v^2(\F,\Sigma,L)$ of $\F$ along $\Sigma$ is the identity element of $E$. This 2-flatness of $\F$ along a 2-sphere $\Sigma\subset M$ is preserved by gauge transformations. A fake-flat 2-gauge configuration $\F$ in $(M,L)$ is said to be {\em 2-flat} if it is 2-flat along the boundaries of all 3-cells of $L$.}

{A crucial fact that we will use in this paper is the following one, a consequence of the work of Brown and Higgins \cite{brown_higgins_colimits,brown_higgins_cubes,Brown_tensor,brown_classifying}; a more modern reference is \cite{brown_higgins_sivera}. A  2-flat 2-gauge configuration $\F$ naturally yields a map $f_\F\colon M \to B_\Gc$, defined up to homotopy, from $M$ into the classifying space $B_\Gc$ of the crossed module $\Gc$; classifying spaces of crossed modules are defined in \cite{brown_higgins_sivera,brown_classifying,brown_hha} and also \cite{martins_porter,martins_crossed}. Moreover, by \cite[THEOREM A]{brown_classifying} and \cite[\S 11.4]{brown_higgins_sivera}),  it follows that given two 2-flat 2-gauge configurations $\F$ and $\F'$, then $f_\F,f_{\F'}\colon M \to B_\Gc$ are homotopic if, and only if, the 2-flat 2-gauge configurations $\F$ and $\F'$ are connected by a {full} gauge transformation. These facts will play a primary role in the proof that the ground state degeneracy of our model is a topological invariant of manifolds $M$, counting the number of homotopy classes of maps from $M$ into  $B_\Gc$; see \S \ref{gsd}.}

\subsubsection*{Overview of the paper}
In {Section \ref{ss:hlgt},} we recap 
and fix conventions for: crossed modules, fundamental crossed
modules, CW-complexes and 2-lattices, defined
in \S \ref{sec:2-lattices}.
In section  \ref{2dholonomy}, we firstly define and discuss fake-flat
2-gauge configurations  (called ``cellular formal $\cal C$-maps''
in  \cite{PorterTuraev});
see  \S \ref{ss:HOGC}.
In \S \ref{revWhite}, we define the lattice
2-groupoid for a pair $(M,L)$, consisting of a manifold $M$ with a
2-lattice decomposition $L$, and show how fake-flat 2-gauge
configurations give rise to 2-dimensional parallel transport
2-functors, from the lattice 2-groupoid of $(M,L)$ into the gauge
crossed module $\Gc$. In \S \ref{atd} we give an algebraic topological
definition of the 2-dimensional holonomy of a fake-flat 2-gauge
configuration along a 2-sphere and along a 2-disk. In \S \ref{cd}, we
give a combinatorial definition of 2-dimensional holonomy along
2-disks and 2-spheres, and prove that the two definitions of
2-dimensional holonomy coincide.

In {Section} \ref{GTRANS} we discuss gauge transformations
between fake-flat 2-gauge configurations defined on a 2-lattice. In
particular we define a group of gauge operators and prove that it acts
on the set of fake-flat 2-gauge configurations in a way such that the
2-dimensional holonomy along cellularly embedded 2-spheres is
preserved. {The underpinning groupoid of fake-flat 2-gauge configurations and full gauge transformations between them is constructed in \ref{fgtb}.}

{In Section \ref{Hamiltonian}, we address a Hamiltonian model for higher lattice gauge theory on a pair $(M,L)$, consisting of a manifold $M$ with 2-lattice decomposition $L$. This will be our proposal for a higher gauge theory version of Kitaev quantum-double model for topological phases: the {\it higher Kitaev model}. 
The underlying Hilbert space of our model is the free vector space on the set of all fake-flat  2-gauge configurations, and hence coincides with the Hilbert space in  \cite{BCKMM} for triangulated manifolds. In \S \ref{HamCalc} we explicitly construct the higher Kitaev model, and give detailed description of all operators involved. {In \ref{loca} we define the local operator algebra.}
The Hamiltonian \ref{thkm} for the higher Kitaev model is a sum of three mutually commuting terms. {We have two sums over 1-cells and 2-cells, respectively, constructed by using the action of the group of gauge operators, which impose higher gauge invariance along gauge transformations of vertex and edge types; and one sum over 3-cells, imposing 2-flatness along their boundary 2-sphere.} {A comparison with the Kitaev model is done in \S\ref{comparison}.} 

In \S \ref{gsd} we show that the dimension 
of the ground state { of the higher Kitaev model} is given by the number of homotopy
classes of maps from the space 
manifold $M$ to the classifying space of the gauge crossed module $\Gc$ and therefore {ground state degeneracy} is a topological (in fact homotopy) invariant {of $M$}. (At this point we needed again to appeal to some basic algebraic topology for crossed modules { and crossed complexes} as given in  \cite{brown_higgins_sivera,martins_crossed,martins_porter}.)
This ground state degeneracy can  be proven to coincide with Yetter's invariant 
on $M\times S^1$ (the level $D$ invariant of the TQFT); see  \cite{BCKMM}.}

\tableofcontents

\section{Preliminaries on crossed modules, CW-complexes and 2-lattices} \label{ss:hlgt}
In \S \ref{ss:HOGC} we give the definition of a fake-flat 2-gauge configuration on a 2-lattice decomposition of a manifold.
This makes extensive use of crossed modules; we assemble the key
definitions in \S\ref{crossed_modules_def}. (Crossed modules {of groups} are well-known to be equivalent to strict 2-groups \cite[\S 2.5]{brown_higgins_sivera} and \cite{baez_lauda}, thus from now on only crossed modules will be mentioned.)
Then in \S\ref{ss:HOL} we recall some facts about CW-complexes which we will need in \S \ref{sec:2-lattices} for defining 2-lattice decompositions of manifolds.

\begin{remark}\label{bdnote}{In this paper
    we will use $\bound(M)$ to denote the boundary of a manifold $M$.
    (We avoid the common notation $\d M$, in order not to overuse
    the symbol $\d$, which appears in several other contexts.)}\end{remark}

\subsection{Crossed modules (of groups and of groupoids)}\label{crossed_modules_def}

Crossed modules of groups are discussed in
\cite{baez_lauda,brown_hha,martins_crossed}.
%The more general case of c
Crossed modules of groupoids, discussed extensively in this paper,
appear in \cite[\S 6.2]{brown_higgins_sivera} and \cite{brown_hha,martins_porter,BrownIcen}.
{Crossed modules of groups and  groupoids can be used for formalising
  2-dimensional (2D) notions of holonomy (surface holonomy),
  in the same way that groups appear in the formulation of the
  holonomy of a connection in a principal bundle.}

\begin{definition}[Crossed modules of groups; Peiffer relations]\label{cm}
Let $E$ and $G$ be groups.
A {\em crossed module} ${\cal G}=(\d\colon E \to G,\trr)$ of groups
is given by a group map $\d\colon E \to G$,
together with a left action $\trr$ of $G$ on $E$ by automorphisms, such that the relations below, called {\em Peiffer relations}, hold for each $g \in G$ and $e,e' \in E$:
\begin{align} \label{pf1}
 \textrm{1st Peiffer relation} \qquad \d(g \trr e)&=g \d(e) g^{-1},
\\ \label{pf2}
 \textrm{2nd Peiffer relation} \qquad \d(e) \trr e'&=ee'e^{-1}.
\end{align}
\end{definition}
\begin{example} \label{eg:G32}
The crossed module
$\Gc = \Gc_{32}  := (\d:\Z_3^+ \rightarrow \Z_2^\times,\trr)$,
where $\Z_3^+ = \{0,1,2 \}$ is
the additive group of integers modulo $3$ and
$\Z_2^\times = \{\pm 1\}$ acts on
$\Z_3$ as $z \trr e=ze$.
The boundary map $\d$ sends everything to $+1$.
\end{example}
\begin{example}[From groups to crossed modules I]
  Given a group $G$, let ${\rm Aut}(G)$ be the group of automorphisms of $G$. Clearly ${\rm Aut}(G)$ acts in $G$ by automorphisms as $f\trr g=f(g)$, for each $f \in {\rm Aut}(G)$ and each $g \in G$. Let ${\rm Ad}\colon G \to {\rm Aut}(G)$ be the morphism that sends $g\in G$ to the inner automorphism ${\rm Ad}_g\colon x \in G \mapsto gxg^{-1} \in G$, obtained by conjugating by $g$. Then ${\cal AUT}(G)=({\rm Ad}\colon G \to {\rm Aut}(G), \trr)$ is a crossed module.
 \end{example}

\begin{example}[From groups to crossed modules II]\label{gtxm}
  If $G$ is a group, then $( \{1\} \to G)$ and $(\id \colon G \to G,
  {\rm Ad})$, where ${\rm Ad}$ is the adjoint action, are crossed
  modules.
  If $G$ is abelian then $(G\to \{1\})$ is also a crossed module.
\end{example}

Let us now {discuss crossed modules of groupoids.}
{Let
$G=(\sigma,\tau\colon G_1 \to G_0)$
denote a groupoid \cite{Higgins,brown_hha} with
set of objects $G_0$; set of morphisms $G_1$;  and source
and target maps  $\sigma,\tau\colon G_1 \to G_0$.
We represent the morphisms $\gamma \in G_1$ as
$a \ra{\gamma} b$. Thus $\sigma(\gamma)=a$ and $b=\tau(\gamma)$.
Given $a,b \in G_0$, the set of morphisms $a \to b$ is
$\hom(a,b) = 
\{\g\in G_1\colon \sigma(\gamma)=a \textrm{ and } \tau(\gamma)=b\}$.
{The composition map in $G$ yields for each triple $(a,b,c)$ of objects} a map
$\circ\colon \hom(a,b) \times \hom(b,c) \to \hom(a,c)$, which we
 represent as (notice composition order):}
$$
\big ( a \ra{\gamma} b  \big)\,\circ\,\big ( b \ra{\gamma'} c \big)
 = \big ( a \ra{\gamma\gamma'} c  \big).
$$

{A {\em totally intransitive groupoid} $E$ is a groupoid of the form
 $E=(\beta,\beta \colon E_1 \to E_0)$. (Thus source and target maps coincide.)
Given $x \in E_0$, we let
${\rm Aut}(x)=\{e \in E_1: \beta(e)=x\}$, which is a group.
And then $E$  is isomorphic
to the totally intransitive groupoid given by 
$\sqcup_{x \in  E_0}\aut(x)$, with the obvious composition and map
$\beta\colon \sqcup_{x \in E_0}\aut(x) \to E_0$. Hence a totally intransitive groupoid can been seen as
being given by a disjoint union of groups.}

{A {\em  left groupoid action} $\; \trr $ \cite{brown_hha,BrownIcen}, by automorphisms, of the groupoid
$G=(\sigma,\tau\colon G_1 \to C)$ on $E=(\beta,\beta\colon E_1 \to C)$,
a totally intransitive groupoid with the same set of objects as $G$,
is given by a set map:}
$$
{(\gamma,e) \in \{(\gamma',e') \in G_1 \times E_1
    : \tau(\gamma')=\beta(e')\} \longmapsto \gamma \trr e \in E_1,}
$$
such that  whenever compositions and actions are well-defined:
\begin{align*}
\beta(\gamma \trr e)&=\sigma(\gamma),
&&(\gamma \gamma') \trr e=\gamma \trr (\gamma' \trr e) \qquad\qquad \textrm{ and } \quad&&
 \gamma\trr(e e')=(\gamma \trr e ) ( \gamma \trr e').
\end{align*}

\begin{definition}[Crossed module of groupoids]\label{cmg}
Let $E$ and $G$ be groupoids with the same object set, with $E$
totally intransitive.
A {\em  crossed module of groupoids}
$
{\Gc}=(\d\colon E\to G,\trr)
$
is given
by a groupoid map $\d\colon E \to G$, which is the identity on
objects, together with a left action of $G$ on $E$, {by automorphisms}, such that the
Peiffer relations (\ref{pf1},\ref{pf2}) are satisfied, whenever actions and
compositions make sense.  {(Full equations are in \cite{BrownIcen}.)}
\end{definition}

{Given $\Gc=(\d\colon E\to G,\trr)$, we call $E$ the {\em top groupoid} of $\Gc$ and  $G$  the {\em underlying groupoid} of $\Gc$.}

\begin{definition}[Crossed module map]
A map
$(\psi,\phi)\colon (\d\colon E\to G,\trr)\to (\d\colon E'\to G',\trr) $
of crossed modules of groupoids is given by two
 groupoid maps $\psi\colon E \to E'$ and
$\phi\colon G \to G'$, which are compatible with actions
and boundary maps in the obvious way. {(Full equations are in \cite[\S 1.1.1]{martins_porter}.)}
\end{definition}
{Since groups can be considered to be groupoids with a single object, we will see group crossed modules as particular cases of crossed modules of groupoids.}

\subsection{Example: % of crossed module of groupoids:
  {the} fundamental crossed module}

The main example of a crossed module of groupoids is a  topological
one crucial to our construction. Our main references are
\cite[\S 2.1, \S 2.2 and \S 6]{brown_higgins_sivera} and \cite{brown_hha}.
We will need to review some algebraic topology definitions.

\begin{definition} \label{de:fg}
  (See e.g. \cite[p.17]{CrowellFox} and \cite{brown_higgins_sivera,brown_hha}.)
Let $Y$ be a locally path-connected space,  and
$C \subset Y$ any subset (in this paper $C$ will always be finite).
The {\em fundamental groupoid of $Y$, with object set $C$},
denoted  $\pi_1(Y,C)$,
is as follows.
The set of objects of
$\pi_1(Y,C)$ is $C$. Given $c,d \in C$, the set of morphism $\hom(c,d)$ is the set of
 equivalence classes of paths $\gamma \colon [0,1] \to Y$, such
that $\gamma(0)=c$ and $\gamma(1)=d$, where two paths $c \to d$ are
equivalent if they are homotopic in $Y$,
relative to the end-points (i.e. end-points remain stable during the homotopy).
The composition in $\pi_1(Y,C)$ is given by concatenation
(and rescaling) of representative paths.
\end{definition}

If $\gamma$ is a path in $Y$, the equivalence class to which it belongs
in $\pi_1(Y,C)$ is denoted by $[\gamma]$.
A morphism in $\pi_1(Y,C)$ from $c$ to $d$ is denoted as $c \ra{[\gamma]} d$ or simply
by $c \ra{\gamma} d$ if no ambiguity arises.
%\end{definition}

\begin{remark}\label{pi1}
Let $c \in C$. The group of morphisms $c \to c$ in the groupoid
$\pi_1(Y,C)$ is exactly the fundamental group $\pi_1(Y,c)$.
Let $S^1=\bound([0,1]^2)$, with a base point $*$ at  $(0,0)$;
{recall Rem. \ref{bdnote}.}
Morphisms $c \to c$ hence can equivalently be seen as pointed homotopy
classes of maps $(S^1,*) \to (Y,c)$.
\end{remark}

Relative homotopy groups, including $\pi_2(X,Y,c)$, of pointed pairs
of spaces ($c$ being the base-point) are classical in homotopy theory
and are defined e.g. in \cite[p.343]{hatcher}. In this paper, we will
use relative homotopy groupoids $\pi_2(X,Y,C)$, with a set $C\subset
Y$ of base-points; see \cite[\S 1.6, \S  6.2 and \S
  6.3]{brown_higgins_sivera}. These are totally intransitive groupoids
built as $\pi_2(X,Y,C)=\sqcup_{c \in C} \pi_2(X,Y,c)$. Let us give a
quick review.

\begin{definition}[The totally intransitive groupoid $\pi_2(X,Y,C)$]\label{defpi2}
Let $X$ be a locally path-connected space. Let $Y \subset X$ be a
locally path-connected subspace of $X$. Choose a subset $C$ of $Y$. In this paper, $C$ will always intersect non-trivially each path-component of $X$ and of $Y$.
For each $c \in C$, consider the {\em relative homotopy group}
$\pi_2(X,Y,c)$.
This group is made out of {homotopy} classes of maps
$\Gamma\colon [0,1]^2 \to X$ such that:
\begin{enumerate}
\item $\Gamma\big( ([0,1] \times \{0\}) \cup (\{0,1\} \times [0,1]) \big)={\{c\}}$,
\item $\Gamma\left ( [0,1] \times \{1\}  \right)\subset Y$.
\end{enumerate}
Specifically, two such maps $\Gamma,\Gamma'\colon [0,1]^2 \to X$ are said to be
homotopic if there exists a homotopy $J\colon [0,1]^3 \to X$,
connecting $\Gamma$ and $\Gamma'$, such that for all $u\in[0,1]$ the
slice of $J$ at $u$, namely $(t,s) \mapsto J_u(t,s)=J(t,s,u)$,
satisfies the properties 1 and 2. The multiplication in $\pi_2(X,Y,c)$ is through horizontal
juxtaposition of maps $[0,1]^2 \to X$, followed by rescaling in the horizontal direction.

We can thus define a totally intransitive groupoid
$\displaystyle\pi_2(X,Y,C)\doteq \bigsqcup_{c \in C} \pi_2(X,Y,c),$ with set of
objects $C$.
\end{definition}

Let $(X,Y,C)$ be as in Def. \ref{defpi2}.
The elements $[\Gamma] \in
\pi_2(X,Y,c)$, or simply $\Gamma\in \pi_2(X,Y,c)$, if no confusion
arises, are visualised as:
$$\xymatrixrowsep{0.3in}
\xymatrixcolsep{0.5in} \xymatrix{ &c \ar[r]|{\textstyle{\partial(\Gamma)}} &c\\
                                                                  &c\ar@{-}[u]|{\textstyle c} \ar@{{}{ }{}}@/^1.8pc/[r]_{\textstyle \Gamma}
                                                                  \ar@{{}{ }{}}@/^1.5pc/[u]|{\textstyle \Gamma =}
                                                                  \ar@{-}[r]|{\textstyle c} &c\ar@{-}[u]|{\textstyle c} }.$$
Let $c \in C$. As indicated by the diagram above, if we restrict a $\Gamma \in
\pi_2(X,Y,c)$ to the top of the square $[0,1]^2$, this  gives rise to an
element $\partial(\Gamma) \in \pi_1(Y,c)$. This yields a group map
$
\partial\colon \pi_2(X,Y,c) \to \pi_1(Y,c)  .
$ {Putting all of these group maps together, yields a groupoid map $\d\colon \pi_2(X,Y,C) \to \pi_1(Y,C)$, which is the identity on objects.}
{We also have an action of the groupoid $\pi_1(Y,C)$ on {the totally intransitive groupoid $\pi_2(X,Y,C)$,} as indicated in figure \ref{action}. Details are in \cite[\S 2.2 and \S 6.1]{brown_higgins_sivera} and (in the pointed case) \cite[pp 355]{hatcher}.}
\begin{figure}[H]
{\centerline{\relabelbox
\epsfysize 3.5cm
\epsfbox{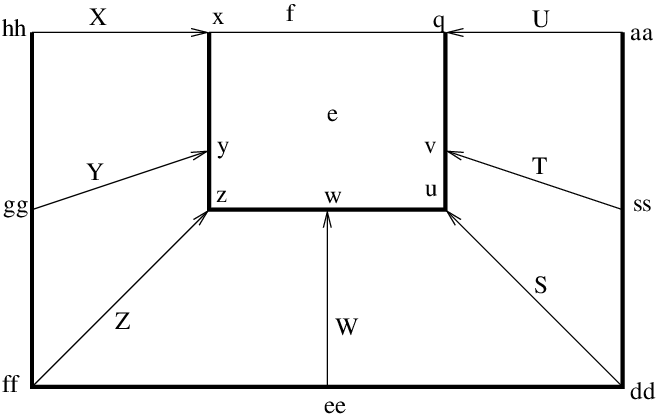}
\relabel{X}{$\scriptstyle{\gamma}$}
\relabel{Y}{$\scriptstyle{\gamma}$}
\relabel{Z}{$\scriptstyle{\gamma}$}
\relabel{W}{$\scriptstyle{\gamma}$}
\relabel{S}{$\scriptstyle{\gamma}$}
\relabel{T}{$\scriptstyle{\gamma}$}
\relabel{U}{$\scriptstyle{\gamma}$}
\relabel{e}{$\scriptstyle{\Gamma}$}
\relabel{f}{$\scriptstyle{\partial(\Gamma)}$}
\relabel{x}{$\scriptstyle{c}$}
\relabel{y}{$\scriptstyle{c}$}
\relabel{z}{$\scriptstyle{c}$}
\relabel{w}{$\scriptstyle{c}$}
\relabel{u}{$\scriptstyle{c}$}
\relabel{v}{$\scriptstyle{c}$}
\relabel{q}{$\scriptstyle{c}$}
\relabel{aa}{$\scriptstyle{d}$}
\relabel{ss}{$\scriptstyle{d}$}
\relabel{dd}{$\scriptstyle{d}$}
\relabel{ee}{$\scriptstyle{d}$}
\relabel{ff}{$\scriptstyle{d}$}
\relabel{gg}{$\scriptstyle{d}$}
\relabel{hh}{$\scriptstyle{d}$}
\endrelabelbox }}
\caption{{\label{action} The action of an element $\gamma \in \pi_1(Y,C)$, with $\gamma(0)=d$ and $\gamma(1)=c$ on a $\Gamma \in \pi_2(X,Y,c)$}.}
\end{figure}

\begin{Theorem}\label{defpi2-2} (JHC Whitehead, \cite[\S 2.3, \S 6]{brown_higgins_sivera} and \cite{brown_hha,brown_2dvk}) Let $(X,Y,C)$ be a triple of spaces, as in Def. \ref{defpi2}.
Considering the natural action $\trr$ of the groupoid $\pi_1(Y,C)$ on the
totally intransitive groupoid $\pi_2(X,Y,C)$, and the  boundary map  $\d\colon \pi_2(X,Y,C) \to \pi_1(Y,C)$,
we have a crossed module of groupoids, called the fundamental crossed module of $(X,Y,C)$. The fundamental crossed module of $(X,Y,C)$ is denoted as:
$$
\Pi_2(X,Y,C) \; = \; (\d\colon \pi_2(X,Y,C) \to \pi_1(Y,C), \trr).
$$
\end{Theorem}
\begin{remark}\label{altdefpi} Let $(X,Y)$ be a pair of spaces and $c \in Y$.
 Recall that the underlying set of the group $\pi_2(X,Y,c)$ can also be defined as the set of all maps $f\colon {[0,1]^2} \to X$, such that $f(*)=c$, where $*=(0,0)$, and $f(\bound[0,1]^2)\subset Y$, up to a  homotopy $H\colon (x,t)\in {[0,1]^2} \times [0,1] \mapsto f_t(x)\in X$, such that, for each $t$, $f_t(*)=c$ and $f_t(\bound[0,1]^2)\subset Y$. The boundary map $\d\colon \pi_2(X,Y,x) \to \pi_1(Y,c)$ is obtained by restricting $f$ to $\bound([0,1]^2)$;  c.f. Rem. \ref{pi1}.

 Analogously \cite[Chapter IV]{hatcher}, the underlying set of the relative homotopy group $\pi_3(X,Y,c)$ can be defined as the set of all maps $f\colon {[0,1]^3}\to X$ such that $f(*)=c$, where $*=(0,0,0)$, and $f(\bound[0,1]^3)\subset Y$, up to a homotopy $H\colon (x,t)\in {[0,1]^3} \times [0,1] \mapsto f_t(x)\in X$ such that, for each $t$, $f_t(*)=c$  and $f_t(\bound[0,1]^3)\subset Y$. We also have a boundary map  $\d\colon \pi_3(X,Y,x) \to \pi_2(Y,c)$  obtained by restricting $f$ to $\bound([0,1]^3)$. %This is the map appearing in the long homotopy exact sequence of $(X,Y,c)$.
\end{remark}

\begin{example}[The fundamental crossed module of the disk]\label{pi2d2}Let $D^2=[0,1]^2$ and $S^1=\bound(D^2)$. {Let $v \in S^1$  be any point}. Then {$\Pi_2(D^2,S^1,v)\cong(\id\colon\Z \to\Z,\trr)$}, where $a \trr b=b$, for each $a,b \in \Z$. To see this, look at the end of the homotopy {long} exact sequence of $(D^2,S^1,v)$; see e.g. \cite[Chapter IV]{hatcher}. This yields {$\{0\} \cong \pi_2(D^2,v) \to \pi_2(D^2,S^1,v) \stackrel{\d}{\to} \pi_1(S^1,v)\cong\Z \to \pi_1(D^2,v)\cong \{1\}.$} Details are in e.g. \cite{martins_crossed}.
\end{example}

\newcommand{\inter}{{\rm int}}
\newcommand{\ignore}[1]{}

\newcommand{\fg}[1]{{\mathsf {FG}}\langle #1 \rangle}  %% free groupoid
\newcommand{\GG}{{\mathcal G}}    %% a generic groupoid
\newcommand{\base}{\beta_{0}}     %% base point function

\subsection{CW-complexes} \label{ss:HOL}

Let $D^n$ denote the closed $n$-disk in the form $D^n = [0,1]^n$.
The open $n$-disk is $\inter(D^n)=(0,1)^n$. Also put:
\[
\bound( D^n) = S^{n-1}  =  D^n \setminus \inter({D^n})
\]
--- the boundary of the $n$-disk.
Let $\N=\{0,1,2,\dots\}$.

Let us briefly review the definition of CW-complexes
\cite[Appendix]{hatcher},
\cite{cellular_structures} and \cite{topology_CW_complexes}.
We will use the definition given in \cite[Prop A2]{hatcher}
and \cite[Chapter II]{topology_CW_complexes}.

\begin{definition}[CW-complex]\label{cwcpl}
A CW-complex $\left( X, \{ \phi_a^n \}_{a \in L^n,n \in \N} \right)$
is a Hausdorff topological space $X$, a collection of sets $L^0$,
$L^1$, $L^2$, ..., and,
for each $n \in \N$, a family of
continuous maps $\{\phi_a^n \colon D^n \to X\}_{a \in L^n}$
{(the  `characteristic maps of the closed $n$-cells')}
satisfying conditions {1,2,3 and 4, below}.

\smallskip

\noindent Let  the set $c^n_a\doteq \phi_a^n(\Dna)\subset X$.
It is called an \emph{open cell of dimension $n$}, and is given the induced topology.
{Put $\overline{c^n_a}\doteq \phi_a^n(D^n)\subset X$.
It is called a \emph{closed cell of dimension $n$},
and is given the induced topology.
Put $\bound(\overline{c^n_a})\doteq \phi_a^n(\bound(D^n))\subset X$.
It is called the boundary of  $c^n_a$. (Note that
$\overline{c^n_a}$ need not be a $\d$-manifold, hence
$\bound(\overline{c^n_a})$ might not be a %genuine
manifold boundary,
though this will be imposed when we define 2-lattices.)} 
%
%\smallskip
%
%\noindent
Then:
\begin{enumerate}
\item 
Each characteristic map  $\phi_a^n\colon D^n \to X$ restricts to
a homeomorphism $\Dna\to \phi_a^n(\Dna)\subset X$.

\item The open cells $c^n_a$  %\doteq\phi^n_a(\Dna)\subset X$,
  where $n \in \N$ and $a \in L^n$, form a partition of $X$.
 {(I.e. they are pairwise disjoint and their union is $X$.)}

\item Each ${\bound(\overline{c^n_a})}$ is 
contained in the union
  of a finite number of open cells of dimension $<n$.  
\item A set $F\subset X$ is closed if, and only if, 
  $(\phi_a^n)^{-1}(F)$ is closed in $D^n$, for each $n \in \N$ and each $a \in L^n$.
\end{enumerate}
\end{definition}

\noindent
A CW-complex is called finite if $L^n$ is finite for each
$n \in \N$ and $L^n=\emptyset$ for all but a finite subset of $n\in \N$.

\noindent
We write $X$ for  $(X, \{ \phi_a^n \}_{a \in L^n,n \in \N} )$.
The data $\{ \phi_a^n\colon D^n \to X \}_{a \in L^n,n \in \N}$
is called a CW-decomposition of $X$.

\begin{definition}   
%\item
A {\em subcomplex} of a CW-complex
{$\left( X, \{ \phi_a^n \}_{a \in L^n,n \in \N} \right)$} is a subspace $A \subset X$
which is the union of open cells of $X$, such that the closure in $X$ of
each of these open cells is contained in $A$. 
\end{definition}

A subcomplex  %of a CW-complex is itself  a CW-complex.
%More precisely,
$A$ can be made into a CW-complex
$\left( A, \{ \phi_b^n \}_{b \in L^n_A,n \in  \N} \right)$,
where for each $n\in \N$,  we put
$L_A^n=\{c\in  L^n\colon \overline{c^n_a} \subset A\}$.  
(For a  proof see e.g. \cite[pg 16]{hatcher}.)

\begin{definition} \label{de:n-skel}
% \item 
The $n$-{\em skeleton} $X^n$ of a CW-complex $X$ is the
subspace  %set  $X^n$ 
given by the union of all the open cells of dimensions $\leq n$, with
the induced topology. Note that $X^n$ is a subcomplex of $X$, hence a
CW-complex.    
\end{definition}

\begin{remark}[CW-complexes: properties and nomenclature]
\label{someprop}
%Let us summarise some well-known properties of CW-complexes and recall
%some nomenclature.
For proofs see e.g. \cite[Appendix]{hatcher} and
\cite{topology_CW_complexes,cellular_structures}.

\begin{itemize}
\item Condition 4. of the definition of a CW-complex is  redundant if $X$ has only a finite number of cells; see \cite[pp 521]{hatcher}. (Essentially this follows since a finite union of closed sets is always closed). In this paper we will only deal with finite CW-complexes, so condition 4. of Def. \ref{cwcpl} will not be mentioned again. 

\item {Cf. \cite[pg 6]{topology_CW_complexes}, as the notation suggests, the closed cell $\overline{c^n_a}\subset X$ is the closure in $X$ of the open cell $c^n_a$.}

\item   The {\em  attaching map} of each closed $n$-cell $\overline{c^n_a}$ is the restriction of $\phi^n_a\colon D^n \to X$ to $\bound(D^n)$, namely:
$$
\psi^n_a \colon \bound(D^n) \to {\bound(\overline{c^n_a})} \subset X^{n-1} \subset X  .
$$ 
The underlying topological space of the $n$-skeleton $X^n$ of $X$ is homeomorphic to the space obtained from $X^{n-1}$ by attaching $\sqcup_{a \in L^n} D^n$ to it, along the attaching maps of the closed $n$-cells.
\end{itemize}
\end{remark}

\begin{definition} \label{de:cellular}
  Given CW-complexes $X$ and $Y$, a  map $f\colon X \to Y$ is called
  cellular if $f(X^n)\subset Y^n$, for all $n \in \N$.
\end{definition}

\begin{definition}[Abstract cells]
  If $(X, \{ \phi_a^n \}_{a \in L^n,n \in \N} )$ is a CW-complex,
  we call $L^n$ the set of abstract $n$-cells. 
\end{definition}

\noindent
Abstract $n$-cells are in one-to-one correspondence with open $n$-cells and with closed $n$-cells. If $a$ is an abstract $n$-cell, the closed and open $n$-cells it corresponds to are (respectively)  $\overline{c^n_a}=\phi^n_a(D^n)$ and ${c^n_a}=\phi^n_a(\inter(D^n))$.

\begin{remark}[(Geometric) vertices, edges, {plaquettes (or faces)}, and blobs]\label{blob} 
Abstract  0, 1, 2 and 3-cells  of a CW-complex will sometimes be called called vertices, edges, {plaquettes (or faces)}, and blobs, respectively. Closed  0, 1, 2 and 3-cells will sometimes be called geometric vertices, geometric edges, geometric {plaquettes (or faces)}, and geometric blobs. 
\end{remark}

\subsection{2-lattices}\label{sec:2-lattices}

Simplicial complexes give rise to CW-complexes; but simplicial
complexes are very rigid, therefore a large number of simplices is
usually required to triangulate a manifold.
CW-complexes allow for the
decomposition of a manifold into fewer cells; however they are too
general for our purposes, since the attaching maps of the closed cells
can be highly singular, making it harder to use CW-complexes in
combinatorial frameworks.  
In order to simplify our discussion later, we will  consider
CW-complexes which are 2-lattices, defined below.  

If $S^n=\bound(D^{n+1})$ is the $n$-sphere, the base-point $*$ of it
is defined to be $*=(0,\dots,0)$.

\begin{definition}[2-lattices. Base point of a cell]\label{2-lattice}
Let $M$ be a topological manifold, with 
CW-complex $ \Delta_M = (M, \{ \phi_a^n \}_{a \in L^n,n \in \N} )$.
This  $ \Delta_M  $ is called a
2-lattice for $M $
if, for each $n \in \N$ and each 
$a \in L^n$:
\\
(1)
a CW-decomposition
$Z_{a}$ of $S^{n-1}=\bound(D^n)$ is given for which the base-point
$*=(0,\dots,0)$ is a 0-cell, and such that the attaching map
$\psi_a^n\colon S^{n-1}\to {M}^{n-1}$ of the corresponding closed
$n$-cell $\overline{c_a^n}$ is cellular (as in Def.\ref{de:cellular}).
{
(Note that in particular
(1) implies that $\psi_a^n(*)=x_a$ is a closed 0-cell of $M$, for
each $a \in L^n$ and  each $n\in \N$. The image $\psi_a^n(*)=x_a$ is called
the base-point of the closed $n$-cell $\overline{c_a^n}$.) 
}
\\
(2)
one of the following two conditions holds:
 \begin{itemize}
  \item The attaching map $\psi_a^n\colon S^{n-1} \to {M}^{n-1}$ of the corresponding  $n$-cell $\overline{c_a^n}$ is constant.
 \item The attaching map $\psi_a^n\colon S^{n-1} \to {M}^{n-1}$ of the
   corresponding  $n$-cell $\overline{c_a^n}$ is an embedding (i.e. it
   is a homeomorphism onto its image).
   {Moreover, for each closed $i$-cell $c$ of ${Z_{a}}$,
   it holds that $\psi_a^n(c)$ is a closed $i$-cell $c_L$ of $M$,
   and the restriction  of $\psi_a^n\colon S^{n-1} \to {M}^{n-1}$
   to  $c$ is a homeomorphism $c \to c_L$.}  
 \end{itemize}
 (3) % For abstract 3-cells
 If $b\in L^3$, we impose that the attaching map
 $\psi_b^3\colon S^2 \to M^2$ of the closed 3-cell $\overline{c_b^3}$ is an embedding  and furthermore that the boundary $\psi_b^3(S^2)=\bound(\overline{c^3_b})$ {of the 3-cell} $\overline{c^3_b}$ is a subcomplex of $M^2$.
 
\end{definition}

The  space $M$
is then said to have a 2-lattice decomposition.

A 2-lattice $(M, \{ \phi_a^n \}_{a \in L^n,n \in \N} )$
will usually be denoted as $(M,L)$, or $\big(M,L=(L^0,L^1,\dots)\big)$.

\begin{remark} In practice, when defining a particular 2-lattice
  decomposition of $M$, normally only the closed $n$-cells will be
  made explicit, as it will always be clear that, for each $n$-cell
  $a$, an attaching map $\psi_a^n\colon S^{n-1}\to {M}^{n-1}$ can be
  found which is cellular by using a suitable CW-decomposition of the
  $(n-1)$-sphere. This does not fully determine a CW-decomposition, as
  some ambiguity rests on the actual characteristic maps of the
  $n$-cells. However the topological space $M$, all closed cells, all
  $i$-skeletons $M^i$, and hence the crossed modules
  $\Pi_2(M,M^1,M^0)$ and $\Pi_2(M^2,M^1,M^0)$ will be defined with no
  ambiguity. This is all we need for this paper.   
\end{remark}
\begin{remark}[{Lax 2-lattice}]
{A CW-complex satisfying only $(1)$ of the definition of 2-lattices is called a lax-2-lattice. All combinatorial constructions in this paper are still true for lax-2-lattices, with the obvious modifications. In particular, the combinatorial construction of the 2-dimensional holonomy operators in \S\ref{cd} remains almost unaltered. The only issue is that the description of edge and vertex gauge spikes in \S\ref{exp} then requires a lot more cases, especially when it comes to edge operators.}
\end{remark}

\begin{remark}\label{sublatices}
 Let $Y$ be a subcomplex of a CW-complex $X$.  If $X$ is a 2-lattice then clearly so is $Y$. 
\end{remark}

\begin{example}\label{S1}
Evidently, 1-dimensional CW-complexes are always  2-lattices.
  The circle $\{z \in \mathbb{C}\colon |z|=1\}$ can be given a 2-lattice decomposition with two vertices {(i.e. 0-cells)} at $z=\pm1 $ and closed 1-cells at   $\{z \in \mathbb{C}\colon |z|=1 \wedge \Im(z)\ge 0\}$ and  $\{z \in \mathbb{C}\colon |z|=1 \wedge \Im(z)\leq 0\}$. Here $\Im(z)$ denotes the imaginary part of $z$. 
 \end{example}
\begin{example}
An example of a CW-complex which cannot be made into a 2-lattice is given by attaching $D^2=[0,1]^2$ to $\{z \in \mathbb{C}\colon |z|=1\}$ along $(x,y)\in \bound(D^2) \mapsto \exp(x 2 \pi i\sin(2\pi/x))$, prolonged by continuity to $\bound(D^2)\cap \{x=0\}$. This is the type of singular attaching maps we want to avoid by restricting to 2-lattices.
\end{example}

\begin{example} Consider the 2-sphere $S^2=\bound (D^3)$, with the
CW-decomposition arising from the polyhedral structure of $D^3=[0,1]^3$.
Let $Y$ be the space obtained from $S^2$ by attaching
$D^3$ along $\psi\colon S^2 \to S^2$ defined as:  
$$
(x,y,z) \in \bound(D^3 )\stackrel{\psi}{\longmapsto} 
\begin{cases} (x,y,z) \textrm{ if } z \ge 1/2\\
(x,y,1-z) \textrm{ if } z \leq 1/2
\end{cases}
$$
This CW-decomposition of $Y$ is not a 2-lattice since the attaching map of its unique $3$-cell is not an embedding.
\end{example}

\begin{example}[Two 2-lattice decompositions of the 3-sphere  $S^3$]\label{LandL0}
  Let us {in this example} 
  model the 2- and 3-spheres as being
$S^2  =  \{x \in \mathbb{R}^3\colon |x|=1\}$ and
$S^3=\{x \in  \mathbb{R}^4\colon |x|=1\}$.
{The following are two 2-lattice decompositions
  of the 3-sphere.}
  %(which will play a major role when discussing
  %higher lattice gauge theory in the 3-sphere in \S  \ref{hgt3s}).}  

\smallskip
\noindent %{\bf $(S^3,L_{\mathfrak g})$}
{\bf $(S^3,L_{\mathfrak{g}})$:}
{We consider the 3-sphere $S^3$ with the globe decomposition {$L_{\mathfrak{g}}=(\{ v\}, \{ t \}, \{ P,P'\}, \{b,b' \})$} as follows. We firstly consider a CW-decomposition $L$ of $S^2$ with  a unique closed 0-cell $v$ at the point
of zero latitude and longitude, and a unique closed 1-cell $t$ making the
equator, oriented eastwards.
We have two closed 2-cells, $P,P'$, one for each
hemisphere, {attaching} along the equator (oriented eastwards). See Fig. \ref{L1}.}
\begin{figure}[H]
\centerline{\relabelbox 
\epsfysize 2.5cm 
\epsfbox{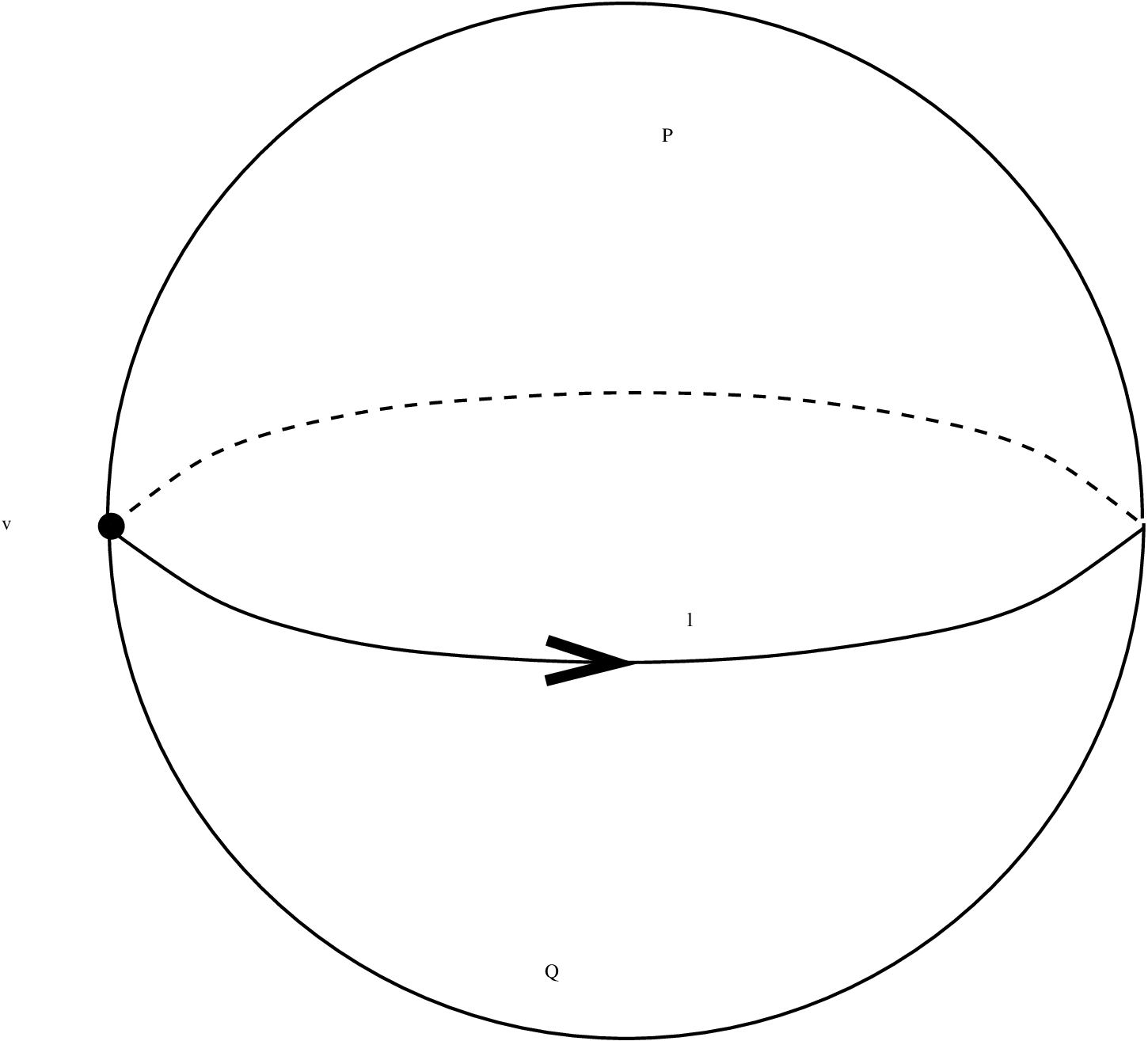}
\relabel{P}{$P$}
\relabel{Q}{$P'$}
\relabel{v}{$v$}
\relabel{l}{$t$}
\endrelabelbox}
\caption{A 2-lattice decomposition $L$ of the 2-sphere $S^2$.\label{L1}}
\end{figure}
\noindent To get from $S^2$ to $S^3$ we now need to add two additional 3-cells $b,b'$ attaching on each side of the 2-sphere. 
\medskip

 \noindent {\bf $(S^3,L_0)$:}   We can choose an even simpler 2-lattice decomposition $L_0$ of $S^3$,  having unique $0$- and $2$-cells (resulting in $S^2$), and two 3-cells $b$ and $b'$, as above, attaching along each side of the 2-sphere.
\end{example}

\begin{definition}[Notation: $\dL(P)$ and $\dL(b)$]\label{notdelta}
{Let $(M,L)$ be a 2-lattice. Let $P\in L^2$ be a geometric 2-cell
  (i.e. plaquette; Def. \ref{blob}).}
Let $\psi_P^2\colon \bound(D^2)\to M^1$
be the attaching map of the corresponding
closed 2-cell $\overline{c_P^2}$ (i.e. geometric plaquette).
By definition,  $*=(0,0)\in S^1$ and $\psi^2_P(*)$ is a closed 0-cell
$x_P$ of $M$. Hence $\psi_P^2$ is a pointed map
$(S^1,*)\to (M^1,x_P)$.
Passing to the pointed homotopy class of $\psi_P^2\colon (S^1,*) \to (M^1,x_P)$ yields an element $\dL(P) \in \pi_1(M^1,x_P) \subset \pi_1(M^1,M^0)$; cf. Rem.  \ref{pi1}. Here $\subset$ means inclusion of a groupoid into another.

Analogously, if $b\in L^3$ is a blob (i.e. an abstract 3-cell), then the attaching map $\psi_b^3\colon S^2=\bound([0,1]^3)\to M^2$ of the corresponding closed 3-cell $\overline{c_b^3}$ (geometric blob) sends the base-point $*=(0,0,0)$ of $S^2$ to a 0-cell $y_b$ (the base-point of $\overline{c_b^3}$; see Def. \ref{2-lattice}). Hence the attaching map $\psi_b^3$ is a pointed map $\psi_b^3\colon (S^2,*) \to (M^2,y_b)$. Passing to the pointed homotopy class of {$\psi_b^3\colon (S^2,*) \to (M^2,y_b)$} gives rise to an element $\dL(b) \in \pi_2(M^2,y_b)$. 
\end{definition}

\begin{definition}[Notation: $\iL(P)$ and $\iL(b)$]\label{incP}
Cf. Rem. \ref{pi1} and \ref{altdefpi}.
Continuing Def. \ref{notdelta},
let $P\in L^2$ be a plaquette. Then the characteristic map
$\phi_P^2\colon D^2 \to\overline{c_P^2}\subset  M^2$ of the
corresponding closed 2-cell {induces a pointed map
  $(D^2,S^1,*) \to  (M^2,M^1,x_P)$.
Passing to the relative homotopy class of
  {$\phi_P^2$} yields an element $\iota_L(P)\in \pi_2(M^2,M^1,x_P)$.}
Analogously if $b \in L^3$ is an abstract 3-cell, then the characteristic map $\phi_b^3\colon D^3 \to\overline{c_b^3}\subset  M^3$  yields an element {$\iL(b)\in \pi_3(M^3,M^2,y_b)$.}  
\end{definition}

{Note that given $P \in L^2$ and $b \in L^3$ it holds that:}
\begin{equation}
  {\d\big(\iL(b)\big) =\d_L(b) \;\;\;\;\;\textrm{ and }
    \;\;\;\;\; \d\big(\iL(P)\big)= \d_L(P).}
\end{equation}

%Recall \cite[pp16]{hatcher} that
A {\em CW-pair}  \cite[pp16]{hatcher}  is a pair  of spaces
$(X,Y)$
where $X$ has a CW-decomposition and $Y$ is a subcomplex of $X$; cf. Rem. \ref{someprop}. 
\begin{definition}[Relative 2-lattice decomposition of pairs and triples of spaces]\label{relcw} Given a pair $(M,N)$ of topological manifolds (i.e. $N$ is a submanifold of $M$), we say that a 2-lattice decomposition of $M$ is a 2-lattice decomposition of $(M,N)$ if $N$ is a subcomplex of $M$. (Note that the CW-decomposition of $N$ rendered from the fact that $N$ is a subcomplex of $M$ is always  a 2-lattice decomposition; see Rem. \ref{sublatices}.) If $L$ is a 2-lattice decomposition of $M$ that yields a relative 2-lattice decomposition of $(M,N)$, we let $L_N$ be the induced 2-lattice decomposition of $N$. CW-decompositions of a triple $(X,Y,Z)$ of manifolds are defined analogously.
\end{definition}
{If $L$ is a 2-lattice decomposition of $(M,\Sigma)$, we say that $\Sigma$ is {\em cellularly embedded} in $(M,L)$.}

\subsection{Paths on the lattice: the lattice groupoid of $(M,L)$} \label{ss:POL}

 Free groupoids on graphs are discussed in
 \cite{brown_groupoids,Higgins,brown_higgins_sivera}.
 Ref. \cite{brown_hha} in addition addresses  groupoid presentations. 

\begin{definition}[Directed graph; totally intransitive graph]\label{graph}
A directed graph  $(V,E)=(\s,\t\colon E \to V)$ is   a pair of
sets $V$ and $E$, the sets of vertices and of edges of $(V,E)$,
together with a pair of maps $\s\colon E \to V$ and
$\t\colon E \to V$, called the source and target maps.

The maps identify, given an edge
$e$, its source $\sigma(e)$ and target $\tau(e)$ (also called initial
and end-points).
Edges of $(V,E)$ may be represented as
$x\ra{e} y$, where
$x=\s(e)$ and $ y=\t(e)$.
%if we want to emphasise source and target.

A
graph map $(V,E) \to (V',E')$ is given by a pair of set maps $V\to V'$
and $E\to E'$ compatible with initial and end-point of edges.

A
totally intransitive graph is a graph for which source and target maps
coincide.  
\end{definition}

\begin{definition}\label{where}
The functor $U$ sends a groupoid
$G=(\sigma,\tau\colon G_1 \to G_0)$ to its underlying
graph $UG$ --- simply forget  the composition in $G$.
Its left adjoint takes a graph to the free groupoid on the graph; see \cite{Higgins}. 
\end{definition}

A directed graph $(V,E)$ gives rise to another graph
$(V,E\sqcup E^{-1})$ obtained by adding formal reverses to the edges of
$(V,E)$. Here $E^{-1}$  is the set of symbols
$\{e^{-1}\colon e \in E\}$, and we put  $\s(e^{-1})=\t(e)$ and $\t(e^{-1})=\s(e)$. 

\begin{definition}[Free groupoid on a graph. Quantised path on a graph.]
\label{de:fgog}
Let $(V,E)$ be a directed graph.
A quantised path $v \ra{\gamma} v'$ from vertex $v$ to $v'$ on $(V,E)$
is a path on $(V,E\sqcup E^{-1})$, i.e. a sequence
$\gamma={t_1^{\theta_1}}  \dots t_n^{\theta_n}$ where $t_i \in E$ 
and $\theta_i \in \{\pm 1\}$, such that the initial point
of $t_i^{{\theta_i}}$ coincides with the final point of
$t_{i-1}^{\theta_{i-1}}$, and also $v=\s({t_1^{\theta_1}})$ and
$v'=\t(t_n^{\theta_n})$.
%We call $v$ and $v'$ the initial and end-point of the quantised path.
Quantised paths $v \to v$  include the empty path $\emptyset_v$ at
$v$.

%Quantised paths  $v\ra{\gamma} w$ and $w\ra{\gamma'} u$  are
%concatenated  %in the obvious way,
%by concatenating words, yielding  $v\ra{\gamma \gamma'} u$. 

We define an equivalence relation on  %the set of
quantised paths as follows.
Firstly quantised paths $\gamma,\gamma'\colon v \to v'$ are related
if we can modify $\gamma$ into $\gamma'$  by
%using a finite number of
%transformations where sub-paths of the form $t_i^{\pm 1} t_i^{\mp 1}$
%are identified with the appropriate empty sequence, and vice-versa.
deleting a subpath of the form $t^{\pm 1} t^{\mp 1}$.
(Initial and end points of quantised paths remain stable
under this  relation.)
Now take the symmetric-reflexive-transitive closure of this relation.
%It will be clear that this is an equivalence.
If $\gamma$ is a quantised path the equivalent class  % in $\fg{V,E}$
to which it belongs is denoted $[\gamma]$. 

The {\em (free) groupoid} $\fg{V,E}$ is the groupoid with object set $V$;
arrows given by the set of equivalence classes of quantised paths;
and
%Quantised paths up to equivalence
arrows
$v\ra{[\gamma]} w$ and
$w\ra{[\gamma']} u$  are composed by
$\big(v\ra{[\gamma][\gamma']} u\big)=\big(v\ra{[\gamma\gamma']} u\big)$.  
(Note that this composition is well-defined.)
\end{definition}

The notion of {\em freeness} will be important in our construction. 
We will use
freeness in the context of crossed module of groupoids. 
The latter is a non-trivial construction, 
so to prepare for this we now take the opportunity to recall what it means to say
that the groupoid is `free' in Def. \ref{de:fgog} above.

\begin{definition} \cite{Higgins}
A  groupoid  $C$ is free over a graph $G$ if  there is a graph map
$P: G \rightarrow UC$
%(recall from Rem. \ref{where} that $UC$ is the underlying graph of $C$) 
satisfying the following property. 
For every  groupoid $B$ and each graph map $D: G \rightarrow UB$ 
there is a unique groupoid map $D' : C \rightarrow B$ so that the diagram below commutes:
\begin{equation}\label{eq:freeg}
\xymatrix{ & G \ar[dr]_{D}\ar[r]^{P}  
                 & UC \ar@/^0.05pc/[d]^{UD'} \\ 
                     && UB}
\end{equation}
\end{definition}

Straightfoward computations, analogous to the free-group construction prove that:
\begin{lemma} \label{lem:gfree}
  {\rm \cite[8.2.1]{brown_groupoids},\cite{brown_hha},\cite{Higgins}}
The groupoid $\fg{V,E}$ is free over $(V,E)$.
\qed
\end{lemma}

Let $(M,L)$ be a 2-lattice (more generally a CW-complex).
Recall that $M^i$ is the $i$-skeleton of $(M,L)$.
{Note that the characteristic maps $\phi_t^1\colon [0,1] \to M^1$ of the {1-cells} give $(L^0,L^1)$  the structure of a directed graph. Given $t \in L^1$ put $\sigma(t)=\phi_t^1(0)$ and $\tau(t)=\phi_t^1(1)$, where we identified $M^0$ and $L^0$.} 
\begin{definition}[Quantised path in a 2-lattice]\label{de:qp}
 A quantised path on a 2-lattice $(M,L)$ is a quantised path on the graph $(L^0,L^1)$; Def. \ref{de:fgog}. Hence quantised paths $\g$ in $(M,L)$ are obtained by formally chaining together closed 1-cells of $M$ and their reverses: $\gamma= t_1^{\theta_1} t_2^{\theta_2} \dots t_n^{\theta_n}$, where $t_1,\dots,t_n$ are {closed} 1-cells, such that the initial point of $t_i^{\theta_i}$ is the end-point of  $t_{i-1}^{\theta_{i-1}}$
\end{definition}

{The fundamental groupoid $\pi_1(M^1,M^0)$ is defined in Def. \ref{de:fg}. Its set of objects is $M^0$.
Let $u\ra{t} v$ be an edge in $L^1$. Let $\overline{c_u^0},\overline{c_v^0}\in M^0$ be the closed 0-cells corresponding to the abstract 0-cells $u,v\in L^0$.
The characteristic map   $\phi_t^1\colon [0,1]\to M^1$ of $t$ is such
that $\phi_t^1(0)=\overline{c_u^0}$ and
$\phi_t^1(1)=\overline{c_v^0}$. Passing to the homotopy class of
$\phi_t^1$, relative to the boundary $\{0,1\}$ of $[0,1]$,  yields a
morphism $\iota_L(t)$ in the homotopy groupoid $\pi_1(M^1,M^0)$;
cf. Rem. \ref{incP}.} Since $\fg{L^0,L^1}$ is free,
$\iota_L$ extends to a groupoid map
$$
\iota\colon \fg{L^0,L^1} \to \pi_1(M^1,M^0) .
$$

The following can be seen as a generalisation of the van Kampen
theorem \cite{hatcher},  for spaces with a set of base points.  Proofs
are in \cite[9.1.5]{brown_groupoids},\cite{brown_higgins_sivera},\cite{brown_hha}.
This holds more generally for CW-complexes.

\begin{Theorem}
\label{pr:freee}
Let $(M,L)$ be a 2-lattice. The groupoid map $\iota \colon \fg{L^0,L^1} \to \pi_1(M^1,M^0)$ is an isomorphism. Hence $\pi_1(M^1,M^0)$
is isomorphic to the free groupoid
$\fg{L^0,L^1}$,
with
{set of objects being $M^0=L^0$ and with a free generator $u \ra{t} v$
for each edge $t \in L^1$. (Here $u$ and $v$ are the source and
target of 
$t$.)}
\qed
\end{Theorem}

In particular, for any group $G$ and for
any map $f\colon L^1 \to G$ there exists a unique groupoid map 
$f'\colon \pi_1(M^1,M^0) \to G$ whose value on each {$\iota_L(t)$}, $t\in L^1$ an edge, is $f(t)$.
(The same holds if $G$
is a groupoid, except that we must pay attention to sources and targets.) 
As we will see in \S\ref{ss:HOGC}, 
this observation lies at the heart of the realisation of gauge theory
that we will lift to the higher case. 

We will hence see $\pi_1(M^1,M^0) \cong \fg{L^0,L^1}$ as being the lattice groupoid of $(M,L)$.

\section{Higher order gauge configurations and discrete 2D holonomy for surfaces embedded in 2-lattices}\label{2dholonomy}

{In order to establish a template for the `higher' construction, we
  start with a suitable characterisation of
  {\em ordinary} gauge configurations, and of their holonomy
  along cellularly embedded circles $S^1$.}

\newcommand{\Fi}{\F^1}  %% gauge conf

\subsection{Gauge configurations, discrete 1D parallel transport and holonomy along circles}\label{dptf}

%Let $(M,L)=(M,L=(L^0 ,L^1 , ... ))$ be a 2-lattice {\S \ref{sec:2-lattices}.}

\begin{definition}%[Gauge configuration]
  \label{gconfig}
Let $G$ be a group. A {\em gauge configuration} on a 2-lattice
$(M,L)$ is a map 
$$
\Fi\colon L^1\to G
$$
%from the set of 1-cells of $L$ to $G$.
We write %$\F^1$ as:
$
%\Fi \colon t \in L^1 \longmapsto
\Fi(t)= g_t ,  %\in G.
$ for each edge $t \in L^1$.
\end{definition}

%\noindent
By the freeness of $\pi_1(M^1,M^0)\cong  \fg{ L^0 , L^1} $
{(Lem. \ref{lem:gfree} and {Thm.} \ref{pr:freee})}  a gauge
configurations $\Fi$ extends to a uniquely defined   groupoid morphism
$$
\Phi_{\F^1}\colon  \pi_1(M^1,M^0) \rightarrow G.
$$ 
Here $G$ is regarded as a groupoid with one object.
All groupoid maps $\fg{ L^0 , L^1 } \rightarrow G$
arise this way. Hence:

\begin{Theorem}[The discrete parallel transport of a gauge configuration]
  \label{d1dh}
Let $(M,L)$ be a 2-lattice. Let $G$ be a group. The correspondence
${\Fi} \mapsto \Phi_{\Fi}$ yields a one-to-one correspondence
between gauge configurations $\Fi$ %\colon L^1 \to G$ {on $(M,L)$}
and groupoid maps $\Phi_{\F^1}\colon \pi_1(M^1,M^0) \to G$. 
\qed
\end{Theorem}

\noindent Those groupoid maps $\Phi_{\F^1}\colon \pi_1(M,M^1)\to G$ associated to a gauge configuration $\Fi$ will sometimes be called discrete parallel transport functors, in analogy with the differential-geometrical construction in \cite{schreiber_waldorf1,martins_picken}.

\begin{definition}%[Notation]
  \label{n1} %Cf. Def. \ref{de:fgog} and \ref{de:qp}.
Let $\g=t_1^{\theta_1} t_2^{\theta_2} \dots t_n^{\theta_n}$
be a quantised path in $(M,L)$.
Let $\F^1$ %\colon t \in L^1 \mapsto g_t \in G$
be a gauge configuration. Put:
$$
g_\g=g_{t_1}^{\theta_1}g_{t_2}^{\theta_2}\dots g_{t_n}^{\theta_n} .
$$
\end{definition}

\noindent Let $[\g] \in \pi_1(M^1,M^0)$ be the equivalence class of $\g$.
{Given {Thm.} \ref{pr:freee}, it is clear that:
$\Phi_{\F^1}([\g])=g_{\g}.$}

\begin{definition}[Holonomy along a circle: combinatorial definition]
\label{holalongcircle}
{Let $\F^1\colon L^1 \to G$ be a gauge configuration on a 2-lattice decomposition $(M,L)$.}
Let $C$ be an oriented circle $S^1$ embedded in $M$.
Suppose %that $C$ can be spanned by closed 1-cells of $L$, i.e. suppose
that $L$ is a 2-lattice decomposition of $(M,C)$;
Def. \ref{relcw}.
Let $v \in C 
\cap M^0=C^0$. 
{{Starting} at the vertex $v$, the path around the circle $C$ in the positive
  direction therefore traces a quantised path
  $\gamma={t_1^{\theta_1} t_2^{\theta_2} \dots t_n^{\theta_n}}$,
  connecting $v$ to $v$.
  The holonomy ${\rm Hol}_v^1(\Fi,C,L)$ of
  $\Fi$, along $C$, with initial point $v$, is defined as:} 
$$
{\rm Hol}_v^1(\Fi,C,L) = g_\g = \Phi_{\Fi}([\g]){\in G}.
$$
\end{definition}

\noindent Note that {the holonomy ${\rm Hol}_v^1(\F^1,C,L)$ of $\F^1$ along $C$}
depends on the chosen starting point $v \in C\cap M^0$
only by conjugation by an element of $G$.

\begin{remark}[Holonomy along a circle: algebraic topological definition]
\label{holalongcircle-conceptual}
% A more conceptual, algebraic topological, definition of ${\rm
%   Hol}_v(\Fi,C,L)$, which generalises to the higher case, is the following.
{Recall $S^1=\d D^2$.  
Choose} a homeomorphism $f\colon \d D^2 \to C$,
preserving the orientation, sending the base-point $*=(0,0)$ of $\d D^2$
to $v$.
By elementary algebraic topology {(as $\pi_1(S^1)=\mathbb{Z}$)},
any two such homeomorphisms are homotopic, relative to $*$.
Let $i_v(C)=f_*(1)$, where $f_*\colon
\pi_1(S^1,*)\cong \Z \to \pi_1(C,v)\subset \pi_1(C,C^0)$ is the
induced map on homotopy groups.
{Clearly $i_v(C)={t_1^{\theta_1}  t_2^{\theta_2} \dots t_n^{\theta_n}}$, as in
  Def. \ref{holalongcircle}.} 
 Let $\F_C^1$ be the restriction of $\F^1$ to the induced lattice
 decomposition $L_C$ of $C$; see Def. \ref{relcw}.
 {It hence clearly holds that:}
$$
 {\rm Hol}_v^1(\F^1,C,L)=\Phi_{\F^1_C}(i_v(C)).
$$
{Here $\Phi_{\F^1_C} \colon \pi_1(C,C^0) \to G$  is the  discrete parallel transport of $\F^1_C$.}
\end{remark}

\newcommand{\Chi}{\chi}
\begin{remark}
{
  Although gauge configurations can formally be defined separately
  from Hamiltonians, as above, they have no physical meaning without an
  associated Hamiltonian. In particular parts of the structure of
  space-time
  are encoded  {\em in a model}
  not in the gauge
  configuration but in the Hamiltonian. We are not ready to give the
  `higher'
  Hamiltonian \S \ref{thkm} (the higher Kitaev model) that will be the central focus of this paper, but we can   already  give an illustrative `standard' example, which also serves as a template for the Kitaev quantum double model \cite{kitaev2003fault}; see \S\ref{comparison}.
  Given a lattice $(M,L)$ and a group $G$, and hence the  %object
  set
  $\catfun$ of functors between $\freegl$ and $G$,
  we may define for each
  character $\Chi: G \rightarrow \CC$ a Wilson action 
  $H_{\Chi} : \catfun \rightarrow \RR$ by (cf. Def. \ref{notdelta}):
  \[
  H_{\Chi} (F) = \sum_{P \in L^2}^{}  Re(\Chi(F(\dL(P))) 
  \]
  where $Re:\CC\rightarrow\RR$ is the real part
  (see e.g. Wilson \cite{Wilson74}, \cite[\S8]{Kogut79}, \cite[\S10.2]{Martin91}
or \cite[\S1.2]{Oeckl}).
  Note that this depends strongly on  the cell-decomposition of $M$, as
well as $M$. %, i.e. it is a metric and not a topological theory. 
The {\em main} thing to note at this point is that the sum is over
plaquettes, thus the Hamiltonian is sensitive to the 2-dimensional
structure in the lattice (whereas the gauge configuration `sees' only
the underlying graph $L^1$). We will return to this point later.
}
\end{remark}

\subsection{Higher order gauge configurations} \label{ss:HOGC}

{In this paper, we  consider fake-flat 2-gauge configurations on a
2-lattice $(M,L)$, {as discretised models for higher gauge fields \cite{baez_schreiber,BaezHuerta11,martins_picken}.}
Instead of a gauge group we have a crossed module of groups;
Def. \ref{cm}.
The main aim  %of the discussion from  now on
is to extend Thm. \ref{d1dh}, Def. \ref{holalongcircle} and
Rem. \ref{holalongcircle-conceptual} to the case of fake-flat 2-gauge
configurations.} This yields 2-dimensional (2D) notions of parallel
transport which restrict to notions of  2D holonomy along surfaces,
cellularly embedded in $M$.
We will address the 2-sphere and 2-disk case,
which play an important role in higher Kitaev models.

\subsubsection{Fake-flat 2-gauge configurations}\label{ff2gc}

Continuing the work of Yetter and Porter
\cite{yetter_tqft,porter_tqft},
fake-flat 2-gauge configurations on CW-complexes were  defined in
\cite{martins_crossed,martins_porter,martins_cw_complex}.
Their algebraic topology interpretation was developed {therein}, following
the work of Brown and Higgins on fundamental crossed modules of pairs
of spaces and 2-dimensional van Kampen theorem
\cite{brown_2dvk,brown_higgins_colimits,brown_higgins_cubes,brown_higgins_sivera}. Homotopy quantum field theory applications of fake-flat 2-gauge configurations appear in \cite{PorterTuraev} (and were there called ``formal $\cal C$ maps'').
The inherent (and independently addressed) differential-geometric
higher gauge theory for 2-bundles  with {a  2-connection} appears in
\cite{baez_schreiber,schreiber_waldorf1,schreiber_waldorf2,martins_picken,Zu1}. The term ``fake-flatness'' appeared originally in the context of  gerbes with connection; see \cite{BrM}.

\begin{definition}[2-gauge configuration]\label{2gcc}
Let ${\Gc}=(\dG\colon E \to G,\trr)$ be a crossed module of groups.
Given $\Gc$, a
2-gauge configuration $\F = (\F^1,\F^2)$, based on a 2-lattice {$(M,L)=\big(M,L=(L^0,L^1,L^2,\dots)\big)$,} is
given by:
\begin{itemize}
\item A map $\F^1\colon L^1 \rightarrow G$, denoted:
  $t \in L^1 \mapsto g_t\in G$, or $t\in L^1 \mapsto \F^1(t)\in G$.
\item A map $\F^2: L^2 \rightarrow E$, denoted:
  $P\in L^2  \mapsto e_P\in E$, or  $P\in L^2 \mapsto \F^2(P)\in E$.

 \end{itemize}
\end{definition}
\noindent A 2-gauge configuration gives rise to a groupoid map $\Phi_{\F}=\Phi_{\F^1}\colon \pi_1(M^1,M^0) \to G$; see Thm. \ref{d1dh}.

We mainly consider {\em fake-flat} 2-gauge configurations. Let us explain what this means.

\begin{definition}[Fake-flat 2-gauge configuration]\label{fakeflatconf} A 2-gauge configuration $\F = (\F^2,\F^1)$, based on a 2-lattice $(M,L)$, is said to be fake-flat if for each plaquette {$P\in L^2$} it holds that (recall the notation of  Rem. \ref{notdelta}):
   $$\dG(e_P)=\Phi_{\F^1}(\dL(P)).$$
Given a crossed module $\Gc$, we denote the set of fake-flat 2-gauge configurations in $(M,L)$ as $\Theta(M,L,\Gc)$.
   \end{definition}

Let us give more explanation on the definition of fake-flatness. This
is one of the points where the fact that we are restricting to
2-lattices Def. \ref{2-lattice}  makes our discussion a lot
simpler. One more definition is needed.

\begin{definition}[Quantised boundary $\dL^Q(P)$ of a plaquette]\label{qbp}
 Let {$P\in L^2$} be a plaquette of a 2-lattice $(M,L)$. Let $\psi_P^2\colon S^1 \to M^1$ be the attaching map of the correspondent closed 2-cell $\overline{c_P^2}.$ We are given a  CW-decomposition ${Z_{P}}$ of $S^1$, which contains $*\in S^1$ as a $0$-cell, such that $\psi_P^2\colon S^1 \to M^1$ is cellular and satisfies the conditions of Def. \ref{2-lattice}.  We will in addition suppose that all characteristic maps {$\phi_{{\gamma_i}}\colon [0,1] \to S^1=\bound ([0,1]^2)$} of the closed 1-cells $\gamma_1,\dots,\gamma_n$ of ${Z_{P}}$ are oriented counterclockwise; see Fig. \ref{R}. We also assume that the 1-cells $\gamma_1,\dots,\gamma_n$ of $Z_P$ appear in that order, as we ``travel'' counterclockwise from $*$ to $*$ around $S^1$.
\begin{figure}[H]
\centerline{\relabelbox
\epsfysize 3cm
\epsfbox{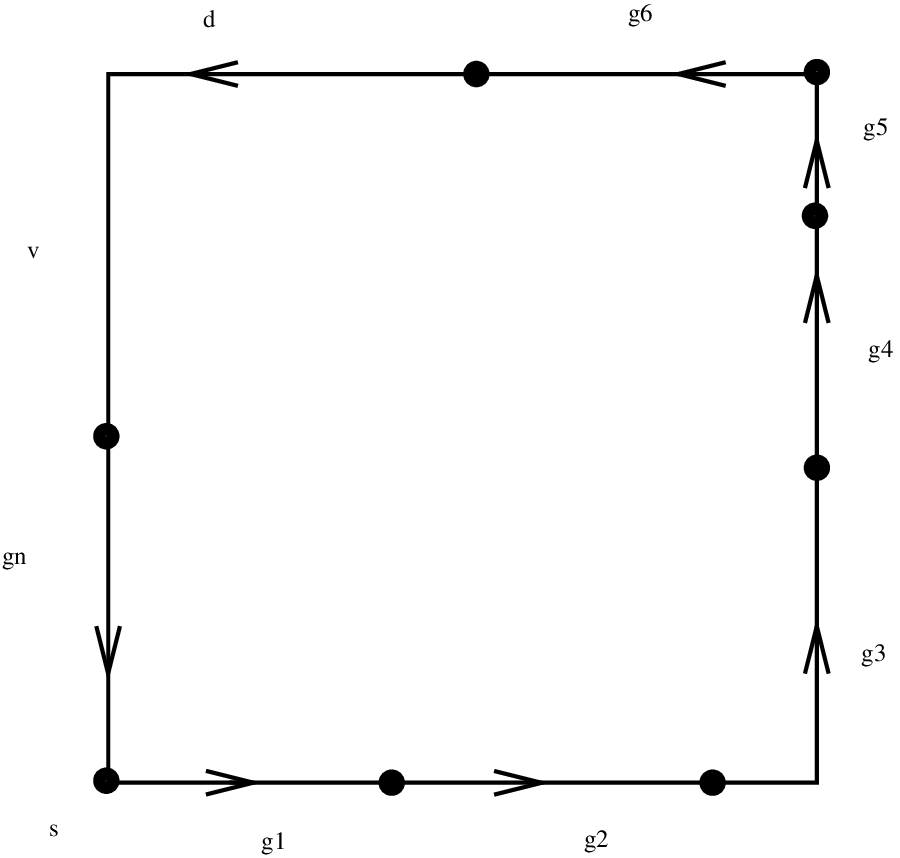}
\relabel{g1}{$\gamma_1$}
\relabel{g2}{$\gamma_2$}
\relabel{g3}{$\gamma_3$}
\relabel{g4}{$\gamma_4$}
\relabel{g5}{$\gamma_5$}
\relabel{g6}{$\gamma_6$}
\relabel{d}{$\dots$}
\relabel{v}{$\vdots$}
\relabel{gn}{$\gamma_n$}
\relabel{d}{$\dots$}
\relabel{s}{$*$}
\endrelabelbox}
\caption{The CW-decomposition ${Z_{P}}$ of the 1-sphere $S^1$. \label{R}}
\end{figure}

 We let $x_P=\psi_P^2(*)\in M^1$ be the base-point of the closed 2-cell $\overline{c_P^2}$ corresponding to $P$.
 Suppose that $\psi_P^2$ is not constant. Then for each $i\in \{1,\dots,n\}$, $\psi_P^2(\gamma_i)$ is a closed 1-cell $t_i$ of $M$ and $\psi_P^2$ restricts to a homeomorphism $\gamma_i \to t_i$. The closed 1-cell $t_i\subset M$ is oriented by its characteristic map. We put $\theta_i=1$ if the  restriction $\gamma_i \to t_i$ of $\psi_P^2$ preserves orientation and $\theta_i=-1$, otherwise. The quantised boundary of $P$ is defined to be the following quantised path in $(M,L)$ (Def. \ref{de:qp}), {connecting} $x_P$ to $x_P$: $\dL^Q(P)=t_1^{\theta_1}t_2^{\theta_2}...t_n^{\theta_n}$.

\noindent {Otherwise, if $\psi_P^2\colon S^1 \to M^1$ satisfies $\psi_P^2(S^1)=x_P$, we define the quantised boundary of $P$ as $\dL^Q(P)=\emptyset_{x_P}$}.
\end{definition}
\noindent An example appears in Ex. \ref{fake0}.

\medskip
Let {$P\in L^2$}. By passing to the equivalence class of the
quantised path $\dL^Q(P)$ (cf. the construction of
$FG\langle L^0,L^1 \rangle \cong \pi_1(M^1,M^0)$
in Def. \ref{de:fgog} and Prop. \ref{lem:gfree}) yields  {$\dL(P)\in \pi_1(M^1,M^0)$} in Def. \ref{notdelta}. Hence:
\begin{proposition}\label{fake-flat-conf-inpractice}
Let $(M,L)$ be a 2-lattice. Let ${\Gc}=(\dG\colon E \to G,\trr)$ be a crossed module of groups. A
2-gauge configuration $\F = (\F^2,\F^1)$ is fake-flat if, and only if:
\begin{itemize}
\item For each plaquette $P$ for which $\psi_P^2$ is not constant, putting  $\dL^Q(P)=t_1^{\theta_1}t_2^{\theta_2}...t_n^{\theta_n}$,  it holds that:
\begin{equation}\label{fakeflatexplicit}
 \dG(e_P) \; = g_{t_1}^{\theta_1} \dots g_{t_n}^{\theta_n}=\Psi_{\F^1}([\d L^Q(P)]).
\end{equation}
 \item If $P$ is a plaquette for which $\psi_P^2(S^1)=x_P$ it should hold that {$e_P \in \ker(\dG\colon E \to G)\subset E$}.
 \end{itemize}
\end{proposition}

\begin{example}\label{fake0} Consider the square $D^2=[0,1]^2$, with the {2-lattice decomposition} indicated in the middle of Fig. \ref{fake}, namely $L=(L^0,L^1,L^2)=(\{v_1,v_2,v_3,v_4\},\{t_1,t_2,t_3,t_4\}, \{P\})$. (Abstract cells and the corresponding closed cells are denoted in the same way.) The geometric 2-cell $\overline{e_P^2}=P$ attaches along the identity map $\psi_P^2\colon S^1 \to S^1$, hence  the attaching map $\psi_P^2\colon S^1 \to S^1$ is positively oriented. The  CW-decomposition ${Z_{P}}$ of $S^1=\bound ([0,1]^2)$ (Def. \ref{2-lattice}) has a vertex for each corner and a positively oriented edge for each side of $[0,1]^2$.
 \begin{figure}[H]
 \centerline{\relabelbox
 \epsfysize 3.7cm
 \epsfbox{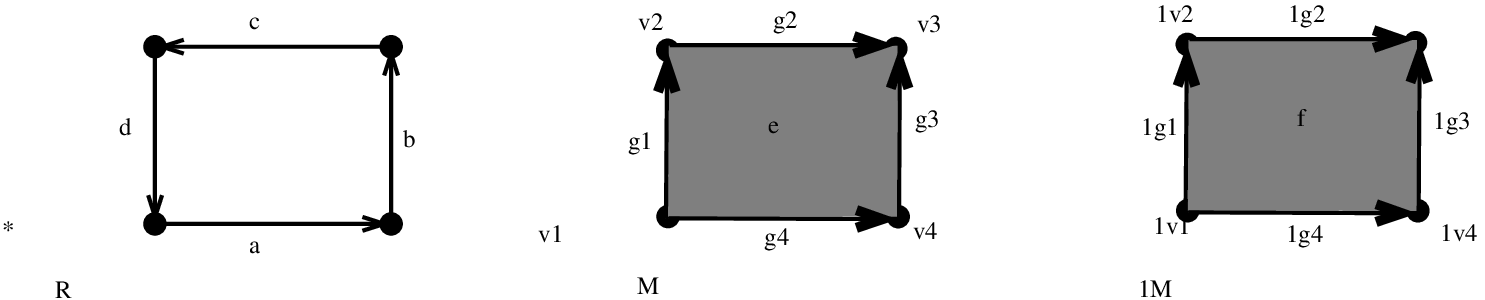}
 \relabel{v1}{$x_P=v_1$}
 \relabel{g1}{$t_1$}
 \relabel{v2}{$v_2$}
 \relabel{g2}{$t_2$}
 \relabel{v3}{$v_3$}
 \relabel{g3}{$t_3$}
 \relabel{v4}{$v_4$}
 \relabel{g4}{$t_4$}
 \relabel{e}{${\bf P}$}
 \relabel{M}{$\dL^Q(P) \; = t_4  t_3 t_2^{-1} t_1^{-1}$}
 \relabel{a}{$\g_{1}$}
 \relabel{b}{$\g_{2}$}
 \relabel{c}{$\g_{3}$}
 \relabel{d}{$\g_{4}$}
 \relabel{1g1}{$g_{t_1}$}
 \relabel{1g2}{$g_{t_2}$}
 \relabel{1g3}{$g_{t_3}$}
 \relabel{1g4}{$g_{t_4}$}
 \relabel{f}{${\bf e_{P}}$}
 \relabel{1M}{$\dG(e_P) \; = g_{t_4} g_{t_3} g_{t_2}^{-1} g_{t_1}^{-1} $}
 \relabel{R}{${Z_{P}}, \textrm{ the CW-decomposition of } S^1$ }
 \relabel{*}{$*=(0,0)$}
 \endrelabelbox}
 \caption{A 2-lattice decomposition $L$ of $D^2$, where ${Z_{P}}$ is the corresponding CW-decomposition of $S^1$: cf. Def, \ref{2-lattice}. The base-point $x_P$ of $P$ is $v_1$. We also show a fake-flat 2-gauge configuration {in $(D^2,L)$}.\label{fake}}
 \end{figure}
{For this example, the} quantised boundary of the plaquette $P$ is
 $\dL^Q(P) \; = t_4  t_3 t_2^{-1} t_1^{-1}.$ Hence a 2-gauge configuration of $([0,1]^2,L)$ is given by four elements $g_{t_1},g_{t_2},g_{t_3},g_{t_4}$ of $G$, the colours of the edges, $t_1$, $t_2$, $t_3$, $t_4$, and an element $e_P\in E$, colouring $P$. The fake flateness conditions says:
 $ \dG(e_P)=g_{t_4} g_{t_3} g_{t_2}^{-1} g_{t_1}^{-1}.$
\end{example}
Let $\PhiM$ denote the set of fake-flat {2-gauge} configurations.
Note that $\PhiM$ is non-empty. In particular the `naive vacuum'
$\vac$ given by $e_P=1_E$ for all plaquettes $P$ of $(M,L)$ and $g_t=1_G$ for all edges $t$ is fake-flat. Here $1_G$ and $1_E$ denote the identities of $G$ and $E$.

\subsection{%Review of
  On Whitehead theorem, 2-gauge configurations and the lattice 2-groupoid}
\label{revWhite}

 Let $(M,L)$ be a  2-lattice. Passing to the 0, 1 and 2-skeletons of the corresponding CW-decomposition of $M$, yields a triple $(M^2,M^1,M^0)$ of  locally path-connected spaces, where $M^0$
intersects {non-trivially any path-connected} component of $M^1$ and $M^2$. Utilising Def. \ref{defpi2}, we can form the fundamental crossed module $\Pi_2(M^2,M^1,M^0)$; Thm. \ref{defpi2-2}. This crossed module  plays the role of lattice 2-groupoid of $(M,L)$.

Observe that to make use of our {fake-flat 2-gauge configurations} we need
corresponding lifts of Lem. \ref{lem:gfree} and {Thm.} \ref{pr:freee}.
Analogously to {Thm.} \ref{pr:freee},
the crossed module $\Pi_2(M^2,M^1,M^0)$ is free
on the attaching maps of the geometric plaquettes {(i.e. closed 2-cells) of $(M,L)$.
This result (which holds in the general case of CW-complexes) is an
old result due to JHC Whitehead \cite{Whitehead}.
Modern treatments can be found in
\cite{brown_higgins_sivera,brown_hha,brown_whitehead,brown_2dvk,Baues}.}

{Consider groupoids $H=(\sigma,\tau \colon H_1 \to H_0)$ and $H'=(\sigma,\tau \colon H_1'\to H_0')$. Throughout this subsection, we use the following notation. If $f\colon H \to H'$ is a groupoid map,   put $f_{\rm MOR}\colon H_1 \to H_1'$ to be the  restriction of $f$ to morphisms and $f_{\rm OBJ}\colon H_0 \to H'_0$ to be the restriction of $f$ to objects.
If $(\d\colon E \to G)$ is a a crossed module of groupoids, thus $E$ and $G$ have the same set $C$ of objects, it holds that $\d_{\rm OBJ}\colon C \to C$ is the identity map.} 

In order not to excessively load our formulae, we use the same notation for the groupoid $\pi_2(M^2,M^1,M^0)$ and for its set of morphisms, and the same for  $\pi_1(M^1,M^0)$. Which one is meant is clear from the context. {The coinciding source and target maps in the groupoid $\pi_2(M^2,M^1,M^0)$ are given by the obvious map:} $${\beta\colon \pi_2(M^2,M^1,M^0)\doteq \bigsqcup_{v \in M^0}  \pi_2(M^2,M^1,v) \to M^0.}$$

First  we specify what ``crossed module freeness'' is.
Let $G=(\sigma,\tau \colon G_1 \to G_0)$ be a groupoid. 
Let also $K$ be a set mapping to $G_0$, through a map $\beta_0$.
Suppose also that we have a map $\d_0\colon K \to G_1$ that makes the
diagram below commute (therefore $\beta_0=\tau \circ \d_0$ and
 $\beta_0=\sigma \circ \d_0)$:
\begin{equation}\label{basepoints}
\xymatrix{ & K\ar[dr]_{\beta_0}\ar[r]^{\d_0}  & G_1 \ar@/^0.5pc/[d]^{\sigma} \ar@/_0.5pc/[d]_{\tau}\\
                     && G_0 }
\end{equation}
(This in particular means that $\d_0$ is a map from $K$ into the set
of arrows {(morphisms)} in $G$ that have the same source and target.)
Let  {$F=(\beta,\beta\colon F_1 \to G_0)$} be a totally intransitive groupoid
with the same set of objects as $G$.
We say that a crossed module of groupoids $(\d\colon F \to G,\trr)$
is free  on $\d_0\colon K \to G_1$ {(or more precisely on $\d_0\colon K \to G_1$ and $\beta_0\colon K \to G_0$ as in \eqref{basepoints})}
if there exists an inclusion {(set)} map ${\rm inc}\colon K \to F_1$ making the diagram below commute:
 \begin{equation}\label{fcm2}
{\xymatrix{ & K \ar[rd]|{\d_0}\ar[r]^{\rm inc} \ar@/_1pc/[rdd]_{\beta_0} &F_1\ar[d]|{\d_{\rm MOR}}\ar@/^2pc/[dd]^{\beta} \\
                          && G_1\ar@/^0.5pc/[d]^{\sigma} \ar@/_0.5pc/[d]_{\tau} \\
                          && G_0}}
\end{equation}
 such that the following universal property is satisfied:
{\it Given any crossed module $(\d'\colon E' \to G',\trr)$ of
  groupoids, where $G'=(\sigma',\tau'\colon G_1' \to G_0')$ and $E'=(\beta',\beta'\colon E_1' \to G_0')$, and  any groupoid map $\phi\colon G \to G'$, and any set map
  $\psi_0\colon K \to E'_1$, 
%(the set of arrows of $E'$) 
 such that
  $\d'_{\rm MOR}\circ \psi_0=\phi_{\rm MOR} \circ \d_0$, there exists a unique
  groupoid map $\psi\colon F \to E'$, {with $\psi_{\rm OBJ}=\phi_{\rm OBJ}$,} making the diagram below
  commutative:
\begin{equation}\label{fcm3}
 \hskip-1cm\xymatrix{
             &K\ar@/^1pc/[rr]^{\psi_0}\ar[rd]_{\d_0}\ar[r]|<<<<{\rm inc}&F\ar[d]_{\d} \ar@{-->}[r]_>>>>>\psi & E'\ar[d]^{\d'} \\
             & &G\ar[r]_>>>>>\phi &G'}\xymatrix@R=1pt{\\\textrm{  \quad or leaving in all information: }}
\hskip-0.7cm{\xymatrix{ & K \ar@/^1.5pc/[rrrr]^{\psi_0}\ar[rd]|{\d_0}\ar[r]|{\rm inc} \ar@/_1pc/[rdd]_{\beta_0} &F_1\ar@{-->}[rrr]_{\psi_{\rm MOR}} \ar[d]|{\d_{\rm MOR}}\ar@/^2pc/[dd]^<<<<<{\beta}  & & & E_1'\ar[d]|{\d'_{\rm MOR}} \ar@/^2pc/[dd]^<<<<<{\beta'} \\
                          && G_1\ar[rrr]_{\phi_{\rm MOR}}|<<<<<<\quad\ar@/^0.5pc/[d]^{\sigma} \ar@/_0.5pc/[d]_{\tau} &&& G'_1 \ar@/^0.5pc/[d]^{\sigma'} \ar@/_0.5pc/[d]_{\tau'}  \\
                          && G_0\ar[rrr]_{\phi_{\rm OBJ}=\psi_{\rm OBJ}}  &&& G'_0 }}
\end{equation}
and so that the pair $(\psi,\phi)$ of groupoid maps is a crossed
module map $(\d\colon F \to G) \to (\d'\colon E' \to G')$.
}

\begin{lemma}\label{modelfree}
Given    $\d_0\colon K \to G_1$ {as in \eqref{basepoints}, the totally intransitive groupoid   $F=(\beta,\beta\colon F_1 \to G_0)$, i.e. the top groupoid of the free crossed module on $\d_0\colon K \to G_1$,} is uniquely specified up to isomorphism by the  universal property above. A model for $F$ is the following. First of
 all note that we have a totally intransitive graph $K'$ having $G_0$ as set of vertices and the set of edges being the set of all pairs $(g,k)$, where $g \in G_1$ and $k\in K$ are such that $\tau(g)=\beta_0(k)$. The coinciding source and target maps {of $K'$} are given by $(g,k) \mapsto \sigma(g)$. Edges of $K'$ therefore take the form:
 $$
\sigma(g) \ra{(g,k)} \sigma(g),\textrm{ where } g \in G_1 \textrm{ and
} k \in K \textrm{ are such that } \tau(g)=\beta_0(k).
$$
We then form the free groupoid $FG(K')$, which is a totally
intransitive groupoid having $G_0$ as set of objects.
We have a groupoid map $\partial\colon FG(K') \to G$ which is the identity on objects and on {generating} morphisms is:
$${\partial_{\rm MOR}\big (\sigma(g) \ra{(g,k)} \sigma(g)\big )
\; = \; \sigma(g)
\ra{g \d_0(k) g^{-1}} \sigma(g).}
$$

The groupoid $FG(K')$ has a natural left action by automorphisms of
the groupoid $G$. On generators of $FG(K')$, {the action} takes the
following form: If $g,h \in {G_1}$ are such that {$\tau(h)=\sigma(g)$,
  and $k \in K$ is such that $\beta_0(k)=\tau(g)$, put
  $h.(g,k)=(hg,k)$.}
Then together with the map {$\d\colon FG(K') \to  G$,}
nearly all conditions that crossed modules of groupoids must
satisfy {(Def. \ref{cm} and \ref{cmg})} hold,
except for the {2nd Peiffer relation.}
The groupoid $F$ is obtained from $FG(K')$ by dividing out the 2nd Peiffer relations, in the obvious way.\end{lemma}

\begin{proof} Routine. Details are in
\cite[\S 7.3(ii)]{brown_higgins_sivera} and \cite{brown_hha}. \end{proof}

\medskip

Let $(M,L)$ be a 2-lattice. Recall Rem. \ref{notdelta} and \ref{incP}, and Def. \ref{2-lattice}. Going back to $\Pi_2(M^2,M^1,M^0)$, to each plaquette {$P\in L^2$} we can associate elements $\dL(P) \in \pi_1(M^1,x_P)\subset \pi_1(M^1,M^0)$ and $\iL(P) \in \pi_2(M^2,M^1,x_P)$, where $x_P=\psi_P^2(*)$ is the base-point of the {closed 2-cell $\overline{c_P^2}$} corresponding to $P$. Also $\d_{\rm MOR}(\iL(P))=\dL(P)$. {In particular we  have a commutative diagram as in (\ref{basepoints},\ref{fcm2},\ref{fcm3}):}
$${\xymatrix{ &L^2\ar@/^1pc/[rr]|\iL  \ar[dr]^{\dL}\ar@/_1pc/[ddr]_{P \mapsto x_P}&   & \pi_2(M^2,M^1,M^0)\ar@/^3pc/[dd]<2ex>^{\beta}\ar[d]^{\d_{\rm MOR}}\\ &   & \pi_1(M^1,M^0)  \ar@/^0.5pc/[d]^{\sigma} \ar@/_0.5pc/[d]_{\tau} \ar[r]|\id& \pi_1(M^1,M^0)\ar@/^0.5pc/[d]^{\sigma} \ar@/_0.5pc/[d]_{\tau}\\
                     && M^0 \ar[r]|\id &M^0 }}$$
Hence we can form the free crossed module
{$(\d\colon F\to \pi_1(M^1,M^0),\trr)$}
on $\dL\colon L^2 \to \pi_1(M^1,M^0)$,
where
{$F=(\beta,\beta\colon F_1 \to M^0)$.}
And we have a unique {groupoid} morphism $\iota\colon F \to \pi_2(M^2,M^1,M^0)$, {which is the identity on objects,} that makes the diagram below commutative, and so that {$(\iota,\id)\colon (F\to \pi_1(M^1,M^0),\trr) \to \Pi_2(M^2,M^1,M^0)$}  is a crossed module map:
{\begin{equation}\label{defiota}\xymatrix{ &L^2\ar[r]|{\rm inc}\ar@/^1pc/[rr]|\iL  \ar[dr]^{\dL}\ar@/_1pc/[ddr]_{P \mapsto x_P}& F_1\ar@/^2.5pc/[dd]<1ex>^<<<<<<<{\beta}|\hole\ar@{-->}[r]_<<<<<<<{\iota_{\rm MOR}} \ar[d]^{\d_{\rm MOR}}  & \pi_2(M^2,M^1,M^0)\ar@/^3pc/[dd]<2ex>^{\beta}\ar[d]^{\d_{\rm MOR}}\\ &   & \pi_1(M^1,M^0)  \ar@/^0.5pc/[d]^{\sigma} \ar@/_0.5pc/[d]_{\tau} \ar[r]|\id& \pi_1(M^1,M^0)\ar@/^0.5pc/[d]^{\sigma} \ar@/_0.5pc/[d]_{\tau}\\
                     && M^0 \ar[r]|\id &M^0 }\end{equation}}

\begin{Theorem}[Whitehead Theorem]\label{wt}
  Let $(M,L)$ be a 2-lattice,
  or indeed any CW-complex.
Then the crossed module
$
\Pi_2(M^2,M^1,M^0)
= (\d \colon\pi_2(M^2,M^1,M^0)
\rightarrow \pi_1(M^1,M^0) , \trr )
$
is free on
$\dL\colon L^2 \to \pi_1(M^1,M^0)$. {Specifically}, the map {$\iota\colon F \to \pi_2(M^2,M^1,M^0)$ {defined from \eqref{defiota}} is an isomorphism of groupoids, and, moreover, the pair $(\iota,\id)\colon (\d\colon F\to \pi_1(M^1,M^0),\trr) \to \Pi_2(M^2,M^1,M^0)$ is an isomorphism of crossed modules.}
\end{Theorem}
\begin{proof}{See \cite[\S 6]{brown_higgins_sivera} and \cite{brown_2dvk,brown_hha,brown_higgins_colimits,brown_higgins_cubes}, where
 Whitehead's theorem is deduced from the  more general 2-dimensional van
 Kampen theorem, and also \cite{brown_whitehead,Baues}. {(We note that Whitehead's original proof was done for crossed modules of groups rather than of groupoids, and only considered spaces with a single base-point.)}}\end{proof}

\begin{remark}\label{genpi2}
Note that Whitehead's theorem together with the construction of free
crossed modules (Lem. \ref{modelfree}) implies that the totally intransitive
groupoid $\pi_2(M^2,M^1,M^0)$ is generated by
$\sigma(g) \ra{g \trr  \iL(P)} \sigma(g)$,
where $g \in \pi_1(M^1,M^0)$ and $P \in L^2$ is
such that:
$ x_P= \tau (g)$; recall that $\sigma(g)$ and $\tau(g)$ are the initial and end-points of $g$ and $x_P$ is the base-point of $P$. {This will have a primary role in the construction of the 2-dimensional holonomy of a fake-flat 2-gauge configuration along a cellularly embebbed surface in \S \ref{atd}.}
\end{remark}

{We will only use the universal property \eqref{fcm3} in the case when
$\Gc=(\d' \colon E' \to G',\trr)$ is a crossed module of groups. In this
case there is not much to worry about maps on object sets of
groupoids, as $E'$ and $G'$ are groupoids with a single object. We hence will not display the related part of the commutative diagrams. Given groupoid maps $f\colon \pi_1(M^1,M^0) \to G'$ and $f'\colon \pi_2(M^2,M^1,M^0) \to E'$ we put $f_{\rm MOR}=f$ and $f'_{\rm MOR}=f'$. In $\Gc$, we put $\d'_{\rm MOR}=\d'$.  Recall that we use the same notation for the groupoid $\pi_2(M^2,M^1,M^0)$ and for its set of morphisms, and the same for  $\pi_1(M^1,M^0)$.}
 
{Whitehead's theorem implies the following. Consider the inclusion map $\iota_L\colon L^2 \to \pi_2(M^2,M^1,M^0)$,  with $\d_{\rm MOR} \circ {\iota_L} =\dL$, cf. Rem. \ref{notdelta} and \ref{incP}. If $G'$ is a group,
$(\d' \colon E' \to G',\trr)$ is a crossed module of groups,
 and
 $\phi\colon \pi_1(M^1,M^0) \to G'$ is a groupoid map, then given any
 {set} map $\psi_0\colon L^2 \to E'$ such that $\phi \circ \dL= \d' \circ
 \psi_0$, there exists a unique groupoid map $\psi\colon
 \pi_2(M^2,M^1,M^0) \to E'$ making the diagram below commutative and
 also making the pair $(\psi,\phi)$ a crossed module map $\Pi_2(M^2,M^1,M^0) \to \Gc$ (thus
 $(\psi,\phi)$ is compatible with boundaries and groupoid actions):}
\begin{equation}\label{partUniversal}
 {
\xymatrix{
             &L^2\ar@/^1pc/[rr]^{\psi_0}\ar[rd]_{\d_L}
                                        \ar[r]|>>>>>{\iota_L}&\pi_2(M^2,M^1,M^0)
                                        \ar[d]^{\d_{\rm MOR}}
                                        \ar@{-->}[r]_>>>>>\psi & E'
                                        \ar[d]_{\d'} \\
             & &\pi_1(M^1,M^0)\ar[r]^>>>>>\phi &G'} }.
\end{equation}

\subsubsection{The discrete 2-dimensional (2D) parallel transport of a fake-flat 2-gauge configuration}\label{2dpar}

{As promised in \ref{dptf}, we now state and prove the analogue of Thm. \ref{d1dh} for fake-flat 2-gauge configurations.}

{Let $(M,L)$ be a 2-lattice and $\Gc=(\dG\colon E\to G,\trr)$ be a group crossed module.}
\begin{Theorem}[The {discrete} 2D parallel transport of a fake-flat 2-gauge configuration]\label{d2dh}
{ There exists a one-to-one correspondence} between fake-flat 2-gauge configurations $\F=(\F^2,\F^1)$ in $(M,L)$ and crossed module maps $(\Psi_\F,\Phi_\F)\colon \Pi_2(M^2,M^1,M^0) \to \Gc. $ {(Note $\Psi_\F\colon \pi_2(M^2,M^1,M^0) \to E$ and $\Phi_\F\colon \pi_1(M^1,M^0) \to G$ therefore are groupoid maps, compatible with boundary maps and groupoid actions in the obvious way.)}
\end{Theorem}
\noindent {In analogy with the differential-geometric construction of  2-dimensional parallel transport 2-functors attached to 2-connections on 2-bundles \cite{schreiber_waldorf2,martins_picken}, the crossed module map $(\Psi_\F,\Phi_\F)\colon \Pi_2(M^2,M^1,M^0) \to \Gc$ associated to a fake-flat 2-gauge configurations $\F$ will  be called the discrete 2D parallel transport 2-functor of $\F$.}
\begin{proof} {Recall the notation introduced after Thm \ref{wt}.}
{A fake-flat 2-gauge configuration $\F=(\F^2\colon L^2 \to E\,\, ,\,\,\F^1\colon L^1 \to G)$ is an
assignment $\g\mapsto g_\g$ of an element of $G$ to each 1-cell $\g$ of $L$, and an assigment $P \mapsto e_P$ of an element of $E$ to  each 2-cell $P$, satisfying the
fake-flatness condition of Def. \ref{fakeflatconf}. Whitehead's theorem (Thm \ref{wt}) states that the fundamental crossed
module $\Pi_2(M,M^1,M^0)$, with set $M^0$ of base points -- a crossed
module of groupoids, is {isomorphic to} the free crossed module on the map $\dL\colon L^2 \to \pi_1(M^1,M^0)$.}

Since the groupoid {$\pi_1(M^1,M^0)$} is free on the 1-cells, the assignment
$\F^1\colon \gamma \in L^1 \mapsto g_\gamma \in G$
uniquely extends to a groupoid map
$\Phi_\F=\Phi_{\F^1}\colon {\pi_1(M^1,M^0)} \to G$.  The fake-flatness condition means that the
outer  part of the diagram below commutes:
$${
\xymatrix{
  &L^2\ar@/^1pc/[rr]^{\F^2}\ar[rd]_{\dL} \ar[r]|<<<<{\iota_L}
    &\pi_2(M^2,M^1,M^0)\ar[d]^{\d_{\rm MOR}} \ar@{-->}[r]_>>>>>>{\Psi_\F}
      & E\ar[d]_{\dG}
  \\
     & &\pi_1(M^1,M^0)\ar[r]^>>>>>{\Phi_\F} &G} .}
$$
And by applying the universal property defining free crossed modules
of groupoids, {in the particular form of \eqref{partUniversal}, we can see that a gauge configuration $\F=(\F^2,\F^1)$} can be
extended (uniquely) to a crossed module map
$(\Psi_\F,\Phi_\F)\colon \Pi_2(M,M^1,M^0) \to {\Gc}$, and all crossed module maps    $\Pi_2(M,M^1,M^0) \to {\Gc}$ arise this way.
\end{proof}

{Cf. \ref{dptf}. We have now explained the crossed module analogue of discrete parallel transport functors (Thm \ref{d1dh}), in terms of discrete 2D parallel transport 2-functors.
In the next two subsections \S \ref{atd} and \S \ref{cd}, we address how the 2D parallel transport of a fake-flat 2-gauge configuration can be used to define notions of discrete 2D holonomy along surfaces $\Sigma$ cellularly embedded in a 2-lattice $(M,L)$.   We will only deal with the case when $\Sigma$ is the 2-disk $D^2$ or the 2-sphere $S^2$. In these cases, which are the ones needed to define higher Kitaev models, a 2D holonomy can be associated to cellularly embedded surfaces $\Sigma \subset M$.  2D holonomy along {2-disks} and 2-spheres is particularly simple {to formulate, given that the} corresponding oriented mapping class groups are trivial.}  

{For surfaces $\Sigma$ not homeomorphic to $S^2$ or to $D^2$, additional information is needed to define a  meaningful 2D holonomy for a cellular embedding of $\Sigma$ into $(M,L)$. Namely (assuming orientability) we must choose an isotopy class of homeomorphisms $\Sigma' \to \Sigma$, where $\Sigma'$ is the boundary of an unknotted handlebody in $\mathbb{R}^3$;  see \cite{martins_picken,Zu1}. We will address this more general 2D holonomy in a forthcoming publication.}

\subsection{Algebraic topological definition of 2D holonomy along {2-disks} and 2-spheres}\label{atd}
Let us fix a crossed module of groups $\Gc=(\dG\colon E \to G,\trr)$.
In this subsection we use elementary algebraic topology  to define precisely and concisely the 2-dimensional (2D) holonomy of a fake-flat 2-gauge configuration along cellularly embedded 2-disks and 2-spheres; as such we present the 2D analogue of Rem. \ref{holalongcircle-conceptual}. A combinatorial definition of this 2D holonomy (therefore the analogue of Def. \ref{holalongcircle}) will be dealt with in \S \ref{cd}.
\subsubsection{The {2-disk} case}\label{dc}

Let $(M,L)$ be a 2-lattice. Let $\Sigma$ be a surface embedded in $M$. Suppose {that} $\Sigma$ is homeomorphic to the 2-disk $D^2$. Suppose in addition that $\Sigma$ is oriented.
Furthermore (cf. Def. \ref{relcw}) suppose that $L$ is a 2-lattice decomposition of the triple $(M,\Sigma,\bound (\Sigma))$, where $(\Sigma,\bound (\Sigma))$ is a pair homeomorphic to $(D^2,S^1)$. We have an induced relative CW-decomposition of $(\Sigma,\bound (\Sigma))$. Note $\Sigma=\Sigma^2$ (the 2-skeleton of $\Sigma$) and $\bound (\Sigma) \subset \Sigma^1.$

Choose a base point $v\in \bound(\Sigma)$.
Since {$\pi_2(\Sigma^1,\bound(\Sigma),v)$} is trivial, the homotopy exact sequences of the triple $(\Sigma,\Sigma^1,\bound (\Sigma))$  and of the pair $(\Sigma^1, \bound (\Sigma))$ {imply} that the inclusion $(\Sigma,\bound (\Sigma)) \to (\Sigma,\Sigma^1)$ yields injections $\pi_2(\Sigma,\bound (\Sigma),v) \to  \pi_2(\Sigma,\Sigma^1,v)$ and $\pi_1(\bound (\Sigma),v) \to \pi_1(\Sigma^1,v)$. We can thus see the crossed module $\Pi_2(\Sigma,\bound (\Sigma),v)\cong {(\id\colon \Z \to \Z)}$ (Ex. \ref{pi2d2}) as canonically included in $\Pi_2(\Sigma,\Sigma^1,v)$.

 Let $*=(0,0)$ be the common base point of $D^2$ and $S^1=\bound(D^2)$. By elementary algebraic topology -- since $\pi_2(D^2,S^1,*)=\Z$ -- any two pointed homeomorphisms $(D^2,S^1)\to (\Sigma,\bound (\Sigma))$ preserving orientation are homotopic as maps of pointed pairs $(D^2,S^1)\to (\Sigma,\bound (\Sigma))$. {This is used in the definition below.}

\begin{definition}[Notation:  $\dv(\Sigma,L)$ and $\iv(\Sigma,L)$]\label{iD} Cf Rem. \ref{notdelta} and \ref{incP}.  Let $\Sigma$ be {an} oriented surface homeomorphic to $D^2$.  Let {$v \in \bound(\Sigma)$, the boundary of $\Sigma$. {We will be mainly interested in the case when $v\in \bound(\Sigma) \cap M^0$.}} Consider a pointed orientation preserving homeomorphism {$j\colon (D^2,S^1,*) \to (\Sigma, \bound (\Sigma),v)$}. {We let $j_*\colon  \Pi_2(D^2,S^1,*) \to \Pi_2(\Sigma, \bound (\Sigma),v)$ be given by the induced map on homotopy groups.}
Let $\dv(\Sigma,v)\in \pi_1(\bound (\Sigma),v)\subset \pi_1(\Sigma^1,v)$ be {$\dv(\Sigma,v)=j_*(1),$} where $1$ is the positive (counterclockwise) generator of $\pi_1(S^1,v)\cong \Z$. Hence $\dv(\Sigma,v)$ is a {positively oriented} loop along the
boundary {$\bound(\Sigma)$} of the 2-disk $\Sigma$, starting and ending at $v$. Analogously put   {$\iv(\Sigma,L)=j_*(1) \in \pi_2(\Sigma,\bound (\Sigma),v)\subset \pi_2(\Sigma,\Sigma^1,v)$,} where $1$ is now the positive generator of $\pi_2(D^2,S^1,v)\cong \Z$.
\end{definition}
Note that by construction (cf. Ex \ref{pi2d2}):
\begin{equation}\label{di}
\begin{split}
\d \big (\iv(\Sigma,L)\big)&=\dv(\Sigma,L),\\ \dv(\Sigma,L) \trr  \iv(\Sigma,L)&=\iv(\Sigma,L).
\end{split}
\end{equation}
\begin{lemma}[Dependence of $\dv(\Sigma,L)$ and $\iv(\Sigma,L)$ on $v\in \bound (\Sigma)$]\label{iDv} Suppose that {$v\in M^0\cap \bound(\Sigma)$}.  Choose another {$v'\in M^0\cap \bound(\Sigma)$}. Consider a  path $\gamma$ in {$\bound(\Sigma)$}, connecting $v'$ to $v$. {(There are two different possible homotopy classes $[\g]$ for $\g$.) Then passing to the corresponding element  {$[\gamma]\in \pi_1(\Sigma^1,\Sigma^0)$,} it holds that:}
\begin{equation}\label{bpd}
{
\begin{split}
[\gamma] \trr \iv(\Sigma,L)&=\ivp(\Sigma,L), \textrm{ in } \pi_2(\Sigma,\Sigma^1,\Sigma^0);\\
[\gamma] \, \dv(\Sigma,L)\,[\gamma]^{-1}&=\dvp(\Sigma,L), \textrm{ in } \pi_1(\Sigma^1,\Sigma^0).
\end{split}
}
\end{equation}
\end{lemma}
\begin{proof} {Follows from geometric considerations and the fact that $\pi_1(S^1,*)$ acts trivially on $\pi_2(D^2,S^1,*)$.}  
\end{proof}

Let now $\Gc=(\dG\colon E \to G,\trr)$ be a crossed module of groups.
\begin{definition}[2D holonomy ${\rm Hol}_v(\F,\Sigma,L)$ of a fake-flat 2-gauge configuration $\F=(\F^2,\F^1)$ along $\Sigma\cong D^2$, with initial point $v$]\label{2dholatd} {Let $M$ be a topological manifold.} Let $\Sigma$ be an oriented disk embedded in $M$. Let $L$ be a 2-lattice decomposition of $(M,\Sigma, \bound (\Sigma))$; {see Def. \ref{relcw}}. Let $v \in \bound(\Sigma)\cap M^0$. Let $\F$ be a fake-flat 2-gauge configuration in $(M,L)$, and $\F_\Sigma$ be its  restriction  to the induced 2-lattice decomposition of $\Sigma$. Let $(\Psi_{\F_\Sigma},\Phi_{\F_\Sigma})\colon\Pi_2(\Sigma,\Sigma^1,\Sigma^0) \to \Gc$ be the 2D parallel transport 2-functor of $\F_\Sigma$;  Thm. \ref{d2dh}. We define the 2D holonomy of $\F$ along $\Sigma$, with initial point $v$, as:
$$ {\rm Hol}_v(\F,\Sigma,L)
   =\left ({\rm Hol}_v^2(\F,\Sigma,L),{\rm Hol}_v^1(\F,\Sigma,L) \right)=\Big(\Psi_{\F_\Sigma}(\iv(\Sigma,L)),\Phi_{\F_\Sigma}(\dv(\Sigma,L)) \Big) \in E \times G.
$$
\end{definition}
In the conditions of {Def.} \ref{2dholatd}, note that:
\begin{equation}\label{fake-flatness-preserved}
\dG\left ({\rm Hol}_v^2(\F,\Sigma,L)\right) ={\rm Hol}_v^1(\F,\Sigma,L).
\end{equation}
This is because, by \eqref{di} and the fact that $(\Psi_{\F_\Sigma},\Phi_{\F_\Sigma})\colon\Pi_2(\Sigma,\Sigma^1,\Sigma^0) \to \Gc$ is a crossed module map:
\begin{align*}
\dG\left ({\rm Hol}_v^2(\F,\Sigma,L)\right) &=\dG(\Psi_{\F_\Sigma}(\iv(\Sigma,L))=\Phi_{\F_\Sigma}\big (\d( \iv(\Sigma,L))\big)\\
 &=\Phi_{\F_\Sigma}(\dv(\Sigma,L))={\rm Hol}_v^1(\F,\Sigma,L).
\end{align*}

\begin{remark}[Dependence of {2D holonomy on} base points]
The 2D holonomy ${\rm Hol}_v(\F,\Sigma,L)$ of a {fake-flat} 2-gauge configuration along a cellularly embedded 2-disk depends  on the
choice of {a} base point {$v\in \bound(\Sigma) \cap M^0$.} However, the dependence is mild. Cf. Rem. \ref{iDv}.
Choose any quantised path $v' \ra{\gamma} v $, in the boundary of the disk $\Sigma$, from the new base point $v'$ to the initial base point $v$. Let $[\gamma]$ be the corresponding element of $\pi_1(\Sigma^1,\Sigma^0)\cong\fg{L^0,L^1}$.
Then, since $(\Psi_{\F_\Sigma},\Phi_{\F_\Sigma})\colon \Pi_2(\Sigma,\Sigma^1,\Sigma^0) \to \Gc$ is a crossed module map:
\begin{equation*}
\begin{split}
{\rm {Hol}}_{v'}^2(\F,\Sigma,L)&=\Psi_{\F_\Sigma}\big(\ivp(\Sigma,L)\big)\\
&=\Psi_{\F_\Sigma}\big([\gamma] \trr \iv(\Sigma,L)\big)\\
&=\Phi_{\F_\Sigma}([\gamma])\trr \Psi_{\F_\Sigma} \big(\iv(\Sigma,L)\big).
\end{split}
\end{equation*}
Hence:
\begin{equation}\label{basepoint}
{\rm {Hol}}_{v'}^2(\F,\Sigma,L)=\Phi_{\F_\Sigma}([\gamma])\trr {\rm {Hol}}_{v}^2(\F,\Sigma,L).
\end{equation}
\end{remark}

\subsubsection{The 2-sphere case}\label{2scat}

  We  resume the notation and ideas of \ref{dc}. Let $*=(0,0,0)$ be the base-point of $S^2=\bound (D^3)$. {Let $M$ be a topological manifold. Cf. Def. \ref{relcw},  let} $L$ be a 2-lattice decomposition of $(M,\Sigma)$, where  {$\Sigma\subset M$} is oriented and homeomorphic to the 2-sphere $S^2$. Choose a base point $v \in  \Sigma\cap M^0$. By elementary algebraic topology, any two orientation-preserving homeomorphisms $f\colon S^2 \to \Sigma$ preserving base-points are pointed homotopic.

  Since $\pi_2(\Sigma^1,v)\cong \{0\}$, the final bits of the homotopy  exact sequence of the pointed pair $(\Sigma,\Sigma^1)$, namely $\{0\}\cong\pi_2(\Sigma^1,v) \to \pi_2(\Sigma,v) \ra{i} \pi_2(\Sigma,\Sigma^1,v) \ra{\partial} \pi_1(\Sigma^1,v)$, yield a monomorphism {$i\colon \pi_2(\Sigma,v) \to \pi_2(\Sigma,\Sigma^1,v)$, thus an isomorphism $i\colon \pi_2(\Sigma,v)\to \ker(\partial)$.} Hence $\pi_2(\Sigma,v)$ can be seen as included in {the set of morphisms  of the groupoid} $\pi_2(\Sigma,\Sigma^1,\Sigma^0)$.

\begin{definition}[Notation: $\ivb(\Sigma)$]\label{defivb} Let $S^2$ carry the orientation arising from its embedding into $\mathbb{R}^3$. Let $\Sigma$ be an oriented manifold homeomorphic to $S^2$. Choose a base point $v\in \Sigma$. Choose an orientation preserving homeomorphism $f\colon (S^2,*) \to (\Sigma,v)$. Let $1$ be the  positive generator of $\pi_2(S^2,*)\cong \Z$.   We define $\ivb(\Sigma)\in \pi_2(\Sigma,v)\subset \pi_2(\Sigma,\Sigma^1,\Sigma^0)$,  to be  $\ivb(\Sigma)=f_*(1)$,
where $f_*\colon \pi_2(S^2,*) \to \pi_2(\Sigma,v)$ is the induced map on homotopy. Note that (cf. Def. \ref{de:fgog}),  {it hence follows that $ \d\big(\ivb(\Sigma)\big)=\emptyset_v,$} where $\d$ is the boundary map in the crossed module of groupoids $\Pi_2(\Sigma,\Sigma^1,\Sigma^0)={(\d\colon \pi_2(\Sigma,\Sigma^1,\Sigma^0) \to \pi_1(\Sigma^1,\Sigma^0)})$.
\end{definition}
In what follows, we will frequently not distinguish $\ivb(\Sigma)\in \pi_2(\Sigma,v)$ from $i(\ivb(\Sigma)) \in \pi_2(\Sigma,\Sigma^1,v)$.

\begin{remark}\label{difbase}
 Let $v$ and $v'$ be different points in the 0-skeleton $\Sigma^0=\Sigma\cap M^0$ of $\Sigma$. Let {$[\g]$} be any path in {$\Sigma^1$} connecting $v'$ to $v$, considered up to homotopy relative to the end-points. Then {$ \ivbp(\Sigma)=[\gamma] \trr \ivb(\Sigma)$}.
\end{remark}

\begin{definition}[2D holonomy ${\rm Hol}_v(\F,\Sigma,L)$ of a fake-flat 2-gauge configuration $\F$ along $\Sigma\cong S^2$, with initial point $v$]\label{2scc} Let $\Sigma$ be an oriented manifold homeomorphic to the 2-sphere. Let $L$ be a {2-lattice} decomposition of $(M,\Sigma)$. Let {$v \in \Sigma\cap M^0$}. Let $\F$ be a fake-flat 2-gauge configuration in {$(M,L)$}. Let $\F_\Sigma$ be the restriction of $\F$ to the induced 2-lattice decomposition of  $\Sigma$. Let $(\Psi_{\F_\Sigma},\Phi_{\F_\Sigma})\colon\Pi_2(\Sigma,\Sigma^1,\Sigma^0) \to \Gc$ be the   2D parallel transport 2-functor of $\F_\Sigma$; see Thm. \ref{d2dh}. We define the 2D holonomy of $\F$ along $\Sigma$, {with initial point $v$,} as:
$${\rm Hol}_v^2(\F,\Sigma,L)=\Psi_{\F_\Sigma}\big(i(\ivb(\Sigma,L))\big)\in E.
$$
\end{definition}

\begin{remark}\label{inker} Continuing  Def. \ref{2scc},
 note that since $(\Psi_{\F_\Sigma},\Phi_{\F_\Sigma})\colon\Pi_2(\Sigma,\Sigma^1,\Sigma^0) \to \Gc$ is a crossed module map:
\begin{align*}
\dG \big({\rm Hol}_v^2(\F,\Sigma,L)\big)& =\dG \circ \Psi_{\F_\Sigma} \big (i(\ivb(\Sigma))\big)\\
& =(\Phi_{\F_\Sigma} \circ \d)\big(i( \ivb(\Sigma))\big)\\
&=\Phi_{\F_\Sigma}(\emptyset_v)={1_G}.
\end{align*}
So it always holds that $ {\rm Hol}_v^2(\F,\Sigma,L) \in \ker(\partial_\Gc) \subset E,$ if $\Sigma\cong S^2.$ This is not the case for the 2-disk; cf. \eqref{fake-flatness-preserved}.
\end{remark}

\begin{lemma}[Dependence of 2D holonomy along 2-spheres on base points and orientations]\label{main3} We resume the notation of Def. \ref{2scc}. Let {$v,v' \in \bound(\Sigma)\cap M^0$} be two base points.
Let $\g=t_1^{\theta_1}\dots t_n^{\theta_n}$  be a quantised path in $\Sigma^1$, from $v'$ to $v$; {see Def. \ref{de:qp}}. Recall $g_\g=g_{t_1}^{\theta_1}\dots g_{t_n}^{\theta_n}=\Phi_{\F^1}([\gamma])$; {see} Def. \ref{n1}. We then have:
 $$ {\rm Hol}_{v'}^2(\F,\Sigma,L)=g_{\gamma} \trr {\rm Hol}_{v}^2(\F,\Sigma,L) .$$
Furthermore, if $\Sigma^*$ is $\Sigma$ with the opposite orientation, then:
 $${\rm Hol}_{v}^2(\F,\Sigma^*,L)=\left ({\rm Hol}_{v}^2(\F,\Sigma,L) \right)^{-1}.$$
\end{lemma}
\begin{proof}
Let $(\Psi_{\F_\Sigma},\Phi_{\F_\Sigma})\colon\Pi_2(\Sigma,\Sigma^1,\Sigma^0) \to \Gc=(\dG\colon E \to G,\trr)$ be (Thm. \ref{d2dh}) the crossed module map (i.e. the discrete parallel transport 2-functor) yielded by  the restriction $\F_\Sigma$ of $\F$ to $\Sigma$. Then:
\begin{align*}
 {\rm Hol}_{v'}^2(\F,\Sigma,L)&=\Psi_{\F_\Sigma}(i\big(\ivbp(\Sigma)\big))\\
 &={\Psi_{\F_\Sigma}\big([\gamma]\trr i\big( \ivb(\Sigma)\big)\big), \textrm{ by Rem. \ref{difbase}}}\\
 &=\Phi_{\F_\Sigma}([\gamma]) \trr \Psi_{\F_\Sigma}(i\big( \ivb(\Sigma)\big))\\
 &=g_{\gamma} \trr {\rm Hol}_{v}^2(\F,\Sigma,L).
\end{align*}
Let $\Sigma^*$ be $\Sigma$ with the opposite orientation. Then $i\big(\ivb(\Sigma^*)\big)=\big(i\big( \ivb(\Sigma)\big))^{-1}$. Hence:
$${\rm Hol}_{v}^2(\F,\Sigma^*,L)=\Psi_{\F_\Sigma}(i\big(\ivb(\Sigma^*)\big))=\Psi_{\F_\Sigma}(i\big(\ivb(\Sigma)^{-1}\big))=\left ({\rm Hol}_{v}^2(\F,\Sigma,L) \right)^{-1}.$$
\end{proof}

\subsection{Combinatorial definition of {2D holonomy along 2-disks} and 2-spheres}\label{cd}

We now prepare a combinatorial description of the 2D holonomy of a fake-flat 2-gauge configuration along {cellularly} embedded 2-disks and 2-spheres. Some algebraic topology preliminaries are yet still needed.
\subsubsection{Algebraic topology preliminaries for the {2-disk} case}\label{more-prelim0}
 {Let $\Sigma$ be an oriented manifold homeomorphic to the 2-disk {$D^2=[0,1]^2$.}} Hence we have an orientation of $\bound (\Sigma)\cong S^1$ as well. Let $L=(L^0,L^1,L^2)$ be a 2-lattice decomposition (Def. \ref{relcw}) of $(\Sigma,\bound (\Sigma))\cong (D^2,S^1)$.  Choose  $v \in \bound (\Sigma)$, to be a 0-cell of $L$.   It will look more or less like the pattern in   Fig. \ref{pattern}. {(Here and in other diagrams later, we put oriented circles inside the plaquettes in order to indicate the orientation of their attaching maps; this is redundant as orientations can be inferred from the form of their quantised boundary.)}

\begin{figure}[H]
\centerline{\relabelbox
\epsfysize 5.3cm
\epsfbox{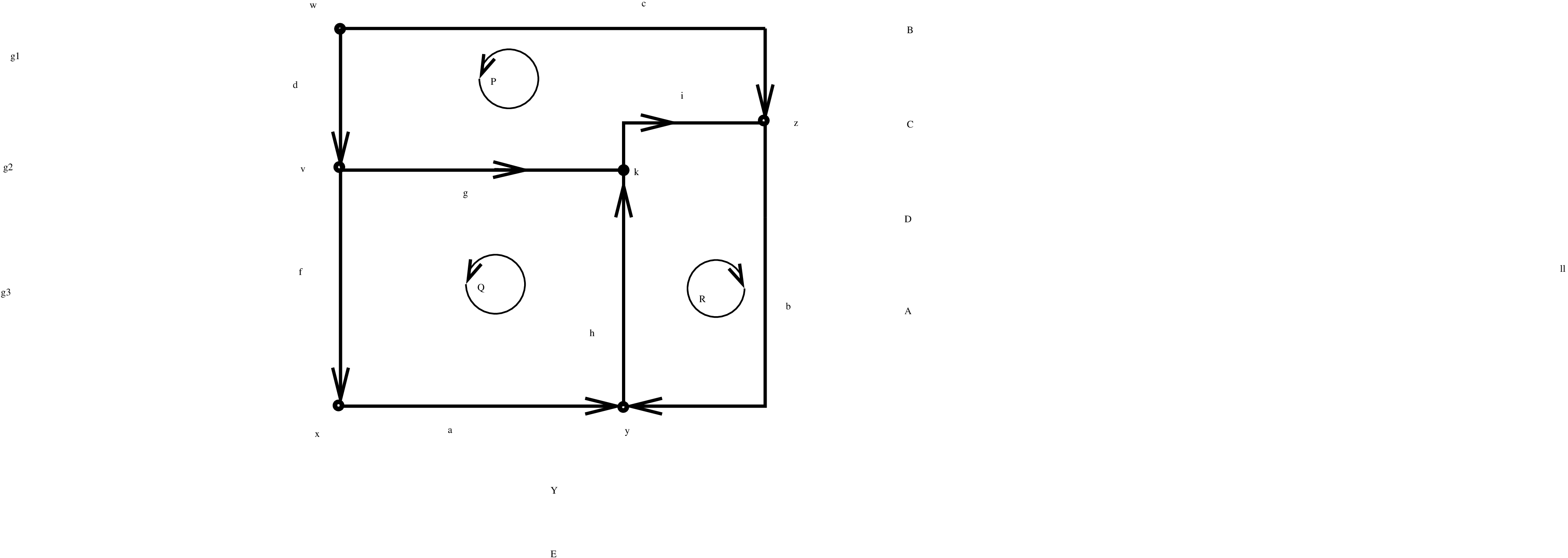}
\relabel{a}{$t_1$}
\relabel{b}{$t_2$}
\relabel{c}{$t_3$}
\relabel{d}{$t_4$}
\relabel{f}{$t_5$}
\relabel{g}{$t_6$}
\relabel{h}{$t_7$}
\relabel{i}{$t_8$}
\relabel{x}{$x$}
\relabel{y}{$x_{P_1}$}
\relabel{z}{$x_{P_3}$}
\relabel{w}{$w$}
\relabel{v}{$v$}
\relabel{k}{$x_{P_2}$}
\relabel{P}{$\small{P_3}$}
\relabel{Q}{$\small{P_1}$}
\relabel{R}{$\small{P_2}$}
\relabel{A}{${\dvQ(\Sigma,L)=t_5 t_1t_2^{-1} t_3^{-1}t_4}$}
\relabel{B}{$\dL^Q(P_3)=t_3^{-1} t_4 t_6 t_8$}
\relabel{C}{$\dL^Q(P_1)=t_7 t_6^{-1}t_5 t_1$}
\relabel{D}{$\dL^Q(P_2)=t_8t_2 t_7$}
\relabel{g1}{$ \gamma_1=t_5 t_1$}
\relabel{g2}{$\gamma_2=t_6$}
\relabel{g3}{$\gamma_3=t_6 t_8 $}
\relabel{E}{${\iv(\Sigma,L)=\big ([\gamma_1] \trr \iL(P_1)\big)\, \big ([\gamma_2] \trr \iL(P_2)^{-1}\big)\, \big ([\gamma_3] \trr \iL(P_3)\big)}$}
\relabel{Y}{${\dvQ(\Sigma,L)\cong\big(\gamma_1 \dL^Q(P_1) \gamma_1^{-1}\big ) \,\,\big(\gamma_2 \dL^Q(P_2)^{-1} \gamma_2^{-1}\big )\,\,\big(\gamma_3 \dL^Q(P_3) \gamma_3^{-1}\big ) }$}
\endrelabelbox}
\caption{A 2-lattice decomposition of $(\Sigma,\bound(\Sigma))\cong(D^2,S^1)$. As shown, the attaching maps of the plaquettes $P_1$ and $P_3$ are oriented counterclockwise, whereas $P_2$ attaches clockwise. The base point of $P_i$ is $x_{P_i}.$ {The quantised boundaries $\dL^Q(P_i)$ of the plaquettes $P_i,i=1,2,3$ are also shown {Def. \ref{qbp}}. The remaining information in the figure will be explained in Def. \ref{pb} and Ex. \ref{patt}.}\label{pattern}}
\end{figure}

Recall that the definition of $\dv(\Sigma,L)$ and $\iv(\Sigma,L)$, which are given in Def. \ref{iD}.
\begin{remark}\label{unicity}
 The homotopy
exact sequence of the pointed pair $(\Sigma,\Sigma^1,v)$ gives
an exact sequence:
\begin{equation}\label{es}
\{0\}\cong\pi_2(\Sigma,v) \to \pi_2(\Sigma,\Sigma^1,v) \ra{\d} \pi_1(\Sigma^1,v) \to
    \pi_1(\Sigma,v)\cong \{1\}.
\end{equation}(see \cite[pp 344]{hatcher}). Therefore $\d\colon \pi_2(\Sigma,\Sigma^1,v) \to \pi_1(\Sigma^1,v)$ is an isomorphism. Hence,
 if $e \in \pi_2(\Sigma,\Sigma^1,v)$:
\begin{align}\label{unicity-form}
 \d(e)=\dv(\Sigma,L) \Longleftrightarrow e=\iv(\Sigma,L).
\end{align}
\end{remark}

\begin{definition}[Quantised boundary $\dvQ(\Sigma,L)$ {of a 2-disk $\Sigma$}]\label{pb}
Choose a base-point $v \in \bound (\Sigma)$, to be a 0-cell. Now go around $\bound (\Sigma)$, following its orientation, starting at $v$ until you go back to $v$. Along the way we pass by the geometric 1-cells $t_1,t_2,\dots,t_n$ of $\bound (\Sigma)$, in that order. Put {$\theta_i=1$} if the characteristic map $\phi_{t_i}^1 \colon [0,1] \to \bound (\Sigma)$ of $t_i$ is oriented positively, and {$\theta_i=-1$} otherwise.
Cf. Fig \ref{pattern}, the quantised boundary $\dvQ(\Sigma,L)$ of $\Sigma$  is
the following quantised path (cf. Defs. \ref{de:fgog}, \ref{de:qp}) in $\bound (\Sigma)$, from $v$ to $v$:
$$
\dvQ(D^2,L)=t_1^{\theta_1} \dots t_n^{\theta_n}.
$$
\end{definition}
\noindent {By passing to morphisms in $\pi_1(\Sigma^1,\Sigma^0)\cong\fg{\Sigma^0,\Sigma^1}$, we hence have $[\dvQ(\Sigma,L)]=\dv(\Sigma,L)$; Def. \ref{de:fgog}, \ref{iD}.}

If we allow the cancellation of consecutive pairs of a 1-cell
and its inverse, thus considering quantised paths up to equivalence {(Def. \ref{de:fgog},)} we can express -- however not uniquely --
$[\dvQ(\Sigma,L)]$ as a product of quantised boundaries (cf. {Def. \ref{qbp}} and {Thm.} \ref{pr:freee}) of plaquettes (or their
inverses), each of which is in addition conjugated by a
(possibly trivial) quantised path in $\Sigma^1$, connecting the base-point $v$ of $\bound (\Sigma)$ with the base-point
of each plaquette.
More precisely:

\begin{lemma}\label{main1} Let $L$ be a relative 2-lattice decomposition of $(\Sigma,\bound (\Sigma))$. {Choose a base point $v \in \bound (\Sigma)$,} to be a 0-cell. {There exists a positive integer $N$, plaquettes $P_i,i =1,\dots, N$
in $L^2$}
(plaquettes can be repeated), integers $\theta_i\in \{\pm 1\}$
(where $i=1,\dots, N$), as well as quantised paths $\gamma_1,\dots,\gamma_N$ in
$\Sigma^1$, connecting $v$ to the base point of each plaquette $P_i$, such that the following equivalence between quantised paths holds:
\begin{equation}\label{main1for}
{\dvQ(\Sigma,L)} \; \cong\; \left (\gamma_1
  \dL^Q(P_1)^{\theta_1}\gamma_1^{-1}\right)\,\,\left(\gamma_2
  \dL^Q(P_2)^{\theta_2}\gamma_2^{-1}\right )  \dots \left(\gamma_N
  \dL^Q(P_N)^{\theta_N}\gamma_N^{-1}\right).
\end{equation}
Here $\cong$ is the equivalence relation on quantised paths in Def. \ref{de:fgog}. Hence we can pass from the left-hand-side of \eqref{main1for} to the right-hand-side by sucessfully inserting or removing pairs $t^{\pm 1} t^{\mp 1}$ where $t$ is a 1-cell of $\Sigma$.
\end{lemma}
\begin{remark}\label{refn}Note that, by Rem. \ref{unicity},  equation \eqref{main1for} holds if, and only if, in $\pi_2(\Sigma,\Sigma^1,\Sigma^0)$ {(cf. Def. \ref{iD}, \ref{incP}):}
$$\iv(\Sigma,L)=([\gamma_1] \trr \iL(P_1))^{\theta_1} ([\gamma_2] \trr \iL(P_2))^{\theta_2} \dots
     ([\gamma_N] \trr \iL(P_N))^{\theta_N}.$$
   {This is because by \eqref{unicity-form}, equation above holds if, and only if, in $\pi_1(\Sigma^1,\Sigma^0)$ we have}:
  \begin{multline*}
  \d \big([\gamma_1] \trr \iL(P_1)^{\theta_1}\,\,[\gamma_2] \trr \iL(P_2)^{\theta_2}\,\, \dots\,\, {[\gamma_N] \trr \iL(P_N)^{\theta_N} }\big) \\=\big([\gamma_1] \dL(P_1)^{\theta_1} [\gamma_1^{-1}]\big ) \,\,\big([\gamma_2] \dL(P_2)^{\theta_2} [\gamma_2^{-1}]\big )\,\,\dots \big([\gamma_N] \dL(P_N)^{\theta_N} [\gamma_N^{-1}]\big )=\dv(\Sigma,L)={\d \big ( \iv(\Sigma,L)\big).}
\end{multline*}
\end{remark}

\begin{example}\label{patt}
{In Fig. \ref{pattern}, we can put $N=3$, $\gamma_1=t_5 t_1$, $\gamma_2=t_6$ and $\gamma_3=t_6 t_8$. Also put $\theta_1=1$, $\theta_2=-1$ and $\theta_3=1$. Then, as indicated in Fig. \ref{pattern}:}
\begin{equation}\label{qw}{\dvQ(\Sigma,L)\cong \big(\gamma_1 \dL^Q(P_1)^{\theta_1} \gamma_1^{-1}\big ) \,\,\big(\gamma_2 \dL^Q(P_2)^{\theta_2} \gamma_2^{-1}\big )\,\,\big(\gamma_3 \dL^Q(P_3)^{\theta_3} \gamma_3^{-1}\big ).}\end{equation}
{This follows from a simple calculation, which we recommend the reader to do.} {From \eqref{qw} it follows that:}
$${\iv(\Sigma,L)=\big([\gamma_1] \trr \iL(P_1)^{\theta_1} \big ) \,\,\big(\gamma_2 \trr \iL(P_2)^{\theta_2} \gamma_2^{-1}\big )\,\,\big(\gamma_3 \iL(P_3)^{\theta_3}\big ).}$$
\end{example}
\begin{proof}{\bf (Lemma \ref{main1})}
Consider the fundamental crossed modules $\Pi_2(\Sigma,\Sigma^1,v)\subset \Pi_2(\Sigma,\Sigma^1,\Sigma^0)$, and the elements $\iv(\Sigma,L)$ and $\dv(\Sigma,L)$ of $\Pi_2(\Sigma,\Sigma^1,\Sigma^0)$; see Rem. \ref{iD}. Recall $\d\big(\iv(\Sigma,L)\big)=\dv(\Sigma,L)$ in $\pi_1(\Sigma^1,\Sigma^0)$.

{Cf. Rem. \ref{notdelta} and Thm. \ref{wt}. We know that $\Pi_2(\Sigma,\Sigma^1,\Sigma^0)$ is a free-crossed module on the map {$\dL\colon L^2 \to \pi_1(\Sigma^1,\Sigma^0)$, where $\pi_1(\Sigma^1,\Sigma^0)$ is the free groupoid on the 1-cells; see Thm. \ref{pr:freee}}. Let us  apply Rem. \ref{genpi2}. Hence there exists a positive integer $N$, such that we can  choose plaquettes $P_{i},i=1,2,\dots, N$, quantised paths  $\gamma_i,i=1,2,\dots, N$,  from the base-point
$v$ of $\Sigma$ to the base point $x_{P_i}$ of $P_i,$ and integers {$\theta_i\in \{\pm 1\},i=1,2,\dots, N$,} such that, in $\pi_2(\Sigma,\Sigma^1,\Sigma^0)$:}
$$
\iv(\Sigma,L)=([\gamma_1] \trr \iL(P_1))^{\theta_1} ([\gamma_2] \trr \iL(P_2))^{\theta_2}
    \dots ([\gamma_N] \trr \iL(P_N))^{\theta_N}.
$$
By using the first Peiffer Law in Def. \ref{cm}, and Rem. \ref{notdelta} and \ref{incP}, it follows that in $\pi_1(\Sigma^1,\Sigma^0)\cong \fg{L^0,L^1}$, and where $[\,\, ]$ means equivalence class of quantised paths (Def. \ref{de:fgog}):
\begin{align*}
[\dvQ(\Sigma,L)]&=\dv(\Sigma,L)\\
&=\d\big(\iv(\Sigma,L)\big)\\
&=\d \big ( ([\gamma_1] \trr \iL(P_1))^{\theta_1}\,\, ([\gamma_2] \trr \iL(P_2))^{\theta_2}
    \,\, \dots\,\, ([\gamma_N]\trr \iL(P_N))^{\theta_N}\big)\\
&=[\gamma_1] \d(\iL(P_1))^{\theta_1}[\gamma_1^{-1}]\,\,
 [\gamma_2] \d(\iL(P_2))^{\theta_2}[\gamma_2^{-1}]
 \dots [\gamma_N] \d(\iL(P_N))^{\theta_N}[\gamma_N^{-1}]\\&=
 [\gamma_1] \dL(P_1)^{\theta_1}[\gamma_1^{-1}]\,\,
 [\gamma_2] \dL(P_2)^{\theta_2}[\gamma_2^{-1}]\,\,
 \dots\,\, [\gamma_N] \dL(P_N)^{\theta_N}[\gamma_N^{-1}]\\
 &=\big[
 \gamma_1 \dL^Q(P_1)^{\theta_1}\gamma_1^{-1}\,\,
 \gamma_2 \dL^Q(P_2)^{\theta_2}\gamma_2^{-1}
 \dots \gamma_N \dL^Q(P_N)^{\theta_N}\gamma_N^{-1}].
\end{align*}
Hence  we can go from  $\dvQ(\Sigma,L)$ to the quantised path
 $\gamma_1 \dL^Q(P_1)^{\theta_1}\gamma_1^{-1}\,\,
 \gamma_2 \dL^Q(P_2)^{\theta_2}\gamma_2^{-1}
 \dots \gamma_N \dL^Q(P_N)^{\theta_N}\gamma_N^{-1}$ in a finite number of steps by inserting, or removing pairs  $t^{\pm 1} t^{\mp 1}$, where $t$ is any 1-cell of $\Sigma$.
\end{proof}

\begin{remark}\label{import-remark}
The choice of a positive integer $N$ and of an assignment:
 \begin{equation}\label{assignment}
\begin{split}
 i &\mapsto P_i, \textrm{ a plaquete},\\
  i &\mapsto \gamma_i, \textrm{ a quantised path from } v \textrm{ to the base point } x_{P_i} \textrm{ of the plaquette } P_i, \\
  i &\mapsto \theta_i\in \{\pm 1\},
  \end{split}
\end{equation}
{where $  i \in \{1, \dots , N\},$} such that we have an equivalence of quantised paths: \begin{equation} \label{refernow}
 {\dvQ(\Sigma,L)}\cong\left (\gamma_1
 \dL^Q(P_1)^{\theta_1}\gamma_1^{-1}\right)\,\,\left(\gamma_2
 \dL^Q(P_2)^{\theta_2}\gamma_2^{-1}\right )  \dots \left(\gamma_N
 \dL^Q(P_N)^{\theta_N}\gamma_N^{-1}\right)
 \end{equation}
-- equivalently (cf. Rem. \ref{unicity}) such that, {in $\pi_2(\Sigma,\Sigma^1,v)$}:
\begin{equation}\label{refernow2}
\iv(\Sigma,L)=([\gamma_1] \trr \iL(P_1))^{\theta_1} ([\gamma_2] \trr \iL(P_2))^{\theta_2} \dots
     ([\gamma_N] \trr \iL(P_N))^{\theta_N}
\end{equation}
or {equivalently} such that, {in $\pi_1(\Sigma^1,v)$:}
\begin{equation}
\dv(\Sigma,L)=\big([\gamma_1] \dL(P_1)^{\theta_1} [\gamma_1^{-1}]\big ) \,\,\big([\gamma_2] \dL(P_2)^{\theta_2} [\gamma_2^{-1}]\big )\,\,\dots \big([\gamma_N] \dL(P_N)^{\theta_N} [\gamma_N^{-1}]\big ),
\end{equation}
-- is far from being unique.
\end{remark}

\subsubsection{A combinatorial description  of the  2D holonomy  along embebbed {2-disks}}\label{dc-0}
Cf. \cite{PorterTuraev}. Let  $(M,L)$ be a 2-lattice. Suppose that $\Sigma$ is homeomorphic to the 2-disk $D^2$ and that $L$ is a decomposition of the triple $(M,\Sigma,\bound (\Sigma))$; {Def. \ref{relcw}}. Fix an orientation on $\Sigma$.  Choose a base-point $v\in \bound (\Sigma)\cap M^0$. Consider a fake-flat 2-gauge configuration $\F=(\F^2,\F^1)$ in $(M,L)$.
Let $L_\Sigma$ be the induced 2-lattice decomposition of $(\Sigma,\bound (\Sigma))\cong (D^2,S^1)$. Let $\dvQ(\Sigma,L)$ be the quantised boundary of $\Sigma$; {Def. \ref{pb}.}

\medskip

\begin{minipage}{0.93\textwidth}{\it Cf. \ref{more-prelim0}, choose a positive integer $N$ and
plaquettes $P_i \in {L_\Sigma^2},i =1,\dots, N$ (plaquettes might be repeated), integers $\theta_i\in \{\pm 1\}$ (where $i=1,\dots ,N$), and quantised paths $\gamma_1,\dots,\gamma_N$ in $\Sigma^1$, from $v$ to the base point $x_{P_i}$ of  $P_i$, such that we have an equivalence of quantised paths {(cf. Def. \ref{de:fgog}, \ref{de:qp}, \ref{qbp})}:}
\begin{equation}\label{to-reduce}
\dvQ(\Sigma,L)\cong \left (\gamma_1
\dL^Q(P_1)^{\theta_1}\gamma_1^{-1}\right)\,\,\left(\gamma_2
\dL^Q(P_2)^{\theta_2}\gamma_2^{-1}\right )  \dots \left(\gamma_N
\dL^Q(P_N)^{\theta_N}\gamma_N^{-1}\right).
\end{equation}
\end{minipage}
 \begin{minipage}{0.095\textwidth}\vskip-0cm \quad (*) \end{minipage}

\medskip
\noindent{Recall that by Rem. \ref{unicity} and \ref{import-remark},  {equation \eqref{to-reduce}} is the same as saying that in $\pi_2(\Sigma,\Sigma^1,\Sigma^0)$ (cf. Rem. \ref{iD}):}
$$\iL(D^2,L)=([\gamma_1] \trr \iL(P_1))^{\theta_1} ([\gamma_2] \trr \iL(P_2))^{\theta_2} \dots
     ([\gamma_N] \trr \iL(P_N))^{\theta_N}.$$

Fix a crossed module $\Gc  = (\dG:E\rightarrow G , \trr)$.
Recall the construction of the 2D holonomy ${\rm Hol}_v(\F,\Sigma,L)$ of a fake-flat 2-gauge configuration (cf. Def. \ref{fakeflatconf}) along $\Sigma \cong D^2$, with initial point $v$; Def. \ref{2dholatd}.
\begin{Theorem}\label{2dholD2}
Suppose that {$L$ is a 2-lattice decomposition} of $(M,\Sigma,\bound (\Sigma))$, {where $\Sigma \cong D^2$ is a surface cellularly embebbed in $M$}. {Let $v \in \bound{(\Sigma)} \cap M^0$.} Let $i\in\{1,\dots,N\}\mapsto (P_i,\theta_i,\gamma_i)$ be as in $(*)$. If 
{$$\F=({\F^1}\colon t \in  L^1 \to g_t \in G\,\,,\,\,{\F^2}\colon P \in L^2 \mapsto e_P \in E) $$}
is a fake-flat 2-gauge configuration {on $(M,L)$, then}  ${\rm Hol}_v(\F,\Sigma,L)\in E\times G$ can be calculated as:
\begin{equation}\label{combfor}
{\rm Hol}_v(\F,\Sigma,L) = \left ({\rm Hol}_v^2(\F,\Sigma,L),{\rm Hol}_v^1(\F,\Sigma,L) \right)=
\left ( \htwo  , \;  g_{\dvQ(\Sigma,L)} \right).
\end{equation}
Cf. Def. \ref{n1}, here $g_{\gamma_i}$ is the product of the elements of $G$ assigned to
the 1-cells of the quantised path $\gamma_i$ (or their inverses), and the same for $g_{\dvQ(\Sigma,L)}$.  In other words  $g_{\gamma_i}=\Phi_{\F^1}([\gamma_i])$ and {$g_{\dvQ(\Sigma,L)}=\Phi_{\F^1}([\dv^Q(\Sigma,L)])$.}
\end{Theorem}
{As an immediate consequence, we have the promised independence theorem of the 2D holonomy of a fake-flat 2-gauge configuration along a pointed 2-disk on the way we combine the group elements associated to the edges and plaquetes, as long as the rules of the assigment (*) are followed.}
\begin{Theorem}\label{independence1}
{Fix $v\in \bound(\Sigma)\cap M^0$. The evaluation ${\rm Hol}_v(\F,\Sigma,L)$ in \eqref{combfor} does not depend on the  assignment $i\mapsto (P_i,\theta_i,\gamma_i)$ as in (*) chosen; see Rem. \ref{import-remark}. Moreover  \eqref{fake-flatness-preserved} holds, i.e. $\dG\left ({\rm Hol}_v^2(\F,\Sigma,L)\right) ={\rm Hol}_v^1(\F,\Sigma,L).$}

\end{Theorem}
\begin{example} Let us consider a fake-flat {2-gauge} configuration on the 2-lattice decomposition of $D^2$ in Fig. \ref{pattern}. In the figure below, we put $g_i=g_{t_i}=\F^1(t_i)\in G$ and $e_i=e_{P_i}=\F^2(P_i) \in E$.  For $(*)$ to hold, we can put (see Ex. \ref{patt}): $\gamma_1=t_5 t_1$, $\gamma_2=t_6$ and $\gamma_3=t_6 t_8$; and $\theta_1=1$, $\theta_2=-1$ and $\theta_3=1$. Hence:
$${\rm Hol}_v(\F,D^2,L)=\left ({\rm Hol}_v^2(\F,D^2,L),{\rm Hol}_v^1(\F,D^2,L) \right)={\Big ( g_{\g_1} \trr e_1 \,\,\, g_{\g_2} \trr e_2^{-1}\,\,\, g_{\g_3} \trr e_3\,\,,\,\, g_{5} g_1 g_2^{-1} g_3^{-1} g_4 \Big).} $$
\begin{figure}[H]
\centerline{\relabelbox
\epsfysize 4.5cm
\epsfbox{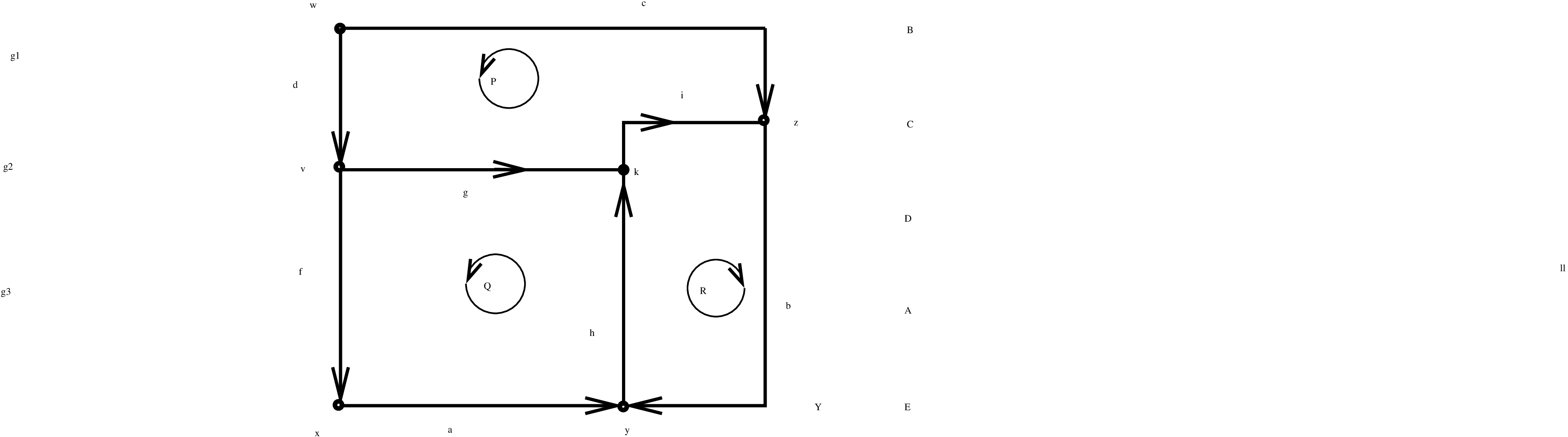}
\relabel{a}{$g_1$}
\relabel{b}{$g_2$}
\relabel{c}{$g_3$}
\relabel{d}{$g_4$}
\relabel{f}{$g_5$}
\relabel{g}{$g_6$}
\relabel{h}{$g_7$}
\relabel{i}{$g_8$}
\relabel{y}{$x_{P_1}$}
\relabel{z}{$x_{P_3}$}
\relabel{v}{$v$}
\relabel{k}{$x_{P_2}$}
\relabel{P}{${e_3}$}
\relabel{Q}{${e_1}$}
\relabel{R}{${e_2}$}
\relabel{B}{$\d(e_3)=g_3^{-1} g_4 g_6 g_8$}
\relabel{C}{$\d(e_1)=g_7 g_6^{-1}g_5 g_1$}
\relabel{D}{$\d(e_2)=g_8g_2 g_7$}
\relabel{E}{${{\rm Hol}^2_v(\F, D^2,L)=g_{\g_1} \trr e_1 \,\,\, g_{\g_2} \trr e_2^{-1}\,\,\, g_{\g_3} \trr e_3}$}
\relabel{A}{${\rm Hol}^1_v(\F, D^2,L)=g_{5} g_1 g_2^{-1} g_3^{-1} g_4$}
\relabel{g1}{$ g_{\gamma_1}=g_5 g_1$}
\relabel{g2}{$g_{\gamma_2}=g_6$}
\relabel{g3}{$g_{\gamma_3}=g_6 g_8 $}
\endrelabelbox}
\caption{ A\,\, fake-flat configuration $\F$ on $(D^2,L)$ and its 2-dimensional holonomy. \label{pattern2}}
\end{figure}
\end{example}

\begin{proof}{ \bf (Theorem \ref{2dholD2})} We use  Rem. \ref{refn}, \ref{unicity}. Given an assignment {$i\in \{1,\dots,N\} \mapsto (P_i,\gamma_i,\theta_i)$ as in $(*)$, then}
\begin{align*}
&[\gamma_1] \trr \iL(P_1)^{\theta_1}\,\,[\gamma_2]\trr \iL(P_2)^{\theta_2} \dots {[\gamma_N]\trr  \iL(P_N)^{\theta_N}}=\iv(\Sigma,L),\end{align*}
{where this holds in  $\pi_2(\Sigma,\Sigma^1,\Sigma^0)$, and, now in $\pi_1(\Sigma^1,\Sigma^0)$:}
\begin{align*} \d \big([\gamma_1] &\trr \iL(P_1)^{\theta_1}\,\,[\gamma_2]\trr \iL(P_2)^{\theta_2} \dots {[\gamma_N]\trr  \iL(P_N)^{\theta_N}\big)}=
  \\&=\big([\gamma_1] \dL(P_1)^{\theta_1} [\gamma_1^{-1}]\big ) \,\,\big([\gamma_2] \dL(P_2)^{\theta_2} [\gamma_2^{-1}]\big )\,\,\dots \big([\gamma_N] \dL(P_N)^{\theta_N} [\gamma_N^{-1}]\big )=\dv(\Sigma,L).
\end{align*}
{The restriction} $\F_\Sigma$ of  $\F$ to $\Sigma$ gives a crossed module map $(\Psi_{\F_\Sigma},\Phi_{\F_\Sigma})\colon\Pi_2(\Sigma,\Sigma^1,\Sigma^0) \to \Gc$. %; cf. Thm \ref{d2dh}. 
 Thus:
\begin{align*}
{\rm Hol}_v^2(\F,\Sigma,L)&\doteq\Psi_{\F_\Sigma}(\iv(\Sigma,L))\\
&= \Psi_{\F_\Sigma}\Big( [\gamma_1] \trr \iL(P_1)^{\theta_1}\,\,[\gamma_2]\trr \iL(P_2)^{\theta_2} \dots {[\gamma_N]\trr  \iL(P_N)^{\theta_N} }\Big)\\
                                         &=   \Phi_{\F_\Sigma}( [\gamma_1] )\trr  \Psi_{\F_\Sigma}(\iL(P_1)^{\theta_1})\,\,\Phi_{\F_\Sigma}([\gamma_2])\trr \Psi_{\F_\Sigma}(\iL(P_2)^{\theta_2}) \dots {\Phi_{\F_\Sigma}([\gamma_N])\trr \Psi_{\F_\Sigma}( \iL(P_N)^{\theta_N}} )\\
 &\doteq g_{\gamma_1} \trr
e_{P_1}^{\theta_1}\,\,g_{\gamma_2} \trr  e_{P_2}^{\theta_2}\,\,
\dots\,\, {g_{\gamma_N} \trr  e_{P_N}^{\theta_N}},
\end{align*}
where, N.B., {firstly $\trr$ is in $\Pi_2(\Sigma,\Sigma^1,\Sigma^0)$ and then it is in  $\Gc$.} Analogously:
\begin{align*}
 {\rm Hol}_v^1(\F,\Sigma,L) &=\Psi_{\F_\Sigma}(\dv(\Sigma,L))=g_{\dvQ(\Sigma,L)}.
\end{align*}
\end{proof}

\subsubsection{Algebraic topology preliminaries for the 2-sphere case}\label{more-prelim}
 Let $(\Sigma,L)$ be a 2-lattice. Suppose that $\Sigma$ is {an oriented surface homeomorphic to the 2-sphere.} Choose a base-point $v\in \Sigma$, to be a 0-cell. Recall that the homotopy  exact sequence of $(\Sigma,\Sigma^1)$ yields:
\begin{equation}\label{ess2}
\pi_2(\Sigma^1,v)\cong \{0\} \to \pi_2(\Sigma,v)\cong \Z \ra{i} \pi_2(\Sigma,\Sigma^1,v) \ra{\d} \pi_1(\Sigma^1,v) \to  \pi_1(\Sigma,v)\cong \{1\}
\end{equation}
(which is exact), and {hence  we have an isomorphism $i\colon \pi_2(\Sigma,v)\cong \Z \to \ker(\d) \subset \pi_2(\Sigma,\Sigma^1,v)$.}

{Cf. Def.  \ref{defivb}, \ref{2scc}, Rem. \ref{incP}} and the construction in \ref{dc-0}. In order to find  a combinatorial definition of ${\rm Hol}_v^2(\F,\Sigma,L)$, we must express $i(\ivb(\Sigma))\in  \pi_2(\Sigma,\Sigma^1,v)$ (see Def. \ref{defivb}) as a product of terms like $[\gamma]\trr \iL(P)$, where $\gamma$ is a quantised path from $v$ to the base point $x_P$ of the plaquette $P$, and $[\g]$ is the element in $\pi_1(\Sigma^1,\Sigma^0)$ it yields; Def. \ref{de:qp} and \ref{de:fgog}.
A crucial point in \ref{more-prelim0} is  that the kernel of the boundary map $\d\colon \pi_2(N,N^1,v) \to \pi_1(N^1,v)$ is trivial if $N\cong D^2$, whereas if $\Sigma \cong S^2$ we have $\ker(\d)=\pi_2(\Sigma,v)\cong \mathbb{Z}$; see  \eqref{ess2}.
{In order to identify $i(\ivb(\Sigma))\in  \pi_2(\Sigma,\Sigma^1,v)$,} we will need to use the Hurewicz map between homotopy and homology long exact sequences of pairs; see \cite[pp 374]{hatcher}.  Our main tool is the fact that the Hurewicz map $h\colon \pi_2(\Sigma,v) \to H_2(\Sigma)$ is an isomorphism, {if $\Sigma\cong S^2$.}

Before continuing, let us define, given a plaquette $P \in L^2$, with characteristic map  {$\phi_P^2\colon D^2 \to \overline{c^2_P} \subset \Sigma$}: 
\begin{equation}\label{defP}
{{\rm sgn}(P)=\begin{cases} 1, &\textrm{ if  the restriction of  } \phi_P^2\colon D^2 \to \Sigma \textrm{ to } {{\rm int}(D^2)} \textrm{ is orientation preserving};\\ -1,  &\textrm{ if  the restriction of  } \phi_P^2\colon D^2 \to {\Sigma \textrm{ to } {\rm int}(D^2)} \textrm{ is orientation reversing}.\end{cases}
}\end{equation}
\begin{lemma}\label{mainS2}
There exist a positive integer $N$, and an assignment $i\in \{1,\dots,N\} \to (P_i,\gamma_i,\theta_i)$, where $P_i\in L^2$, $\gamma_i$ is a quantised path from $v$ to the base point $x_{P_i}$ of the plaquette $P_i$, and $\theta_i \in  \{\pm 1\}$, such that:
\begin{enumerate}
\item we have an equivalence of quantised paths (Def. \ref{de:qp} and \ref{de:fgog}):
$$\emptyset_v \cong\; \left (\gamma_1
  \dL^Q(P_1)^{\theta_1}\gamma_1^{-1}\right)\,\,\left(\gamma_2
  \dL^Q(P_2)^{\theta_2}\gamma_2^{-1}\right )  \dots \left(\gamma_N
  \dL^Q(P_N)^{\theta_N}\gamma_N^{-1}\right),$$
\item given any $P \in L^2$, $\displaystyle\sum_{i \in \{1,\dots, N\} \textrm{ such that } P_i=P} \theta_i={\rm sgn} (P)$.
\end{enumerate}

Moreover, given a positive integer $N$ and an assignment $i\in \{1,\dots,N\} \to (P_i,\gamma_i,\theta_i)$, where $P_i\in L^2$, $\gamma_i$ is a quantised path from $v$ to the base point $x_{P_i}$ of the plaquette $P_i$, and $\theta_i \in  \{\pm 1\}$,
then:
$$[\gamma_1] \trr \iL(P_1)^{\theta_1}\,\,[\gamma_2]\trr \iL(P_2)^{\theta_2} \dots {[\gamma_N]\trr  \iL(P_N)^{\theta_N}}=i(\ivb(\Sigma)),$$
{(cf. Def. \ref{defivb}  and Rem. \ref{incP}) happens if, and only if,} conditions 1. and 2. of the lemma are satisfied.
\end{lemma}

\begin{proof}
Cf. Rem. \ref{notdelta}, Thm. \ref{wt} and Rem. \ref{genpi2}.
Since $\Pi_2(\Sigma,\Sigma^1,\Sigma^0)$ is the free crossed module on  $\dL\colon L^2 \to \pi_1(\Sigma^1,\Sigma^0)$, there exist a positive integer $N$, and an assignment $i\in \{1,\dots,N\} \mapsto (P_i,\gamma_i,\theta_i)$, where $P_i\in L^2$, $\gamma_i$ is a quantised path from $v$ to the base point $x_{P_i}$ of the plaquette $P_i$, and $\theta_i \in  \{\pm 1\}$, such that:
$$([\gamma_1] \trr \iL(P_1))^{\theta_1} ([\gamma_2] \trr \iL(P_2))^{\theta_2}
    \dots ([\gamma_N] \trr \iL(P_N))^{\theta_N}=i\big(\ivb(\Sigma,L)\big).
$$
We claim that $i\in \{1,\dots,N\} \mapsto (P_i,\gamma_i,\theta_i)$ satisfies items 1 and 2, {of the statement of the lemma}.

\noindent {\bf Item 1.} Since $\d( i\big(\ivb(\Sigma,L)\big))=\emptyset_v$, combining with:
\begin{multline*}\d \big(([\gamma_1] \trr \iL(P_1))^{\theta_1} ([\gamma_2] \trr \iL(P_2))^{\theta_2}
    \dots ([\gamma_N] \trr \iL(P_N))^{\theta_N}\big)\\=[\gamma_1] \dL(P_1)^{\theta_1} [\gamma_1]^{-1}\,\, [\gamma_2] \dL(P_2))^{\theta_2}[\gamma_2]^{-1} \dots
[\gamma_N]  \dL(P_N))^{\theta_N}[\gamma_N]^{-1},
\end{multline*}
yields that   $i\in \{1,\dots,N\} \mapsto (P_i,\gamma_i,\theta_i)$ satisfies item 1.

\noindent {\bf Item 2.}
Consider the map of exact sequences obtained from the Hurewicz map between homotopy and homology
long exact sequences, \cite[pp 374]{hatcher}:
\begin{equation}\label{mapexact}
\xymatrix{ &\pi_2(\Sigma^1,v) \cong \{0\}\ar[r]\ar[d]^\cong & \mathbb{Z} \cong \pi_2{(\Sigma,v)}\ar[d]^h_\cong\ar@{^{(}->}[r]^i &\pi_2(\Sigma,\Sigma^1,v)\ar[d]^{h_r} \ar[r]^{\d} &\pi_1(\Sigma^1,v)\ar[d]^h \ar[r]^p& \pi_1(\Sigma,v)\cong \{1\}\ar[d]^\cong\\
              & H_2(\Sigma^1) \cong \{0\}\ar[r] &\mathbb{Z} \cong H_2{(\Sigma)}\ar@{^{(}->}[r]^i &H_2(\Sigma,\Sigma^1) \ar[r]^{\d} &H_1(\Sigma^1) \ar[r]^p& H_1(\Sigma)\cong \{0\}
}.\end{equation}
 The group $H_2(\Sigma,\Sigma^1)$ is the free abelian group on the relative homology classes $a(P)=h_r(\iL(P))$ determined by the plaquettes $P\in L^2$; \cite[pp 137]{hatcher}. Moreover $h_r(\gamma \trr \iL(P))=a(P)$, for
each plaquette $P$ and each path $\gamma\in \pi_1(\Sigma^1,\Sigma^0)$ connecting $v$ to the
base-point of $P$. We let {$K=h(\ivb(\Sigma))\in H_2(\Sigma)$.} Then $K$ is the positive generator of $H_2(\Sigma)\cong \Z$. We now need the following claim:

\noindent{\bf Claim} $\displaystyle i(K)=\sum_{P \in L^2} {\rm sgn}(P)a(P)\in H_2(\Sigma,\Sigma^1)$.\\
\noindent {\bf Proof of the claim (sketch)} This is seemingly well know, however we could not find a proof anywhere. Since  $H_2(\Sigma,\Sigma^1)$ is the free abelian group on the  $a(P)=h_r(\iL(P))$, we know that there exist unique $\lambda_P\in \Z$, where $P \in L^2$, such that
$i(K)=\sum_{P \in L^2} \lambda_P\,\, a(P)$. We need to prove that $\lambda_P={\rm sgn}(P)$, for each $P \in L^2$.

 Let $P \in L^2$. Let $x$ be an interior point of open cell $c_P^2$ corresponding to $P$. We have a commutative diagram {\eqref{abx}}, where all morphisms are induced by inclusion. {The vertical line $p_x$} corresponds to the identity map $\Z \to \Z$, in the sense that it sends the positive generator $K\in H_2(\Sigma)\cong \Z$ to the positive generator $K_x$ of $H_2(\Sigma,\Sigma\setminus\{x\})\cong \Z$. (Note $\Sigma$ is oriented, so it makes sense to speak about those positive generators.)
\begin{equation}\label{abx}\xymatrix{&\Z\cong H_2(\Sigma)= H_2(\Sigma,\emptyset) \ar@{^{(}->}[r]^>>>>>>i\ar[d]^{p_x} & H_2(\Sigma,\Sigma^1)\ar[dl]^{p'_x}\\ & \Z \cong H_2(\Sigma,\Sigma\setminus\{x\}) }.
\end{equation}
Then $p'_x(a(P))={\rm sgn}(P)K_x$, by definition of ${\rm sgn}(P)$. Whereas if $Q\in L^2$ is another plaquette, then since the corresponding {closed 2-cell $\overline{e^2_Q}$} is contained in $\Sigma\setminus \{x\}$ it holds $p'_x(a(Q))=0$. Hence ${\rm sgn}(P)=\lambda_P$. {\bf QED.}
\medskip

Having proven the claim, {item 2 of the statement of the lemma} now follows from the fact that:
\begin{align*}
 \smash{\sum_{P \in L^2} {\rm sgn}(P)a(P)}={i(K)}&=i\circ h \big (\ivb(\Sigma)\big)\\
                                  &=h_r\circ i(\ivb(\Sigma))\\
                                  &=h_r\left(([\gamma_1] \trr \iL(P_1))^{\theta_1} ([\gamma_2] \trr \iL(P_2))^{\theta_2}
    \dots ([\gamma_N] \trr \iL(P_N))^{\theta_N}\right)\\
    &=\sum_{i=1}^N \theta_i a(P_i).
\end{align*}
We note that  $H_2(\Sigma,\Sigma^1)$ is the free abelian group on the  $a(P)$, where $P \in L^2.$

To finalise, let us be given an assignment $i\in \{1,\dots,N\} \mapsto (P_i,\gamma_i,\theta_i)$, where $P_i\in L^2$, $\gamma_i$ is a quantised path from $v$ to  $x_{P_i}$, and $\theta_i \in  \{\pm 1\}$.
Let $A=[\gamma_1] \trr \iL(P_1)^{\theta_1}\,\,[\gamma_2]\trr \iL(P_2)^{\theta_2} \dots [\gamma_m]\trr  \iL(P_m)^{\theta_m}\in \pi_2(\Sigma,\Sigma^1,v){\in \pi_2(\Sigma,\Sigma^1,\Sigma^0)}$. From \eqref{mapexact} it follows:
\begin{align*}
A=i(\ivb(\Sigma))&\Leftrightarrow \d(A)=\emptyset_v \textrm{ and } h_r(A)=(i\circ h)(\ivb(\Sigma)) \\
            & \Leftrightarrow\textrm{conditions of item 1 and item 2 each are satisfied}.
\end{align*}
\end{proof}

\subsubsection{A combinatorial description  of the  2D holonomy  along embedded 2-spheres}\label{sc}

Let $\Sigma$ be an oriented {$S^2$} embedded in a manifold $M$. {Let $L$ be a 2-lattice decomposition of $(M,\Sigma)$. {Let $v\in \Sigma\cap M^0$.} Let
$\F=({\F^1}\colon t \in  L^1 \to g_t \in G\,\,,\,\,{\F^2}\colon P \in L^2 \mapsto e_P \in E) $
be a fake-flat 2-gauge configuration in $(M,L)$.
Recall the definition of the 2D holonomy of  $\F$ along $\Sigma$ as ${\rm Hol}_v^2(\F,\Sigma,L)=\Psi_{\F_\Sigma}\big(i(\ivb(\Sigma,L)\big)\in \ker(\partial) \subset E$; Def. \ref{2scc}.}

\begin{Theorem}\label{main-3}
 Let $L$ be a 2-lattice decomposition of $(M,\Sigma)$. Let $L_\Sigma$ be the induced {2-lattice} decomposition of $\Sigma$. Let $\F$ be a fake-flat gauge 2-configuration  in $(M,L)$. Let $\F_\Sigma$ be its restriction to $L_\Sigma$. Recall Lem. \ref{mainS2}.

\medskip
\begin{minipage}{0.93\textwidth}{\it Find a positive integer $N$, and, for each $i \in \{1,\dots,N\}$,  a plaquette $P_i\in {L_\Sigma^2}$ (plaquettes might be repeated), an integer $\theta_i\in \{\pm 1\}$, and a quantised path $\gamma_i$, connecting $v$ to the base point $x_{P_i}$ of  $P_i$, such that:}
\begin{enumerate}
\item we have an equivalence of quantised paths (Def. \ref{de:qp} and \ref{de:fgog}):
$$\emptyset_v \cong\; \left (\gamma_1
  \dL^Q(P_1)^{\theta_1}\gamma_1^{-1}\right)\,\,\left(\gamma_2
  \dL^Q(P_2)^{\theta_2}\gamma_2^{-1}\right )  \dots \left(\gamma_N
  \dL^Q(P_N)^{\theta_N}\gamma_N^{-1}\right),$$
\item given any $P \in L^2$, $\displaystyle\sum_{i \in \{1,\dots, N\} \textrm{ such that } P_i=P} \theta_i={\rm sgn} (P)$; cf. \eqref{defP}.
\end{enumerate}
\end{minipage}
\begin{minipage}{0.095\textwidth}\vskip1.9cm \quad (**) \end{minipage}

\medskip
\noindent Then we have the following combinatorial formula for ${\rm Hol}_v^2(\F,\Sigma,L)\in \ker(\partial)\subset E$
\begin{equation}\label{2holcomb}
{\rm Hol}_v^2(\F,\Sigma,L)= g_{\gamma_1} \trr
e_{P_1}^{\theta_1}\,\,g_{\gamma_2} \trr  e_{P_2}^{\theta_2}\,\,
\dots\,\, {g_{\gamma_N} \trr  e_{P_N}^{\theta_N}.}
\end{equation}
 {Here $g_{\gamma_i}$ is the product of the elements of $G$ assigned to
the 1-cells of the quantised path $\gamma_i$ (or their inverses).}
And in particular, fixing $v \in \Sigma$, to be a 0-cell of $L$, the expression \eqref{2holcomb} for ${\rm Hol}_v^2(\F,\Sigma,L)$ does not depend on the assignment $i\in \{1,\dots,N\} \mapsto (P_i,\gamma_i,\theta_i)$ as in $(**)$ chosen. {Moreover, Lem. \ref{main3} holds.}
\end{Theorem}
\begin{proof}
By Lem. \ref{mainS2}, {condition (**)}  is equivalent to:
$[\gamma_1] \trr \iL(P_1)^{\theta_1}\,\,[\gamma_2]\trr \iL(P_2)^{\theta_2} \dots {[\gamma_N]\trr  \iL(P_N)^{\theta_N}}=i\big(\ivb(\Sigma)\big)$.
Hence:
\begin{align*}
{\rm Hol}_v^2(\F,\Sigma,L)&= \Psi_{\F_\Sigma}\big(i(\ivb(\Sigma,L))\big)= \Psi_{\F_\Sigma}\Big( [\gamma_1] \trr \iL(P_1)^{\theta_1}\,\,[\gamma_2]\trr \iL(P_2)^{\theta_2} \dots {[\gamma_N]\trr  \iL(P_N)^{\theta_N}} \Big)\\
                                         &=   \Phi_{\F_\Sigma}( [\gamma_1] )\trr  \Psi_{\F_\Sigma}(\iL(P_1)^{\theta_1})\,\,\Phi_{\F_\Sigma}([\gamma_2])\trr \Psi_{\F_\Sigma}(\iL(P_2)^{\theta_2}) \dots {\Phi_{\F_\Sigma}([\gamma_N])\trr \Psi_{\F_\Sigma}( \iL(P_N)^{\theta_N}} )\\
 &\doteq g_{\gamma_1} \trr
e_{P_1}^{\theta_1}\,\,g_{\gamma_2} \trr  e_{P_2}^{\theta_2}\,\,
\dots\,\, {g_{\gamma_N} \trr  e_{P_N}^{\theta_N}},
\end{align*}
since $(\Psi_{\F_\Sigma},\Phi_{\F_\Sigma})\colon\Pi_2(\Sigma,\Sigma^1,\Sigma^0) \to \Gc$ is a crossed module map.
\end{proof}

\begin{example}\label{L0}
 Consider the 2-lattice decomposition $L_0$ of the 2-sphere $S^2$ with a single 0-cell $v$ and a single 2-cell $P$, whose characteristic map $\phi_P^2\colon D^2 \to S^2$ is positively oriented; cf. Ex. \ref{LandL0}. The base point of $P$ is $v$. A  2-gauge configuration $\F$ is simply an element  $m_P \in E$, colouring its unique plaquette $P$. Fake-flateness imposes  {$m_P \in \ker(\partial).$} Let $\Sigma=S^2$, positively oriented. An assignment as in $(**)$ is such that $N=1$, $P_1=P$, $\gamma_1=\emptyset_v$ and $\theta_1 =1.$
Hence {${\rm Hol}_v^2(\F,\Sigma,L)=m_P$,} as it should.
\end{example}

\begin{example}
To facilitate {drawing} diagrams, let us now see the 2-sphere $S^2$ has being the square $D^2$, where we squash the upper edge and the lower edge to be single points (the north and south poles $v_N$ and $v_S$), and we identify the left and right boundary edges. We give $S^2$ the reverse orientation to the one induced by $[0,1]^2$.
Consider the 2-lattice decomposition $L$ of the 2-sphere, with two zero cells, at  $v_N$ and  $v_S$, and four one cells $s,t,u,v$, all connecting $v_S$ to $v_N$.
We have 2-cells $P,Q,R,S$, indicated in figure below. All plaquettes are based at the south pole. The characteristic map of each plaquette preserves orientation, so $\sgn(P),\sgn(Q),\sgn(R), \sgn(S)=1$. The quantised  boundary of each plaquette is  indicated in figure below.
\begin{figure}[H]
\centerline{\relabelbox
\epsfysize 4.5cm
\epsfbox{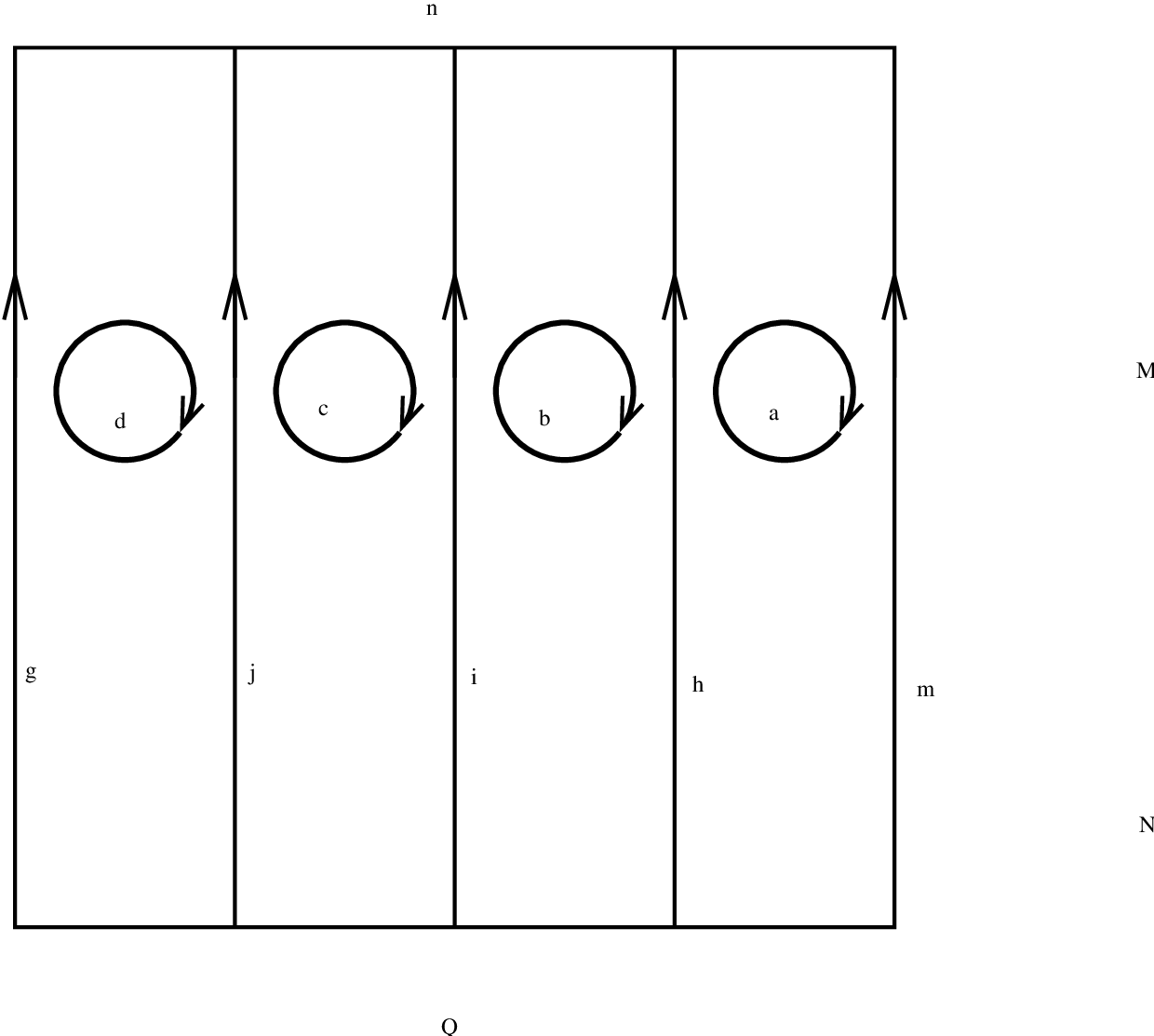}
\relabel{a}{$\small{P}$}
\relabel{b}{$\small{Q}$}
\relabel{c}{$\small{R}$}
\relabel{d}{$\small{S}$}
\relabel{n}{$\small{v_N}$}
\relabel{Q}{$\small{v_S}$}
\relabel{g}{$\small{s}$}
\relabel{m}{$\small{s}$}
\relabel{h}{$\small{t}$}
\relabel{i}{$\small{u}$}
\relabel{j}{$\small{v}$}
\relabel{M}{$\begin{matrix} \dL^Q(P)= ts^{-1}\\
                           \dL^Q(Q)=ut^{-1}\\
                           \dL^Q(R)=vu^{-1}\\
                           \dL^Q(S)=sv^{-1}\end{matrix}$}
\relabel{N}{$\dL^Q(S) \,\, \dL^Q(R)\,\, \dL^Q(Q) \,\, \dL^Q(P)\cong\emptyset_{v_S}$}
\endrelabelbox}
\end{figure}
\noindent{Let $\Sigma=S^2$, with the same orientation. Let $v=v_S$.} An assignment $i\mapsto (P_i,\gamma_i,\theta_i)$ (for $N=4$) satisfying $(**)$ can be $1 \mapsto (S,\emptyset_{v_s},1)$, $2 \mapsto (R,\emptyset_{v_s},1)$, $3 \mapsto (Q,\emptyset_{v_s},1)$ and $4 \mapsto (P,\emptyset_{v_s},1)$. Any cyclic permutation will also work.

A 2-gauge configuration is given by elements $g_s,g_t,g_u,g_v\in G$, colouring the edges $s,t,u,v$, and elements  $d,c,b,a \in E$ colouring $S,R,Q,P$, as indicated in the figure   below. Conditions for fake-flatness to hold {are also made explicit in figure {below}}.
Hence {${\rm Hol}^2_{v_S}(\F,\Sigma,L)=dcba\in \ker(\d) \subset E$.}
\begin{figure}[H]
\centerline{\relabelbox
\epsfysize 4cm
\epsfbox{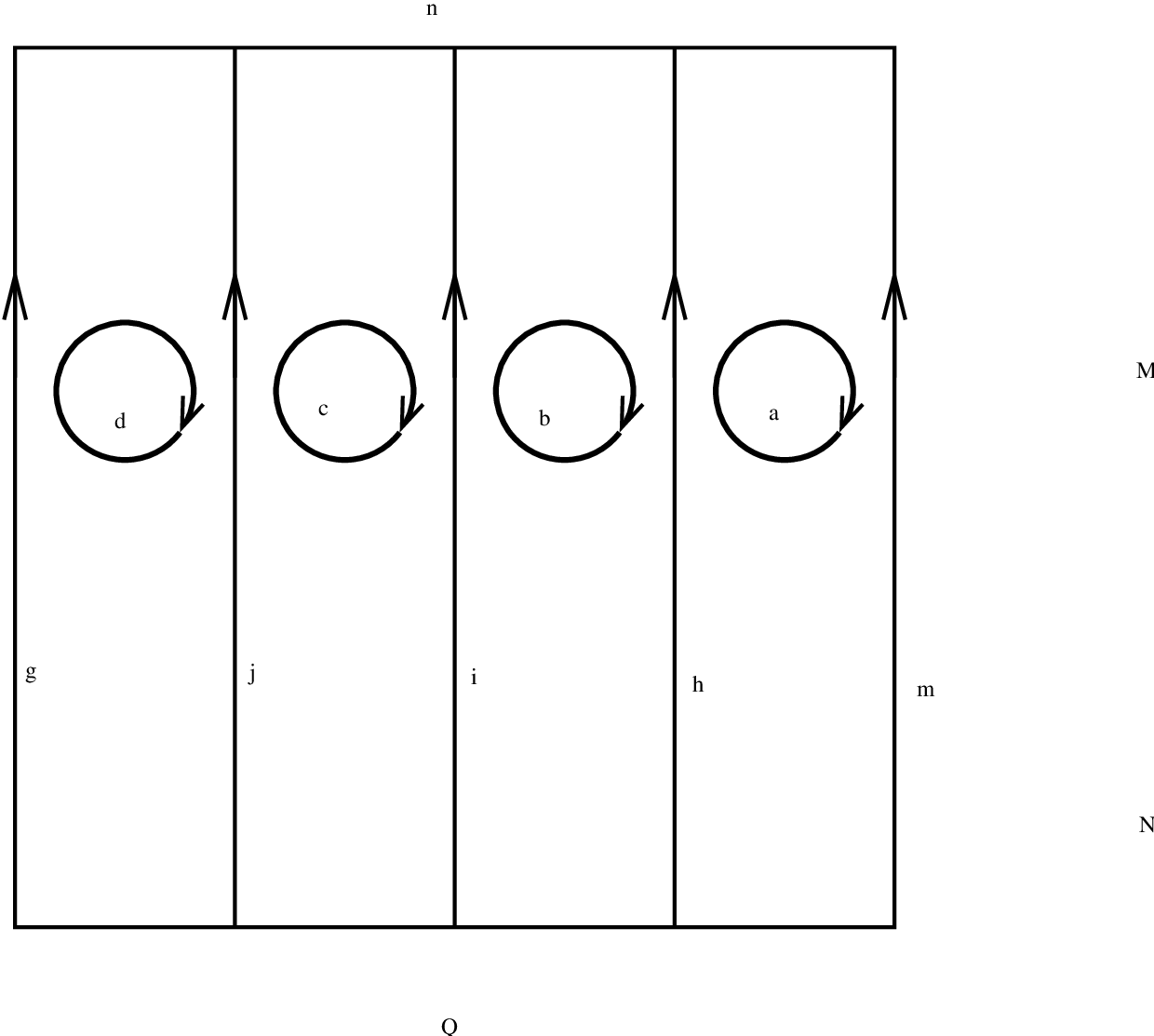}
\relabel{a}{$\small{a}$}
\relabel{b}{$\small{b}$}
\relabel{c}{$\small{c}$}
\relabel{d}{$\small{d}$}
\relabel{g}{$\small{g_s}$}
\relabel{m}{$\small{g_s}$}
\relabel{h}{$\small{g_t}$}
\relabel{i}{$\small{g_u}$}
\relabel{j}{$\small{g_v}$}
\relabel{M}{$\begin{cases} \d(a)= g_tg_v^{-1}\\
                           \d(b)=g_ug_t^{-1}\\
                           \d(c)=g_vg_u^{-1}\\
                           \d(d)=g_s g_v^{-1}\end{cases}$}
\relabel{N}{${{\rm Hol}^2_{v_S}(\F,D^2,L)=dcba}$}
\relabel{n}{$\small{v_N}$}
\relabel{Q}{$\small{v_S}$}
\endrelabelbox}\label{globe}
\end{figure}
\noindent By {Thm. \ref{main-3},} cyclic permutations of $i \mapsto (P_i,\gamma_i,\theta_i)$ {must} yield the same value for {${\rm Hol}^2_{v_S}(\F,\Sigma,L)$, {as (**) is still satisfied. The former can be directly proven:} note $\d(dcba)=1_G$,} thus $dcba$ is central, by the second Peiffer law of the definition of crossed modules (Def. \ref{cm}). Hence  $dcba=d^{-1} dcba d=cbad$.

\end{example}

\begin{example}\label{tet}
  Consider the standard tetrahedron $T\subset \mathbb{R}^3$ displayed   below. Hence the boundary $\Sigma$,  of $T$, with the induced orientation, is given by {the two triangles below} identified along their boundaries. We give $T$ a 2-lattice decomposition derived from the obvious triangulation of $T$. The quantised boundary of each plaquette is indicated in the figure below. {Note,} $P_1,P_2$ and $P_3$ are based in $v_0$, whereas $P_4$ is based in $v_1$.
 $$
 \centerline{\relabelbox
\epsfysize 4.5cm
\epsfbox{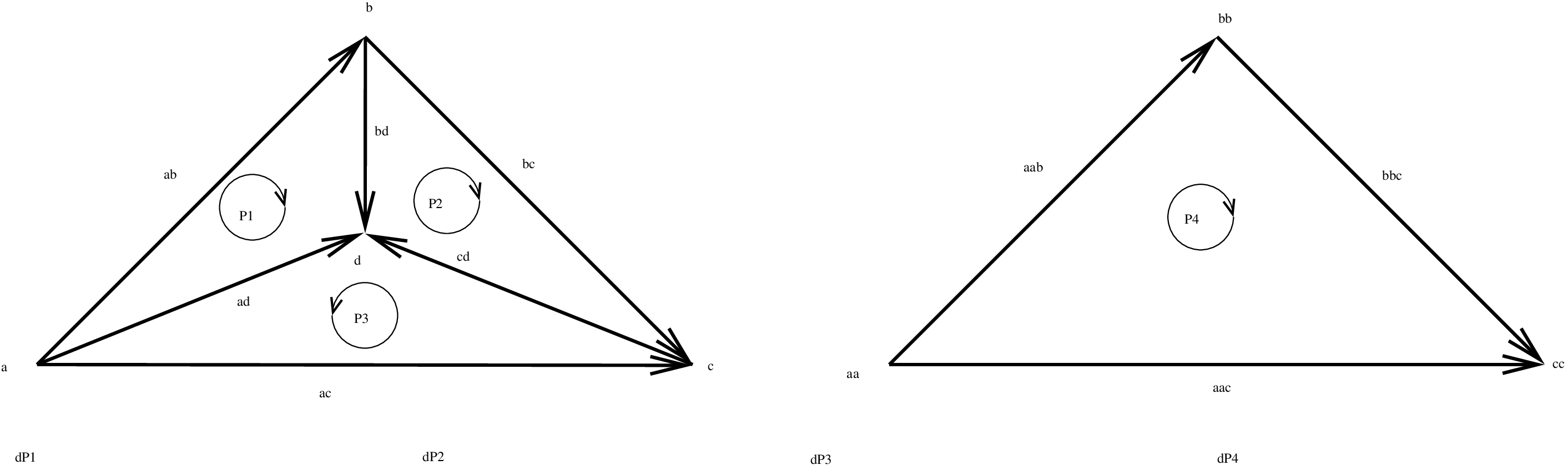}
\relabel{a}{$v_0$}
\relabel{b}{$v_1$}
\relabel{c}{$v_2$}
\relabel{d}{$v_3$}
\relabel{aa}{$v_0$}
\relabel{bb}{$v_1$}
\relabel{cc}{$v_2$}
\relabel{P1}{$P_1$}
\relabel{P2}{$P_4$}
\relabel{P3}{$P_2$}
\relabel{P4}{$P_3$}
\relabel{ab}{$t_{01}$}
\relabel{ad}{$t_{03}$}
\relabel{bc}{$t_{12}$}
\relabel{cd}{$t_{23}$}
\relabel{ac}{$t_{02}$}
\relabel{ab}{$t_{01}$}
\relabel{bd}{$t_{13}$}
\relabel{aab}{$t_{01}$}
\relabel{bbc}{$t_{12}$}
\relabel{aac}{$t_{02}$}
\relabel{dP1}{$\d_L^Q(P_1)=t_{01}t_{13}(t_{03})^{-1}$}
\relabel{dP2}{$\d_L^Q(P_4)=t_{12}t_{23}(t_{13})^{-1}$}
\relabel{dP3}{$\d_L^Q(P_2)=t_{02}t_{23}(t_{03})^{-1}$}
\relabel{dP4}{$\d_L^Q(P_3)=t_{01}t_{12}(t_{02})^{-1}$}
\endrelabelbox}
$$
Let consider 2D holonomy along $\Sigma$ based at $v_0$. An assignment satisfying (**) is such that $N=4$, and:
\begin{align*}
&1\mapsto (P_1,\emptyset_{v_0},1), && 2\mapsto (P_2,\emptyset_{v_0},-1),\\
&3\mapsto (P_3,\emptyset_{v_0},-1), && 4\mapsto (P_4,t_{01},1).
\end{align*}
The general form of a  fake-flat 2-gauge configuration $\F$ of $(T,L)$ is presented below:
 $$
 \centerline{\relabelbox
\epsfysize 4.5cm
\epsfbox{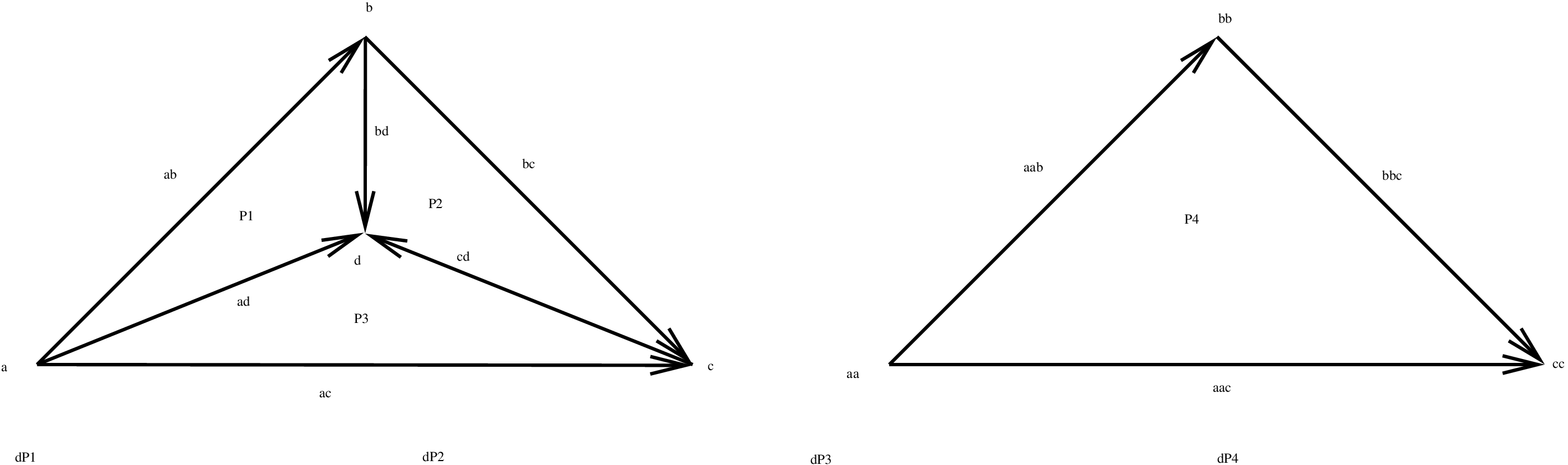}
%\relabel{a}{$v_0$}
% \relabel{b}{$v_1$}
% \relabel{c}{$v_2$}
% \relabel{d}{$v_3$}
% \relabel{aa}{$v_0$}
% \relabel{bb}{$v_1$}
% \relabel{cc}{$v_2$}
\relabel{P1}{$e_1$}
\relabel{P2}{$e_4$}
\relabel{P3}{$e_2$}
\relabel{P4}{$e_3$}
\relabel{ab}{$g_{01}$}
\relabel{ad}{$g_{03}$}
\relabel{bc}{$g_{12}$}
\relabel{cd}{$g_{23}$}
\relabel{ac}{$g_{02}$}
\relabel{ab}{$g_{01}$}
\relabel{bd}{$g_{13}$}
\relabel{aab}{$g_{01}$}
\relabel{bbc}{$g_{12}$}
\relabel{aac}{$g_{02}$}
\relabel{dP1}{$\dG(e_1)=g_{01}g_{13}(g_{03})^{-1}$}
\relabel{dP2}{$\dG(e_4)=g_{12}g_{23}(g_{13})^{-1}$}
\relabel{dP3}{$\dG(e_2)=g_{02}g_{23}(g_{03})^{-1}$}
\relabel{dP4}{$\dG(e_3)=g_{01}g_{12}(g_{02})^{-1}$}
\endrelabelbox}
$$
Hence, by \eqref{2holcomb} it follows that: ${\rm Hol}_{v_0}^2(\F,\Sigma,L)=e_1\,\, e_2^{-1}\,\,  e_3^{-1}\,\, g_{12} \trr e_4$.
\end{example}

\subsection{{2-flat 2-gauge configurations}}\label{2flatn}

Let $(M,L)$ be a 2-lattice. Let $b\in L^3$. The corresponding closed 3-cell {(called a blob)} is {also} denoted by $b=\overline{c_b^3}$. From the definition of 2-lattices (Def. \ref{2-lattice}), the attaching map $\psi_b^3\colon S^2 \to M^2$ of $b$ is an embedding {and}  $\psi_b^3(S^2)=\bound(b)\cong S^2$ is a subcomplex of $M^2$, {called the boundary of the blob $b$.}  Orient $\bound(b)$ by using  $\psi_b^3\colon S^2 \to \bound(b)$.
 Cf. Rem. \ref{notdelta}, \ref{incP} and Def. \ref{defivb}, we have $\ivb\big (\bound(b) \big )= \dL(b)$, in {$\pi_2(M^2,v) \subset \pi_2(M^2,M^1,v)\subset \pi_2(M^2,M^1,M^0)$,} where $v=\psi_b^3(*)$ is the base-point of {$b=\overline{c_b^3}$.}

\begin{definition}[2-flat  2-gauge configuration]\label{2flatconf}
Let ${\Gc}=(\d\colon E \to G,\trr)$ be a crossed module.
Consider a fake-flat {2-gauge} configuration $\F$ {based on
$(M,L)$.}
The boundary $\bound(b)$ of each blob {$b\in L^3$} inherits a 2-lattice {decomposition}. {The} fake-flat 2-gauge configuration $\F$ is said to be 2-flat if for
every blob $b$, we have:
$$
{\rm Hol}^2_v(\F,\bound(b),L)=1_E, \textrm{ where } 1_E \textrm{ is the identity of } E.
$$
Recalling {(Def. \ref{fakeflatconf})} that $\Theta(M,L,\Gc)$ denotes the set of fake-flat 2-gauge configurations {in $(M,L)$}, the set of 2-flat {2-gauge} configurations is denoted  $\Theta_{\rm 2flat}(M,L,\Gc)$.

\end{definition}
{More generally, a fake-flat 2-gauge configuration $\F$ is said to be 2-flat along a cellularly embedded 2-sphere $\Sigma \subset  M$ if, for some $v \in \Sigma \cap M^0$, hence -- by Lem. \ref{main3} -- for all $v \in \Sigma\cap M^0$, it holds that ${\rm Hol}^2_v(\F,\Sigma,L)=1_E$.}

\begin{example}\label{nv} The fake-flat 2-gauge configuration $\vac$ from {the end of \S\ref{ff2gc}} {(the naive vacuum)} is 2-flat. \end{example}
\begin{example}The fake-flat 2-gauge configurations in  {Ex.} \ref{tet} is 2-flat if, and only if, $e_1\,\, e_2^{-1}\,\,  e_3^{-1}\,\, g_{12} \trr e_4 =1_E.$
\end{example}

{Let us provide an algebraic-topological interpretation of 2-flat 2-gauge configurations.}
 Let $\F$ be a {fake-flat 2-gauge configuration in $(M,L)$.} Cf. the construction of ${\rm Hol}^2_v(\F,\bound(b),L)$ in \ref{2scat}. Consider the {discrete 2D parallel transport 2-functor $(\Psi_\F,\Phi_\F)\colon \Pi_2(M^2,M^1,M^0) \to \Gc$ of $\F$; see Thm. \ref{d2dh}.} By the construction in {\ref{2scat} and \ref{sc}}, it holds that ${\rm Hol}^2_v(\F,\bound(b),L)=\Psi_\F(\dL(b))$.
Recall  that $\F \mapsto   (\Psi_\F,\Phi_\F)$ gives a one-to-one correspondence between fake-flat 2-gauge configurations and crossed module maps $\Pi_2(M^2,M^1,M^0) \to \Gc$.

The map of crossed modules  induced by the inclusion $(M^2,M^1,M^0) \to (M^3,M^1,M^0)$ is denoted by $p_2\colon \Pi_2(M^2,M^1,M^0) \to \Pi_2(M^3,M^1,M^0)$. In components $p_2=(p_2',\id)$, where  $p_2'\colon \pi_2(M^2,M^1,M^0)\to {\pi_2(M^3,M^1,M^0)}$ is a surjection and $\id$ is the identity   $\pi_1(M^1,M^0)\to \pi_1(M^1,M^0)$.

 {Cf. Rem. \ref{notdelta}.} Given any 3-cell $b$, note that {$p_2'(\dL(b))=1_{\pi_2(M^3,x_b)}$, the identity of $\pi_2(M^3,x_b)$,} where $x_b$ is the base-point of $b$. These are the only  relations we need {to} impose in order to pass from $\pi_2(M^2,M^1,M^0)$ to $\pi_2(M^3,M^1,M^0)\cong\pi_2(M,M^1,M^0)$; {what is meant by this is in Lem. \ref{desc}. This follows from} the long homotopy exact sequence of the triple $(M^3,M^2,M^1)$, applied to each choice of base point $x \in M^0$, namely:
$${\pi_3(M^3,M^2,x)\to \pi_2(M^2,M^1,x) \ra{p_2'} \pi_2(M^3,M^1,x) \to \pi_2(M^3,M^2,x) \cong \{0\},}$$ together with the fact that the group $\pi_3(M^3,M^2,x)$ is isomorphic to the free $\mathbb{Z}\big(\pi_1(M^1,x)\big)$-module on $L^3$; \cite[Lemma 4.38]{hatcher}.

Cf. the diagram below, a crossed module map $f\colon \Pi_2(M^2,M^1,M^0)\to \Gc$ is said to descend to  $\Pi_2(M,M^1,M^0)$ if there exists a (necessarily unique) crossed module map $f^\flat \colon \Pi_2(M^3,M^1,M^0) \to \Gc$ such that $f^\flat\circ p_2=f$.  $$\xymatrix{&\Pi_2(M^2,M^1,M^0) \ar[d]_{p_2}\ar[r]|>>>>>>f & \Gc \\ & \Pi_2(M^3,M^1,M^0)\ar[ur]|{f^\flat} } .$$
\begin{lemma}\label{desc}
 A crossed module map {$f=(f_2,f_1)\colon \Pi_2(M^2,M^1,M^0)\to \Gc$} descends to  $\Pi_2(M,M^1,M^0)$ if, and only if, for each blob {$b\in L^3$ we have $f_2(\dL(b))=1_E$.}
\end{lemma}
{Note that $\Pi_2(M,M^1,M^0) = \Pi_2(M^3,M^1,M^0)$, by the cellular approximation theorem.}
\begin{proof}
{As mentioned above, this follows from the long homotopy exact sequence of the triple $(M^3,M^2,M^1)$, applied to each choice of base point $x \in M^0$; details can be found in \cite{martins_crossed}, for the case of CW-complexes with a single base-point.}
Alternatively we can also use the higher-dimensional van Kampen theorem of Brown and
Higgins; see \cite{brown_2dvk,brown_higgins_colimits,brown_higgins_cubes,brown_hha} and \cite[\S 6]{brown_higgins_sivera}, stating that (under mild conditions) the fundamental crossed module functor preserves colimits. {Note that the conditions of 2-lattices (Def. \ref{2-lattice}) imply that for each $b\in L^3$, the corresponding closed 3-cell $\overline{c_b^3}=b$ is  a subcomplex of $M$, homeomorphic to $D^3$. {Moreover $b^2=\bound(b)$,  $b^1=\bound(b)^1$ and $b^0=\bound(b)^0$}. From \cite[\S 6.3]{brown_hha}, it  follows that the diagram {\eqref{po}} below is a pushout diagram in the
category of crossed modules of groupoids:}
\begin{equation}\label{po}
{\xymatrix{& {\displaystyle\bigsqcup_{b \in L^3}} \Pi_2(\bound(b),\bound(b)^1,\bound(b)^0)  \ar[rrr]^{p_1}\ar[d]_{i_1}&&& {\displaystyle\bigsqcup_{b \in L^3}}\Pi_2(b,b^1,b^0)\ar[d]^{i_2}\\
  &\Pi_2(M^2,M^1,M^0)\ar[rrr]^{p_2} &&&\Pi_2(M^3,M^1,M^0)   }}
\end{equation}
{In the diagram \eqref{po} above, all arrows are induced by inclusions. Also, $\pi_1(b^1,b^0)=\pi_1(\bound(b)^1,\bound(b)^0)$. Given $x \in b^0=\bound(b)^0$, we pass from $\pi_2(\bound(b),\bound(b)^1,x)$ to $\pi_2(b,b^1,x)$ by quotienting by the normal closure of $\ivx(\bound(b))\in \pi_2(\bound(b),x) \subset \pi_2(\bound(b),\bound(b)^1,x) $; we are using the notation of Def \ref{defivb}.   
By inspecting \eqref{po}, and applying the universal property of pushouts, it hence follows that a crossed module map $f\colon \Pi_2(M^2,M^1,M^0)\to \Gc$ descends to  $\Pi_2(M,M^1,M^0)$ if, and only if, for each $b \in L^3$, and for each $x \in b^0$, it holds that $f_2(\ivx(\bound(b)))=1_E$. {by Lem. \ref{main3}, in order for the latter to happen, it suffices to check that $f_2(\ivb(\bound(b)))=1_E$, if $v$ is the base-point of $b$,} which is the same as saying that $f_2(\dL(b))=1_E$.} 
\end{proof}

{Combining Lem. \ref{desc} with Thm. \ref{d2dh},  yields the following interpretation of fake-flat 2-gauge configurations.} 
\begin{Theorem}[2-flat 2D parallel transport 2-functors]\label{2-flat-conf-theorem}
 The bijection $\F \in  \Theta(M,L,\Gc)\mapsto (\Psi_\F,\Phi_\F)$ of Thm. \ref{d2dh} yields a bijection between 2-flat  configurations $\F\in \Theta_{2{\rm flat}}(M,L,\Gc)$ and crossed module maps $\Pi_2(M,M^1,M^0) \to \Gc$, from now on called 2-flat 2D parallel transport 2-functors. The correspondence sends $\F\in  \Theta_{2{\rm flat}}(M,L,\Gc)$ to the crossed module map  $(\Psi_\F^\flat,\Phi_\F^\flat)\colon \Pi_2(M,M^1,M^0) \to \Gc$ that $(\Psi_\F,\Phi_\F)$ descends to.
\end{Theorem}

\newcommand{\op}{operator}
\newcommand{\Tgo}{{\cal T}}   %% gp of gauge operators

\section{Gauge transformations}\label{GTRANS}
Throughout this section,  we fix a crossed module $\Gc=(\dG\colon E \to G,\trr)$    of groups \S \ref{crossed_modules_def} and a 2-lattice $(M,L)$  \S \ref{sec:2-lattices}. Recall that  $(L^0,L^1)$ has the structure of a directed graph $\s,\t\colon L^1 \to L^0$, see \S \ref{ss:POL}. We  denote the edges (1-cells) of $L$ as $x \ra{t} y$, where $x=\s(t)$ and $y=\t(t)$  are the source and target of $t$. (It may be that $x=y$.)
\subsection{The group {$\Tgo=\Tgo(M,L,{\Gc})$} of gauge \op s}\label{2gaugetrans}

If $G$ is a group and  $S$ is a set we put $G^S$ to denote the group $\prod_{s \in S} G$, with pointwise multiplication.

\newcommand{\sdp}{\rtimes_\bullet} %% ExV semidirect product

\begin{definition}[The group $\Tgo(M,L,{\Gc})$ of gauge \op s]\label{2gt}
There is a left-action $\bullet$ of
${\cal V}(M,L,{\Gc}) = G^{L^0}$ on
 ${\cal E}(M,L,{\Gc}) = E^{L^1}$
by automorphisms.
Given $\eta \in  {\cal E}(M,L,{\Gc})$ and $u \in {\cal V}(M,L,{\Gc})$, the action $\bullet$ has the form:
\begin{equation}\label{bullet}
(u \bullet \eta)\left (\sigma(t) \ra{t} \tau{(t)} \right)
                  = u\big(\sigma(t)\big) \trr \big (\eta \left(\sigma(t) \ra{t} \tau{(t)} \right)\big),
\end{equation}
for each $\sigma(t) \ra{t} \tau{(t)} $ in $L^1$. (Note that $\trr$ denotes the underlying action of $G$ on $E$, which exists since $(\dG\colon E \to G, \trr)$ is a crossed module; Def. \ref{cm}.)
We define the group $\Tgo={\cal T}(M,L,{\Gc})$ of gauge \op s to be:
\begin{equation}\label{eq:defT}
\Tgo(M,L,{\Gc})
 \;  = \; {\cal E}(M,L,{\Gc}) \sdp  {\cal  V}(M,L,{\Gc})
 \;  = \; E^{L^1} \sdp G^{L^0} .
\end{equation}
Here $\sdp$ denotes semidirect product.
In particular we take:
\begin{equation}\label{semidef}(\eta,u)(\eta',u') \; = \; \big( \eta\, u \bullet \eta', u\,u'\big).\end{equation}
\end{definition}

\begin{example}
Recall the 2-lattice decomposition $L$ of $S^2$ from Fig. \ref{L1}. We have a unique edge and a unique vertex. Hence
{$ \Tgo(S^2, L,\Gc) = E \rtimes_{\trr} G.$} If we extend  $L$ to be a 2-lattice decomposition $L_\mathfrak{g}$ of $S^3$, by adding 3-cells (Ex. \ref{LandL0}), it also holds  {that $ \Tgo(S^3, L_\mathfrak{g},\Gc)= E \rtimes_{\trr} G$. This is because groups of gauge \op s on 2-lattices depend only on 1-skeletons.}
\end{example}

The group ${\cal T}(M,L,{\Gc})$ of gauge operators
acts on the set  $\Theta(M,L,{\Gc})$ of fake-flat 2-gauge configuration {on $(M,L)$}  in a way such that the 2D
holonomy is preserved; see \S \ref{gaugeaction}, below.  Moreover, this action restricts to an action of  ${\cal T}(M,L,{\Gc})$ on the set   $\Theta_{\rm 2flat}(M,L,{\Gc})$ of 2-flat configurations; Def. \ref{2flatconf}.
In order to  present the action, we now define a double
groupoid ${\cal D}({\Gc})$ out of the crossed module ${\Gc}$; see \cite[\S 6.6]{brown_higgins_sivera}, \cite{brown_higgins_cubes} and \cite{martins_picken}.

\subsection{The double groupoid ${\cal D}(\Gc)$}\label{adg}

The definition of a double groupoid appears e.g. in \cite[\S 6.1]{brown_higgins_sivera} and \cite{martins_picken,Zu1}. In this paper double groupoids are edge-symmetric and have a unique  object. We explain the definition of such double groupoids as we elaborate how  a crossed module $\Gc=(\dG\colon E \to G, \trr)$ of groups gives rise to one, denoted $\D(\Gc)$; \cite[\S 6.6]{brown_higgins_sivera}.

We have a unique object $*$, and sets $\D^1_H(\Gc)$ and  $\D^1_V(\Gc)$ of horizontal and vertical 1-squares in $\Gc$; \cite{martins_picken}. These sets of horizontal and vertical 1-squares in $\Gc$ {consist of}  diagrams of the form:
$$\left ( * \ra{X} * \right) \quad \textrm{ and } \quad \left (\,\,\begin{CD} &* \\
              &@A ZAA \\
              &*
  \end{CD}\right), \textrm{ where } X,Z \in G.  $$
Horizontal and vertical 1-squares  in $\Gc$ are  composed in the obvious way, here shown for horizontal 1-squares:
$${\big(
* \ra{\quad X \quad } *\big) \circ \big(* \ra{\quad Y \quad } *\big) =\big(*\ra{\quad XY \quad }* \big), \textrm{ where } X,Y \in G.}
$$  We therefore have horizontal and vertical groupoids, also denoted $\D^1_H(\Gc)$ and  $\D^1_V(\Gc)$, with a single object. The sets of morphisms are in one-to-one correspondence with $G$.  The groupoids $\D^1_H(\Gc)$ and  $\D^1_V(\Gc)$ are isomorphic.  An obvious isomorphism $\D^1_V(\Gc) \to \D^1_H(\Gc)$ is obtained by clockwise rotation.

We have a set $\D^2(\Gc)$ of {squares} in $\Gc$. This set consists of  diagrams $K$ of the form below:
\begin{equation} \label{square}K=\quad\begin{CD} &* @>W>> &*\\
              &@A ZAA \hskip-1.3cm \scriptstyle{e}  &@AA YA\\
              &* @>> X> &*
  \end{CD} \textrm{ \quad \quad \quad where } X,Y,Z,W \in G \textrm{ and } e \in E \textrm{ are such that } \d(e)=XYW^{-1}Z^{-1}.\end{equation}
\begin{remark}[{Squares in $\Gc$ and fake-flat 2-gauge configurations of $[0,1]^2 $}]\label{transport}{Elements $K \in \D^2(\Gc)$  hence can be seen as fake-flat 2-gauge
  configurations on a certain 2-lattice decomposition of
   $D^2=[0,1]^2$; see Ex. \ref{fake0}.}
   %This will have prime importance in \S \ref{fgtb}.}
   \end{remark}

 Several maps exist connecting $\D^1_H(\Gc)$, $\D^1_V(\Gc)$ and $\D^2(\Gc)$; see  below {($K$ is as in \eqref{square})}:
\begin{equation}\label{bounds}
d_l (K)=\left(\,\,\,\begin{CD} &* \\
              &@A ZAA \\
              &*
  \end{CD}\right), \quad d_r(K)=\left(\,\,\,\begin{CD} &* \\
              &@A YAA \\
              &*
  \end{CD}\right), \quad d_u(K) = \left( \begin{CD} &* @>W>> &*\end{CD}\right) \,\,\textrm{ and } \quad d_d(K)\quad =\quad \left(\begin{CD} &* @>X>> &*\end{CD}\right).
\end{equation}
\begin{equation}\label{ids}
\id_V(\begin{CD} &* @>X>> &*\end{CD})\quad=\quad\begin{CD} &* @>X>> &*\\
              &@A 1_G AA \hskip-1.3cm \scriptstyle{1_E}  &@AA 1_G A\\
              &* @>> X> &*\end{CD} \quad \quad \textrm{ and  } \quad  \id_H\left(\,\, \begin{CD} &* \\
              &@A YAA \\
              &* \end{CD} \right)\quad=\quad \begin{CD} &* @>1_G>> &*\\
              &@A YAA \hskip-1.3cm \scriptstyle{1_E}  &@AA YA\\
              &* @>> 1_G> &*
  \end{CD}\,\,\,\,\,. \end{equation}

 Horizontal and vertical compositions of 2-squares can be done when
  squares match on the relevant sides:
\begin{minipage}[v]{0.50\textwidth}
$$\hskip-0.4cm {\begin{CD} &* @>W>> &*\\
              &@A ZAA \hskip-1.3cm \scriptstyle{e}  &@AA YA\\
              &* @>> X> &*
  \end{CD}  \quad \quad {\begin{CD} &* @>W'>> &*\\
              &@A YAA \hskip-1.3cm \scriptstyle{e'}  &@AA Y'A\\
              &* @>> X'> &*
  \end{CD}}
  \quad   =  \quad  \begin{CD} &* @>WW'>> &*\\
              &@A ZAA \hskip-1.3cm \scriptstyle{(X \trr e') e} &@AA Y'A\\
              &* @>> XX'> &*
\end{CD}}$$
\end{minipage}
\begin{minipage}[v]{0.01\textwidth}
\quad and \end{minipage}
  \begin{minipage}[v]{0.47\textwidth}
 \begin{equation}\label{vertcomp}{
\begin{CD}  \begin{CD} &* @>W'>> &*\\
              &@A Z' AA \hskip-1.3cm \scriptstyle{e'}  &@AA Y'A\\
              &* @>> W> &*
 \end{CD} \\ \begin{CD} &* @>W>> &*\\
              &@A ZAA \hskip-1.3cm \scriptstyle{e}  &@AA YA\\
              &* @>> X> &*
  \end{CD}  \end{CD}\quad \quad  =  \quad \quad  \begin{CD} &* @>W'>> &*\\
              &@A ZZ' AA \hskip-1.3cm \scriptstyle{e Z \trr e'} &@AA YY'A\\
              &* @>> X> &*
  \end{CD}}
\end{equation}
\end{minipage}

\noindent These compositions are associative.
Hence the set $\D^2(\Gc)$ of squares in $\Gc$ is the set of morphisms of two categories, called horizontal and vertical categories. The correspondent sets of objects are the sets of vertical and horizontal squares in $\Gc$, respectively. Source and target maps are in \eqref{bounds}. Unit maps are in \eqref{ids}.

\begin{remark}[Interchange law in $\D(\Gc)$]\label{il}
{Horizontal and vertical
compositions  in $\D(\Gc)$ satisfy the interchange law, which says that the
  composition  indicated below does not depend on the order whereby it is done. } 
$$
\begin{CD}
 \begin{CD} &* @>W'>> &*\\
              &@A Z' AA \hskip-1.3cm \scriptstyle{f}  &@AA C A\\
              &* @>> W> &*
 \end{CD} \qquad  \begin{CD} &* @>W''>> &*\\
              &@A CAA \hskip-1.3cm \scriptstyle{f'}  &@AA Y''A\\
              &* @>> W'''> &*
 \end{CD} \\ \begin{CD} &* @>W>> &*\\
              &@A ZAA \hskip-1.3cm \scriptstyle{e}  &@AA BA\\
              &* @>> X> &*
  \end{CD} \qquad \begin{CD} &* @>W'''>> &*\\
              &@A BAA \hskip-1.3cm \scriptstyle{e'}  &@AA Y'''A\\
              &* @>> X'> &*
  \end{CD}
  \end{CD}
$$
{This means that we can either first perform horizontal compositions, and then vertical compositions, or vice-versa, yielding the same result.}
To prove the interchange law  we must make explicit use of the 2nd Peiffer
condition  in  Def. \ref{cm}, for crossed modules of groups. (All other mentioned properties follow from the 1st Peiffer relation in Def. \ref{cm}, and the fact that $G$ acts on $E$ by automorphisms.)
\end{remark}

{The horizontal and vertical categories are both groupoids. Given $K\in \D^2(\Gc)$, {the inverses $r_V$ and $r_H$ of $K$, with respect to the vertical and horizontal compositions, {called vertical and horizontal reverses of $K$,  are given in  \eqref{transfor}, below.} This finishes the construction of the double groupoid $\D(\Gc)$.} }
\begin{equation}\label{transfor}  {r_V\left(\quad \begin{CD} &* @>W>> &*\\
              &@A ZAA \hskip-1.3cm \scriptstyle{e}  &@AA YA\\
              &* @>> X> &*
  \end{CD}\quad \right)=\qquad\begin{CD} &* @>X>> &*\\
              &@A Z^{-1}AA \hskip-1.3cm \scriptstyle{Z \trr e^{-1}}  &@AA Y^{-1}A\\
              &* @>> W> &*
  \end{CD}\quad\quad,\quad \quad r_H\left ( \quad\begin{CD} &* @>W>> &*\\
              &@A ZAA \hskip-1.3cm \scriptstyle{e}  &@AA YA\\
              &* @>> X> &*
  \end{CD}\quad\right)=\qquad\begin{CD} &* @>W^{-1}>> &*\\
              &@A YAA \hskip-1.3cm \scriptstyle{X\trr e^{-1}}  &@AA ZA\\
              &* @>> X^{-1}> &*
  \end{CD}}\end{equation}

\subsection{Full gauge transformations between  {fake-flat 2-gauge} configurations}\label{fgtb}
 The action of the group  ${\cal T}(M,L,{\Gc})$ {\S \ref{2gaugetrans}} of gauge operators on the set  $\Theta(M,L,{\Gc})$ {\S \ref{ff2gc}} of fake flat 2-gauge configurations
is given in \S\ref{gaugeaction}.
We still  need  some
technicalities.
\subsubsection{Groupoid $\Theta^\#(M,L,\Gc)$ of fake-flat 2-gauge configurations and full gauge transformations}\label{2fgroupoid}

\begin{definition}[Full gauge transformation]\label{fgt} 
%Cf. \eqref{bounds}. 
A full gauge transformation $\U=(U_2,U_1)$, starting in the  fake-flat 2-gauge configuration $\F=({\F^1}\colon L^1 \to G\,\,,\,\,{\F^2}\colon L^2 \to E)$, is given by a pair of maps:
$U_2: L^1 \rightarrow \D^2(\Gc)$ and $U_1\colon L^0 \to \D^1_V(\Gc)$,
such that:
\begin{itemize}
\item
{Let $v \in L^0$. Put: 
 $U_1(v)=\left(\,\,\, \begin{CD} &* \\
              &@A \,\,g_v AA \\
              &* \end{CD}\right)$. 
Let $\E_v$ be the set of edges {of $L$} incident to $v$. If  $t \in \E_v$ then:}
\begin{equation}
{\hskip-.3cm\s(t)=v \implies d_l \big( U_2(t)\big)=\left(\,\,\,\begin{CD} &* \\
              &@A \,\,g_v AA \\
              &* \end{CD}\right)  \textrm{ and }  \t(t)=v \implies d_r \big( U_2(t)\big)=\left(\,\,\, \begin{CD} &* \\
              &@A g_v AA \\
              &* \end{CD}\right)}\,\,.
\end{equation} Hence if two edges $t$ and $t'$ share a vertex, the corresponding vertical sides of
   $U_2(t)$ and $U_2(t')$ match.
\item For each edge $t \in L^1$,  it holds that $\dd_u\big(U_2(t))= {\F^1}(t)$. {(For notation see \eqref{bounds}.)}
\end{itemize}
\end{definition}

\begin{definition}\label{idf}An example of a full gauge transformation starting in $\F$ is $\id_\F$. It is such that $g_v=1_G$, for each $v \in L^0$, {hence $U_1(v)$ is an identity vertical 1-square. $\id_\F$ assigns $\id_V({\F^1}(t))$ in \eqref{ids} to each $t \in L^1$.}
\end{definition}

\begin{remark}[Full gauge transformations and crossed module homotopies]
\label{bh}
In the language of \cite[\S 2.1]{BrownIcen} and \cite{BROWNHOL}, full gauge transformations, starting at $\F$, boil down to  crossed module homotopies {starting} on the associated crossed module map $(\Psi_\F,\Phi_\F)\colon \Pi_2(M^2,M^1,M^0) \to \Gc$ of Thm. {\ref{d2dh}}.  (An explanation of crossed module homotopies in a language close to this paper's is in \cite{martins_porter,martins_cw_complex,martins_gohla,martins_CW_complex2}.) In order to prove this fact, we must use the fact that the groupoid $\pi_1(M,M^0)$ is free on the set of 1-cells {of $M$}. Therefore, a {crossed module homotopy starting on the crossed module map}  $(\Psi_\F,\Phi_\F)\colon \Pi_2(M^2,M^1,M^0) \to \Gc$ can be arbitrarily (and uniquely) specified by its value on the set of 0  and 1-cells of $M$, yielding our full gauge transformations.   \end{remark}

For an explanation of crossed module homotopy in the general framework of {\em crossed complexes}, {see \cite[\S 9]{brown_higgins_sivera} and \cite{Brown_tensor}.}

\medskip

A full gauge transformation $\U$, starting {in} the fake-flat 
2-gauge configuration $\F=({\F^1},{\F^2})$, {\em transforms} $\F$ into 
another fake-flat 2-gauge configuration,
denoted 
$$
\U \trr \F=(\U \trr {\F^1},\U \trr {\F^2}) .
$$ 
The definition of $ \U \trr \F$ is given below.
\footnote{
This is a consequence of the general construction in 
\cite[\S 2.1]{BrownIcen}, \cite{BROWNHOL} and \cite[\S 9]{brown_higgins_sivera} of crossed module and crossed complex homotopy. {Some explicit calculations are in \cite{martins_porter,martins_CW_complex2,martins_gohla}.}
}

\medskip

We give a combinatorial explanation of $\U \trr \F=(\U \trr {\F^1},\U \trr {\F^2})$, based on the framework of this paper.

At the level of edge
colourings, if $t\in L^1$ then $(\U\trr{\F^1})(t)$
is the bottom colour of the square $U_2(t)\in \D^2(\Gc)$, 
i.e.   $(\U\trr{\F^1})(t)=\dd_d(U_2(t))$; see \eqref{bounds}.
Let us describe  $\U\trr \F^2(P)$, where $P$ is a plaquette.  We  consider two cases, depending on the attaching map $\psi_P^2\colon S^1 \to M^1$ of the  closed 2-cell $\overline{c_P^2}$ (also denoted $P$); cf. Def. \ref{2-lattice}.
\begin{enumerate}
 \item If $\psi_P^2\colon S^1 \to M^1$ is constant, then  $\psi_P^2(S^1)=\{v\}$, where $v$ is a 0-cell of $L$. And we then put:
$$(\U\trr \F^2)(P)=g_v \trr \F^2(P),$$
where $g_v$ is defined in {Def. \ref{fgt}}.
Since $\dG( \F^2(P))=1_G$ (cf. Prop. \ref{fake-flat-conf-inpractice}),  $\U\trr \F$ is fake-flat at $P$, because: $$\dG\big((\U\trr \F^2)(P)\big)=\dG(g_v \trr \F^2(P))=g_v\dG( \F^2(P)) g_v^{-1}=1_G.$$
\item Otherwise, we now make explicit use of the fact that  $ \psi_P^2\colon S^1 \to M^1$ must then be an embedding; as such the characteristic map $\phi_P^2\colon D^2 \to \overline{e^2_P}=P$ {of $P$} is a homeomorphism. Consider $P\times [0,1]\subset M^2 \times [0,1]$, with the obvious product lattice decomposition, where $[0,1]$ has unique 0-cells at $0$ and $1$. The 2-dimensional lattice made out of the top and lateral sides of $P\times [0,1]$ can be given a fake-flat 2-gauge configuration, obtained by putting together $\F$ and $\U$.   {We refer to Fig. \ref{full_trans}.} It depicts a fake-flat 2-gauge configuration in the boundary of the cylinder $P\times [0,1]$. The base-point of plaquette $P$, whose attaching map is oriented counterclockwise, is {$v=v_1$}. The top {$P \times \{1\}$} of the cylinder is coloured by the restriction of $\F$ to $P$. The sides of the cylinder are coloured by the  1 and 2-squares $U_1(x)\in \D^1_V(\Gc)$ and $U_2(t)\in \D^2(\Gc)$, where $x$ is a vertex of $\bound(P)$ and $t$ is an edge of $\bound(P)$. {(Each $U_2(t)$ can be seen as a fake-flat 2-gauge configuration of the 2-disk $[0,1]^2$; see Rem. \ref{transport}. The direction of the edges of $\bound(P)$ gives an unambiguous way to transport the fake-flat 2-gauge configuration $U_2(t)$ of $[0,1]^2$ onto the correspondent lateral square in Fig \ref{full_trans}.)}  {Finally, the bottom $P \times \{0\}$  of the cylinder $P \times [0,1]$ is coloured with $\U\trr \F=\F'.$}

By definition, the `gauge-transformed' colour {$e_P'=\U \trr {\F^2}(P)$} of the plaquette $P$ {(which is
based at $v$)} is the 2D holonomy, based at $v'$, along the {2-disk}
consisting of the top and lateral sides of the cylinder in Fig. \ref{full_trans}, with the fake-flat 2-gauge configuration obtained by putting together $\F$ and $\U$. By Thm.
\ref{independence1}, the 2D holonomy {along} this {2-disk} is well
defined. Also by \eqref{fake-flatness-preserved} {$\U\trr \F$} is fake-flat at $P$.
\end{enumerate}

\begin{figure}
\centerline{\relabelbox
\epsfysize 6.8cm
\epsfbox{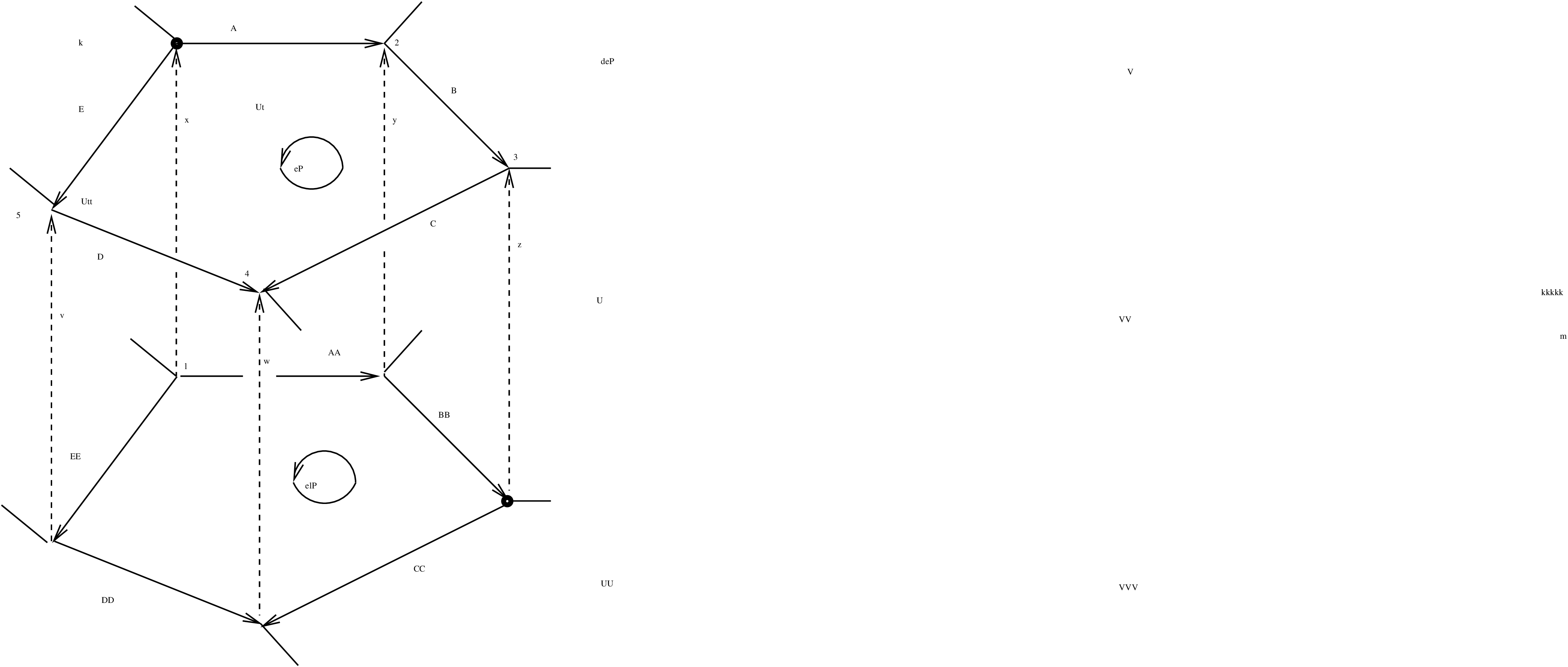}
\relabel{A}{$g_{t_1}$}
\relabel{AA}{$g_{t_1}'$}
\relabel{B}{$g_{t_2}$}
\relabel{BB}{$g_{t_2}'$}
\relabel{C}{$g_{t_3}$}
\relabel{CC}{$g_{t_3}'$}
\relabel{D}{$g_{t_4}$}
\relabel{DD}{$g_{t_4}'$}
\relabel{E}{$g_{t_5}$}
\relabel{EE}{$g_{t_5}'$}
\relabel{k}{$v=v_1$}
\relabel{2}{$v_2$}
\relabel{3}{$v_3$}
\relabel{4}{$v_4$}
\relabel{5}{$v_5$}
\relabel{l}{$v'$}
\relabel{x}{$u_1$}
\relabel{y}{$u_2$}
\relabel{z}{$u_3$}
\relabel{w}{$u_4$}
\relabel{v}{$u_5$}
%\relabel{Ut}{$U_2(t_1)$}
%\relabel{Utt}{$U_2(t_5)$}
\relabel{eP}{$e_P$}
\relabel{elP}{$e_P'$}
\relabel{U}{$U_2(t_1)=\qquad\begin{CD} &* @>  \quad  g_{t_1} \quad   >> &*\\
              &@A u_1AA \hskip-2.1cm \scriptstyle{\eta(t_1)}  &@AA u_2A\\
              &* @>> \quad   g_{t_1}'  \quad > &* \end{CD}$}
\relabel{V}{$U_2(t_2)=\qquad\begin{CD} &* @>  \quad  g_{t_2} \quad   >> &*\\
              &@A u_2AA \hskip-2.1cm \scriptstyle{\eta(t_2)}  &@AA u_3A\\
              &* @>> \quad   g_{t_2}'  \quad > &* \end{CD}$}
\relabel{VV}{$U_2(t_3)=\qquad\begin{CD} &* @>  \quad  g_{t_3} \quad   >> &*\\
              &@A u_3AA \hskip-2.1cm \scriptstyle{\eta(t_1)}  &@AA u_4A\\
              &* @>> \quad   g_{t_3}'  \quad > &* \end{CD}$}
\relabel{VVV}{$U_2(t_4)=\qquad\begin{CD} &* @>  \quad  g_{t_4} \quad   >> &*\\
              &@A u_5AA \hskip-2.1cm \scriptstyle{\eta(t_4)}  &@AA u_4A\\
              &* @>> \quad   g_{t_4}'  \quad > &* \end{CD}$}
\relabel{UU}{$U_2(t_5)=\qquad\begin{CD} &* @>  \quad  g_{t_5} \quad   >> &*\\
              &@A u_1AA \hskip-2.1cm \scriptstyle{\eta(t_5)}  &@AA u_5A\\
              &* @>> \quad   g_{t_5}'  \quad > &* \end{CD}$}
\relabel{deP}{$\dG(e_P)=g_{t_5}g_{t_4}g_{t_3}^{-1} g_{t_2}^{-1} g_{t_1}^{-1} $}
              \endrelabelbox}
\caption{\label{full_trans} {A full gauge transformation $\U$, transforming  $\F$ into $\U\trr \F=\F'$, in the vicinity of a plaquette $P$. The  2-gauge configurations $\F$ and $\F'=\U \trr \F$  are  (respectively) at the top and at the bottom of the cylinder $P\times I$.} {Note that, given an edge $t \in L^1$, the element of $E$ associated to $U_2(t)$ is here denoted by $\eta(t)$. Also $u_1,\dots, u_5\in G$ are given by $g_{v_1},\dots g_{v_5}$. The squares in $\Gc$ on the right are used to give $E$ labellings to the lateral squares of the cylinder on the left, in the obvious way.}}
\end{figure}

\begin{remark}\label{concrete}A more concrete expression for {$e_P'=\U \trr {\F^2}(P)$} can be {derived by using the double groupoid $\D^2(\Gc)$.} In the example in {Fig. \ref{full_trans},} we evaluate the following composition in $\D^2(\Gc)$. (We note that the elements of $E$ assigned to the three squares in the bottom right arise from the horizontal reverses of the squares $U_2(t_3)$, $U_2(t_2)$ and $U_2(t_1)$, above; {cf. \eqref{transfor}}.) And then $e_P'$ is the element of $E$ assigned to the resulting square in $\Gc$.
\begin{equation}\label{compconcr}
{\hskip-1cm\centerline{\relabelbox
\epsfysize 4cm
\epsfbox{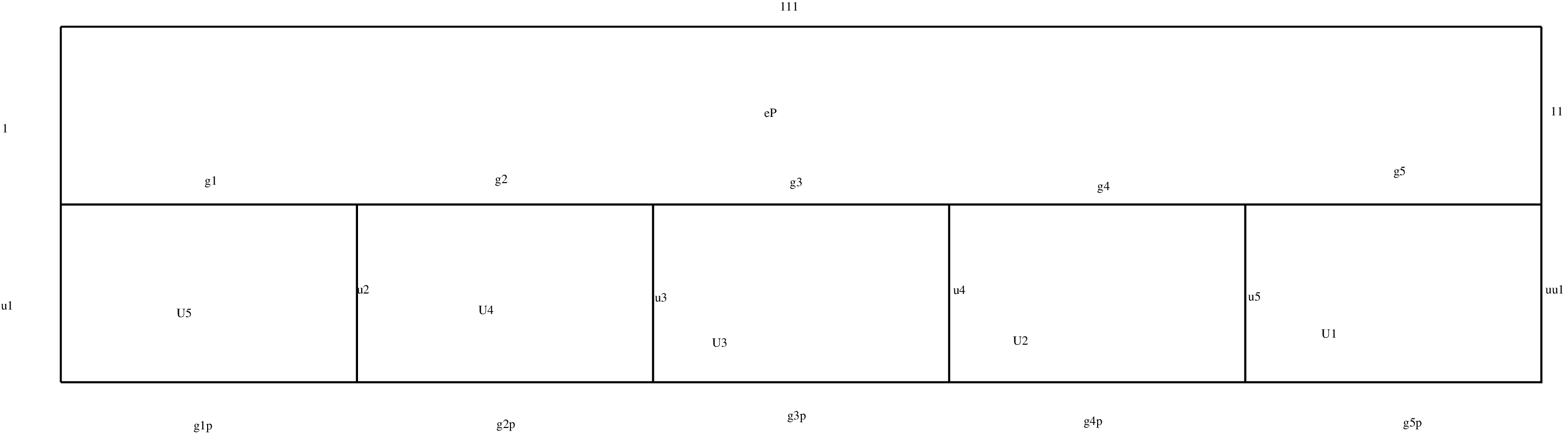}
\relabel{1}{$1_G$}
\relabel{11}{$1_G$}
\relabel{111}{$1_G$}
\relabel{g1}{$g_{t_5}$}
\relabel{g1p}{$g_{t_5}'$}
\relabel{g2}{$g_{t_4}$}
\relabel{g2p}{$g_{t_4}'$}
\relabel{g3}{$g_{t_3}^{-1}$}
\relabel{g3p}{${g_{t_3}'}^{-1}$}
\relabel{g4}{$g_{t_2}^{-1}$}
\relabel{g4p}{${g_{t_2}'}^{-1}$}
\relabel{g5}{$g_{t_1}^{-1}$}
\relabel{g5p}{${g_{t_1}'}^{-1}$}
\relabel{u1}{$u_1$}
\relabel{u2}{$u_5$}
\relabel{u3}{$u_4$}
\relabel{u4}{$u_3$}
\relabel{u5}{$u_2$}
\relabel{U1}{$g'_{t_1}\trr \eta(t_1)^{-1}$}
\relabel{U2}{$g'_{t_2}\trr \eta(t_2)^{-1}$}
\relabel{U3}{$g'_{t_3}\trr \eta(t_3)^{-1}$}
\relabel{U4}{$\eta(t_4)$}
\relabel{U5}{$\eta(t_5)$}
\relabel{uu1}{$u_1$}
\relabel{eP}{$e_P$}
\endrelabelbox}}
\end{equation}
\end{remark}

\begin{remark}[Notation]If $\U$ is a full gauge transformation starting at $\F$, and
transforming $\F$ into $\F'=\U \trr \F$, we  {use the notation}:
$
\F \ra{\U} \F' .
$ {By construction it clearly follows that $\F \ra{\id_\F} \F$; see Def. \ref{idf}.}
\end{remark}
\begin{lemma}\label{composition}{Consider a sequence of full gauge transformations $\F \ra{\U} \F' \ra{\U'} \F''$. A full gauge transformation $\U'*\U$, starting in $\F$, can be defined. Its underlying 1 and 2-squares in $\Gc$  are obtained by vertically composing the 1- and 2-squares in $\Gc$ of $\U$ and $\U'$, in the obvious way. {Therefore the squares in $\Gc$ making $\U'$ will be put under the squares in $\Gc$ making $\U$ in \eqref{vertcomp}}. Then ${\U'*\U}$ connects $\F$ to $\F''$; i.e. $\F \ra{\U'*\U} \F''$.}
\end{lemma}
\begin{proof} {\bf (Sketch)} {We must prove that $\big((\U'*\U) \trr (\F^1)\big)(t)={\F''}^1(t)$ for each $t \in L^1$, and that $\big((\U'*\U) \trr (\F^2)\big)(P)={\F''}^2(P)$, for each $P \in L^2$.}
 This is trivial to verify for edges, and for plaquettes attaching along constant maps. Otherwise,  cf. Fig. \ref{full_trans}.  Put the squares of the full gauge
 transformation $\U'$ on the bottom of the ones of $\U$. {This yields a fake-flat 2-gauge configuration $\M$, defined on the 2-disk $\Sigma$ made out of the top and lateral faces of $P\times [0,2]$, where $[0,2]$ has 0-cells at $0,1$ and 2.  There are two
 different ways to explicitly compute the 2D holonomy of $\M$ along $\Sigma$. They must yield the same element of $E$; {see Thm. \ref{independence1}}.} a)  Either we firstly multiply the
 squares standing on top of each other, and then compose with the top
 2-disk, and in this case the {result} will be
 $\big((\U'*\U) \trr (\F^2)\big)(P)$. Or b) compose $\U$ with $\F$ and only
 after that compose with $\U'$; and then  the {result}
  will be  ${\big(\U'  \trr (\U \trr \F^2)\big)(P)=  \big(\U' \trr  {\F'}^2\big)(P)={\F''}^2(P).}$

 {Cf. Rem. \ref{concrete}, if we put the squares of $\U'$ under those of $\U$ {in \eqref{compconcr}, then} the statement of the lemma also follows from the interchange law for the vertical and horizontal {compositions} in $\D^2(\Gc)$; see Rem. \ref{il}.}
\end{proof}

{Cf. \eqref{transfor}. By using the vertical reverse of 2-squares in $\Gc$,  we conclude that full gauge transformations can be
reversed.  Namely if we have $\F \ra{\U} \F'$, then $\U^{-1}$, obtained by applying vertical reverses to the 1-squares $U_1(x)$, $x \in L^0$, and  the 2-squares $U_2(t)$, $t \in L^1$, is such that $\F' \ra{\U^{-1}} \F$. Also $\U*\U^{-1}=\id_{\F'}$ and $\U^{-1}*\U=\id_\F$; see Def. \ref{idf}. Therefore we have the following result.}

\begin{Theorem}[The groupoid $\Theta^\#(M,L,\Gc)$ of fake-flat 2-gauge
configurations and full gauge transformations]\label{group-gauge-transf}
Let $(M,L)$ be a 2-lattice and $\Gc$  a crossed module.  We have a groupoid {$\Theta^\#(M,L,\Gc)$, whose objects are the
 fake-flat 2-gauge configurations $\F \in \Theta(M,L,\Gc)$.} The
morphisms are the full gauge transformations $\F \ra{\U} \F'$.
\end{Theorem}
\begin{remark}[Algebraic topological definition of $\Theta^\#(M,L,\Gc)$ -- following Brown and Higgins]\label{bh2} Recall (Thm. \ref{d2dh}) that we have a bijection $\F \mapsto f_\F$,  between fake-flat 2-gauge configurations $\F$ and crossed module maps $f_\F\colon \Pi_2(M^2,M^1,M^0) \to \Gc.$ As mentioned in Rem. \ref{bh}, there exists a relation of homotopy between crossed module maps $f,f'\colon \Gc'\to \Gc$, where $\Gc$ and $\Gc'$ are crossed modules of groupoids, discussed in \cite{BrownIcen,BROWNHOL,martins_porter}. This relation is a particular case of  homotopy of  crossed complex maps \cite[\S 7.1.vii, \S 9.3]{brown_higgins_sivera}, \cite{Brown_tensor}. By using the notation in \cite[\S 9.3.i]{brown_higgins_sivera}, given crossed modules $\Gc'$ and $\Gc$, we have a groupoid ${\rm CRS_1}(\Gc',\Gc)$ of crossed module maps {$\Gc' \to \Gc$} and homotopies between them. When $\Gc'=\Pi_2(M^2,M^1,M^0)$, where $M$ a CW-complex, and $\Gc$ is a crossed module of groups, {a homotopy  connecting $f$ and $f'$ boils down to a full gauge transformation $f\ra{\U} f'$.} For discussion see \cite{martins_porter,martins_cw_complex,martins_CW_complex2}. Hence   $\Theta^\#(M,L,\Gc)={\rm CRS}_1(\Pi_2(M^2,M^1,M^0), \Gc)$ in \cite[\S 9.3.i]{brown_higgins_sivera}.
\end{remark}

\begin{remark}[A 2-groupoid of fake-flat  2-gauge configurations, full gauge transformations and 2-fold gauge transformations]\label{bh3} The groupoid  $\Theta^\sharp(M,L,\Gc)$
is part of a
more general construction. Let $\Gc,\Gc'$ be crossed module of groupoids. By considering 2-fold homotopies between crossed module homotopies (see \cite{BrownIcen,BROWNHOL,martins_porter} and \cite[\S 9.3.i]{brown_higgins_sivera}), {we can furthermore define a 2-groupoid ${\rm CRS}_2(\Gc',\Gc)$, whose objects are crossed module maps $\Gc' \to \Gc$, 1-morphisms are homotopies between 2-crossed module maps, and 2-morphisms are 2-fold homotopies between homotopies. Explicit formulae are in \cite{martins_porter,martins_cw_complex,martins_gohla,martins_CW_complex2}.  This leads to a notion of 2-fold gauge transformation between full gauge transformations, prominent in higher gauge theory \cite{martins_picken,baez_schreiber,schreiber_waldorf2}.
These  2-fold gauge transformation between full gauge transformations do not appear to have
large importance for this paper.} However they have prime importance
  for addressing  algebraic topology descriptions of 
higher gauge theory invariants of
  manifolds (namely Yetter invariant \cite{yetter_tqft,porter_tqft}), as
  explained in \cite{martins_porter,martins_cw_complex}.
\end{remark}
\subsubsection{{Full gauge transformations preserve 2D holonomy along embedded 2-spheres}}\label{ph1}
{Full gauge transformations preserve the 2D holonomy of fake-flat 2-gauge configurations along embedded surfaces. Let us address how to prove this using some basic algebraic topology, closely  following the work of Brown and Higgins \cite{Brown_tensor,brown_higgins_cubes,brown_2dvk,brown_classifying}. Our proof is done  in very identical lines to the one of \cite{martins_picken}, which was done for the case of differential-geometric 2-connections.} 
 
We temporarily denote the fundamental crossed module of a CW-complex $X$ by $\Pi_2(X)=\Pi_2(X,X^1,X^0)$. If $X$ is a CW-complex, let $X \times [0,1]$ be the product CW-complex, where $[0,1]$ is given the obvious CW-decomposition with 0-cells at $0$ and $1$. If $\F$ is a fake-flat 2-gauge configuration in a 2-lattice $(M,L)$, we hence denote the {(Thm. \ref{d2dh})} 2D parallel transport 2-functor of $\F$ by $(\Psi_\F,\Phi_\F)\colon \Pi_2(M^2) \to \Gc$.

The main result underpinning our discussion is the following interpretation of full gauge transformations between fake-flat 2-gauge configurations. It essentially appears in \cite[\S 9.3.i, \S 9.7 and \S 9.8]{brown_higgins_sivera} and \cite{Brown_tensor,brown_classifying}, in the more general case of crossed complexes. 

\begin{lemma}Let $(M,L)$ be a 2-lattice. Let $\Gc=(\dG\colon E \to G,\trr)$ be a group crossed module.
Let $\F \ra{\U} \F'$ be a full gauge transformation connecting $\F$ and $\F'$.  We then have a crossed module map $H_\U\colon \Pi_2(M^2 \times [0,1]) \to \Gc$,  making the diagram below {(in the category of crossed modules)} commute: %: (\cite[Prop. 9.8.3] {brown_higgins_sivera}):
\begin{equation}\label{keycom}
 \xymatrix{  &  \Pi_2(M^2)\ar[drr]^{ (\Psi_\F,\Phi_\F)}\ar[d]_{i_1}\\
& \Pi_2\big(M^2 \times [0,1]) \ar[rr]|>>>>>>>>>>>{H_\U} && \Gc\\
&  \Pi_2(M^2)\ar[urr]_{ (\Psi_{\F'},\Phi_{\F'})}\ar[u]^{i_0}}
\end{equation}
Here $i_0$ and $i_1$ are induced by  $m\in M^2 \mapsto (m,0)\in M^2 \times [0,1]$ and $m\in M^2 \mapsto (m,1)\in M^2 \times [0,1]$.   
\end{lemma}
\begin{proof}
Let us for simplicity assume that the product CW-decomposition of $M^2 \times [0,1]$ is a 2-lattice $J$. The general case in analogous. Recall the construction of the usual CW-decomposition of $M^2 \times [0,1]$, where {$M^2\times \{0\}$ and $M^2 \times \{1\}$} embed as subcomplexes, and we have an additional $(i+1)$-cell $c \times [0,1]$ of $M^2 \times [0,1]$ for each $i$-cell $c$ of $M^2$. {(See \cite[Page 523]{hatcher}.)}

{Cf. the discussion in \S \ref{2fgroupoid}, particularly the construction of $\F'=\U \trr \F$, and Fig. \ref{full_trans}.} The fake-flat 2-gauge configurations $\F$, $\F'$, and the full gauge transformation $\U$, can together be assembled to yield a fake-flat 2-gauge configuration $\M$ of $(M^2 \times [0,1], J)$. The restriction of $\M$ to $M\times \{1\}$ is obtained from $\F$, and the restriction of $\M$ to  $M\times \{0\}$ is obtained from $\F'$. Finally the restriction of $\M$ to the 1- and 2-cells $v \times [0,1]$ (where $v \in L^0$) and $t \times [0,1]$ (where $t \in L^1$) is obtained from $U_1(v)$ and $U_2(t)$; for conventions on how to do this see the discussion in  \S \ref{2fgroupoid}.

We have a 3-cell $P\times I$ of $J$ for each 2-cell $P \in L^2$. And $\M$ is 2-flat along $P \times I$, given the explicit construction of the value of $\F'=\U \trr \F$ at $P$. Since there are no more 3-cells in $M^2 \times [0,1]$, we hence conclude that $\M$ is a 2-flat 2-gauge configuration in $(M \times [0,1],J)$. We now just need to apply Thm \ref{2-flat-conf-theorem} to $\M$.  Clearly 
$H_\U=(\Psi_\M,\Phi_\M)\colon \Pi_2\big(M^2 \times [0,1])\to \Gc$ makes the diagram in \eqref{keycom} commute. \end{proof}
\begin{remark}A stronger result can be proved, and is  implicit in \cite{Brown_tensor} and \cite[Chapter 9]{brown_higgins_sivera}.  Namely, there is a  one-to-one correspondence between full-gauge transformations $\F  \ra{\U} \F'$ and maps $H_\U\colon \Pi_2(M^2 \times [0,1]) \to \Gc$, making \eqref{keycom} commute. This can be inferred by combining the beginning of \cite[\S 9.3.i]{brown_higgins_sivera} with Thm \cite[9.8.1]{brown_higgins_sivera}. 
 
\end{remark}

We state the result concerning invariance of 2D holonomy under full gauge transformations for a
2-sphere cellularly embedded in a 2-lattice, only; {see \S \ref {2scat}, \S \ref{sc}.} This is the case whose behaviour under full gauge transformations is the neatest and it is the generality needed to formulate higher Kitaev models in \S\ref{HamCalc}.

\begin{Theorem}\label{main4}{Let $\Gc=(\dG\colon E \to G, \trr)$ be a crossed module of groups.}
Let $(M,L)$ be a 2-lattice. Let $\Sigma$ be a 2-sphere  cellularly embedded in $M$. Let $\F$ be a fake-flat 2-gauge configuration {on $(M,L)$.} Let $\U$ be a full gauge transformation starting in $\F$. Let $v \in \Sigma$, to be a 0-cell of $M$. {Let  $g_v\in G$ be the element of $G$ associated to $U_1(v)$; see Def. \ref{fgt}}. Then: $${\hol}_v^2(\U \trr \F , \Sigma,L ) \; = \;  g_v \trr {\hol}_v^2( \F ,\Sigma,L ).$$
\end{Theorem}
\begin{proof} We strongly use the previous lemma, and resume the notation there introduced. 

Let $\F'= \U \trr \F$. Cf. \eqref{keycom}, put $H_\U=(H_\U^2,H_\U^1)$. Let $\ivb(\Sigma) \in \pi_2(M^2,v){\subset \pi_2(M^2,M^1,v)}$ be as in Def. \ref{defivb}. 
By Def. \ref{2scc}, {we have that:}
$${\hol}_v^2( \F , \Sigma,L ) =\Psi_{\F}(\ivb(\Sigma))=H_\U^2\big({i_1}(\ivb(\Sigma))\big) \textrm{ and }
{\hol}_v^2(\U \trr \F , \Sigma,L ) =\Psi_{\F'}(\ivb(\Sigma))=H_\U^2\big({i_0}(\ivb(\Sigma))\big).
  $$
{Let $\gamma_v$ be the following path in {$M^2 \times [0,1]$,} connecting {$(v,0)$ to $(v,1)$}:} 
$${t \in [0,1] \mapsto {(v,t)} \in {M^2 \times [0,1]}}.$$
Then, passing to the correspondent element $[\gamma_v]$ in the underlying groupoid $\pi_1\big ( ( M^2 \times [0,1])^1,  (M^2 \times [0,1])^0 \big)$ of {the crossed module $\Pi_2(M^2 \times [0,1])$,} it holds that:
$${i_0(\ivb(\Sigma))= [\gamma_v] \trr (i_1(\ivb(\Sigma)), \textrm{ in } \pi_2\big(M^2 \times [0,1], ( M^2 \times [0,1])^1,  (M^2 \times [0,1])^0 \big) . }$$
 By construction  we have that $g_v=H_\U^1([\gamma_v])$. Cf. \eqref{keycom}, it hence follows that:
\begin{align*}
 {\hol}_v^2( \F' , \Sigma,L ) &=\Psi_{\F'}(\ivb(\Sigma))= H_\U^2\big({i_0}(\ivb(\Sigma))\big)\\
&=H_\U^2\big( [\gamma_v] \trr ({i_1}(\ivb(\Sigma))  \big)\\
&=H_\U^1\big( [\gamma_v]\big) \trr H_\U^2\big(({i_1}(\ivb(\Sigma))  \big)\\
&=g_v \trr  \Psi_{\F}\big(\ivb(\Sigma)\big)=g_v \trr {\hol}_v^2( \F ,\Sigma,L ).
\end{align*}

\end{proof}
\subsubsection{Groupoid $\Theta^\#_{{\rm flat}}(M,L,\Gc)$ of 2-flat 2-gauge configurations and full gauge transformations}\label{2flatgroup}
{Let $\Gc=(\dG \colon E \to G,\trr)$ be a crossed module of groups. Let $(M,L)$ be a 2-lattice}. Recall the definition of a 2-flat 2-gauge configuration in \S \ref{2flatn} and details therein. Let $\F$ be a fake-flat 2-gauge configuration. The 2D holonomy ${\rm Hol}^2_v(\F,\bound(b),L)$ of $\F$  along the boundary $\bound(b)$ of a 3-cell $b$ is invariant under full gauge transformations, in the sense of Thm. \ref{main4}. Suppose that $\F$ is 2-flat, hence that ${\rm Hol}^2_v(\F,\bound(b),L)=1_E$, for each $b \in L^3$. Since $G$ acts on $E$ by automorphisms, if  $\U$ is {any full gauge transformation, starting in $\F$,} it {follows} that  ${\rm Hol}^2_v(\U \trr \F,\bound(b),L)=1_E$, for each $b \in L^3$. Hence full gauge transformations transform 2-flat 2-gauge configurations into 2-flat 2-gauge configurations.

In particular, the groupoid $\Theta^\sharp(M,L,\Gc)$ of fake-flat 2-gauge
configurations and full gauge transformations of Thm \ref{group-gauge-transf} has a full subgroupoid $\Theta^\sharp_{{\rm flat}}(M,L,\Gc)$, whose objects are the 2-flat 2-gauge transformations.
\begin{remark}[Algebraic topological definition of $\Theta^\#_{{\rm flat}}(M,L,\Gc)$ -- following Brown and Higgins]\label{bh6} Cf. Rem. \ref{bh2}. Recall (Thm. \ref{2-flat-conf-theorem}) that we have a bijection $\F \mapsto f_\F$,  between 2-flat 2-gauge configurations $\F$ and crossed module maps $f_\F\colon \Pi_2(M,M^1,M^0) \to \Gc.$  {Cf. \cite[\S 9.3.i]{brown_higgins_sivera},} given crossed modules $\Gc'$ and $\Gc$ we have a groupoid ${\rm CRS_1}(\Gc',\Gc)$ of crossed module maps {$\Gc' \to \Gc$} and homotopies between them. When $\Gc'=\Pi_2(M,M^1,M^0)$, where $M$ a CW-complex,  $\Gc$ is a group crossed module, {and $\F,\F'$ are 2-flat 2-gauge configurations,  a homotopy $H$, connecting $f_{\F}$ to $f_{\F'}$, boils down to a full gauge transformation $\F\ra{\U} \F'$. Hence   $\Theta^\#_{{\rm flat}}(M,L,\Gc)\cong {\rm CRS}_1(\Pi_2(M,M^1,M^0), \Gc)$ in \cite[\S 9.3.i]{brown_higgins_sivera}.}
\end{remark}

\subsection{Gauge operators on fake-flat 2-gauge  configurations}
\label{gaugeaction}

{Let $\Gc=(\dG \colon E \to G,\trr)$ be a crossed module of groups.} {We now finally define a left-action {$\trrr$} of the group  of gauge operators
${\cal  T}(M,L,{\Gc})$ of \S \ref{2gaugetrans}
on the set of fake-flat 2-gauge configurations
$\Theta(M,L,\Gc)$. Our main tool is the groupoid $\Theta^\sharp(M,L,\Gc)$ of fake-flat 2-gauge
configurations and full gauge transformations between them; see Thm \ref{group-gauge-transf}. We will define a {left-action} whose action-groupoid  is isomorphic to $\Theta^\sharp(M,L,\Gc)$.}

The main observation is that given a {fake-flat 2-gauge} configuration {$\F=(\F^2\colon L^2 \to E,\F^1\colon L^1 \to G )$,}  the set of full gauge
transformations starting in $\F$ can be put in one-to-one
correspondence with elements
$(\eta,u)\in {\cal T}(M,L,{\Gc})
  = {\cal  E}(M,L,{\Gc}) \rtimes_\bullet  {\cal V}(M,L,{\Gc})$. {Given $\F\in \Theta(M,L,\Gc)$ and $(\eta,u)\in {\cal  T}(M,L,{\Gc})$,} we have a full gauge transformation $\U_{(\eta,u,\F)}=\U=(U_2,U_1)$,  starting at $\F$, defined as:
\begin{equation}\label{conventions}
\begin{split}
{v \in L^0} &\stackrel{U_1}{\longmapsto}  \left(\begin{CD} &* \\
              &@AA u(v)A  \\
              &*  \end{CD}\,\,\,\,\,\,\,\,\right) \\
\left (v \ra{t} v' \right)\in L^1 & \stackrel{U_2}{\longmapsto}  \quad \left  (\,\quad \begin{CD} &* @>  \qquad  {\F^1}(t) \qquad   >> &*\\
              &@A u(v)AA \hskip-3.1cm \scriptstyle{\eta(t)}  &@AA u(v')A\\
              &* @>> \qquad  \U \trr {\F^1}(t)  \qquad  > &* \end{CD} \quad\,\,\,  \right);\textrm{ for } \U\trr {\F^1}(t)=\d\big (\eta(t) \big) \, u(v)\, {\F^1}(t)\, u(v')^{-1} .
              \end{split}
\end{equation}
\begin{lemma}\label{lem:9}
Let $(\eta,u),(\eta',u')\in {\cal T}(M,L,{\Gc})$. Let {$\F \in \Theta(M,L,\Gc)$.} We have:
\begin{equation}
  {\U_{\big( (\eta,u)(\eta',u'), \F\big)}=   \U_{\big(  \eta\, u \bullet \eta', u\,u' , \F\big)}=  \U_{ \big( \eta, u, \U_{(\eta',u',\F)}\trr \F\big) }  * \U_{ \big( \eta', u',   \F\big)}.}
\end{equation}
\end{lemma}
\noindent{{N.B.}: See \S \ref{2gaugetrans} for conventions on the product  $(\eta,u)(\eta',u')=  (\eta\, u \bullet \eta', u\,u')$ of  $(\eta,u),(\eta',u')\in {\cal T}(M,L,{\Gc}).$ The composition of $*$  of full gauge transformation is made explicit in Lem. \ref{composition}.}
\begin{proof} {This follows by construction, by looking at the conventions \eqref{conventions} for
$\U_{(\eta,u,\F)}$. Just compare the explicit definition of the group operation in ${\cal T}(M,L,{\Gc})$ in \eqref{semidef} with the explicit form of the vertical composition of squares in $\Gc$ in \eqref{vertcomp}. The latter yields the composition $*$ of full gauge transformations in  Lem. \ref{composition}.}
\end{proof}

An operation {$\trrr$} of the group ${\cal T}(M,L,{\Gc})$  %of gauge operators
on  $\Theta(M,L,\Gc)$ {can then be} defined by:
\begin{equation}\label{eq:TonP}
(\eta,u) \trrr \F  \; =  \;  \U_{(\eta,u,\F)}\trr \F .
\end{equation}
{By Lem. \ref{lem:9}, $\trrr$} is indeed a left-action of ${\cal T}(M,L,{\Gc})$ on
$\Theta(M,L,\Gc)$ and  $\Theta^\sharp(M,L,\Gc)$ is its action groupoid.

By construction of $\U_{(\eta,u,\F)}$ and Thm. \ref{main4}, it follows that:
\begin{Theorem}\label{main5}
 Let $(M,L)$ be a 2-lattice. Let $\Sigma$ be a 2-sphere  cellularly embedded in $M$. Let $v \in \Sigma \cap M^0$. Let $\F$ be a fake-flat 2-gauge configuration {in $(M,L)$.} Let $(\eta,u) \in {\cal T}(M,L,{\Gc})$.
  Then: $${\hol}_v^2\big((\eta,u) \trrr \F , \Sigma,L \big) \; = \;  u(v) \trr {\hol}_v^2( \F ,\Sigma,L ).$$
\end{Theorem}
{Cf. \ref{2flatgroup}, in particular $\trrr$ restricts to an action of ${\cal T}(M,L,\Gc)$ on the set of 2-flat 2-gauge configurations.}

\newcommand{\Tc}{{\cal T}}
\newcommand{\Hc}{{\cal H}}
\newcommand{\TcM}{\Tc (M,L,\Gc ) }  % group of g ops
\newcommand{\HcM}{\Hc (M,L,\Gc ) }  % Hilbert space

\newcommand{\spike}[3]{\setlength{\unitlength}{.1mm} \begin{picture}(110,110)
    \put(10,0){\vector(0,1){100}}
      \put(-7,45){$#1$}
    \put(10,0){\vector(1,0){100}}
      \put(50,45){$#2$}
    \put(10,100){\vector(1,0){100}}
    \put(110,0){\vector(0,1){100}}
      \put(114,45){$#3$}
  \end{picture} }

\newcommand{\spikes}[6]{\setlength{\unitlength}{.1mm} \begin{picture}(110,210)
    \put(10,100){\vector(0,1){100}}
      \put(-7,145){$#4$}
    \put(10,100){\vector(1,0){100}}
      \put(50,145){$#5$}
    \put(10,200){\vector(1,0){100}}
    \put(110,100){\vector(0,1){100}}
      \put(114,145){$#6$}
             \put(10,0){\vector(0,1){100}}
      \put(-7,45){$#1$}
    \put(10,0){\vector(1,0){100}}
      \put(50,45){$#2$}
    \put(10,100){\vector(1,0){100}}
    \put(110,0){\vector(0,1){100}}
      \put(114,45){$#3$}
  \end{picture} }

\section{The Hamiltonian models}\label{Hamiltonian}

\subsection{A Hamiltonian model for higher gauge theory}\label{HamCalc}
{In this subsection, we fix a manifold $M$, a 2-lattice decomposition $(M,L)$ of $M$ (Def. \ref{2-lattice}), and a  crossed module of groups $\Gc=(\dG\colon E \to G, \trr)$; Def. \ref{cm}.
We suppose that $M$ is compact (thus that $L$ is finite) and that $\Gc$ is finite, meaning that both $G$ and $E$ are finite {groups}. Hence the set $\Theta(M,L,\Gc)$ of fake-flat 2-gauge configurations is finite. Note that $\Theta(M,L,\Gc)$  is non-empty, as the naive vacuum
 is in $\Theta(M,L,\Gc)$; see \S\ref{ff2gc}.}

Recall from \S\ref{2gaugetrans}  the group
$
\TcM \; =
{\cal E}(M,L,{\Gc}) \rtimes_\bullet  {\cal  V}(M,L,{\Gc})
$
of gauge operators.
We have  \S \ref{gaugeaction} a left-action  $\trrr$
of the group of gauge operators on $\PhiM$,
preserving 2D holonomy, as in  Thm. \ref{main5}.

\begin{definition}[Hilbert space $\mathcal{H}(M,L,\mathcal{G})$]\label{hs}The Hilbert space for Hamiltonian higher gauge theory is the free vector space $\mathcal{H}(M,L,\mathcal{G})  = \CC \Theta(M,L,\Gc)$
 on the set of fake-flat 2-gauge configurations, with the unique inner product
$\langle -, - \rangle$
 that renders different {fake-flat 2-gauge configurations orthonormal.}
\end{definition}
{The group algebra  $\CC \TcM$ of $\TcM$ has a representation on $\mathcal{H}(M,L,\mathcal{G})$, obtained by linearising the action of  $\TcM$ on  $\Theta(M,L,\Gc)$. Given $m \in \CC \TcM$, we denote the corresponding linear operator  by $\widehat m\colon  \mathcal{H}(M,L,\mathcal{G})\to \mathcal{H}(M,L,\mathcal{G}).$ Thus  $\widehat m(\F)= m \trrr \F$, for each fake-flat 2-gauge configuration $\F$.}

{By construction, if $m \in \TcM$, then the operator $\widehat m \colon \mathcal{H}(M,L,\mathcal{G})\to \mathcal{H}(M,L,\mathcal{G})$ is unitary, i.e. $\widehat{m}^{\dagger}=\widehat{m}^{-1}=\widehat{m^{-1}}$, where $\dagger$ denotes Hermitian adjoint.}

\subsubsection{Vertex and edge gauge spikes $U^g_v$ and  $ U^e_t$}\label{exp}

\begin{definition}[Edge gauge spikes and vertex gauge spikes] A  vertex gauge spike, supported in
  $v \in L^0$, is a gauge operator $(\eta,u)\in\TcM $ such that $\eta(t)=1_E$
  for each $t \in L^1$, {and such that for each $w \in L^0$ it holds that $u(w)=1_G$, unless $w=v$.}
  Analogously, given an edge $t\in L^1$, a gauge spike
  supported in $t$ is a gauge operator $(\eta,u)$ such that $u(v)=1_G$
  for each $v \in L^0$, and {such that for each $s \in L^1$, it holds that $\eta(s)=1_E$,} unless $s=t$.
  
{For a vertex $v\in L^0$} and $g \in G$, we let $U_v^g$ be the unique vertex gauge spike
  supported in $v$ such that $u(v)=g$.  {For $t\in L^1$} and $e \in E$,
  we let $U_t^e$ be the unique edge gauge spike, supported in $t$, such that $\eta(t)=e$.

  The  linear operators $\widehat{ U_v^g},\widehat{ U_t^e}\colon \mathcal{H}(M,L,\mathcal{G})\to \mathcal{H}(M,L,\mathcal{G}) $ are also called vertex and edge gauge spikes.
\end{definition}

\begin{lemma}\label{comm1}
{(I) Vertex gauge spikes supported in different vertices commute, and edge gauge spikes supported in different edges commute. I.e., given $v,v' \in L^0$, with $v\neq v'$, and $t,t' \in L^1$, with $t \neq t'$ (including the case when $t$ and $t'$ share a vertex), it holds that,  given $e,e' \in E$ and $g,g' \in G$:}
$${ [U^g_v, U^{g'}_{v'}]=0\qquad \textrm { and } \qquad [U^e_t, U^{e'}_{t'}]=0. }$$
(II) {For each vertex $v\in L^0$ and each edge  {$t\in L^1$,} then given any $g,h \in G$ and any
$e,f \in E$, we have:}
  $${ U^g_v U^h_v  = U^{gh}_v \qquad \textrm{ and } \qquad    U^e_t U^f_t   = U^{ef}_t.  }$$ 
(III) {If $t=\big(\sigma(t) \ra{t} \tau(t)\big) \in L^1$, including the case when $\sigma(t)=\tau(t)$, and $v\in L^0$ is such that  $v\neq \sigma(t)$, then given $g \in G$ and $e \in E$, it follows that
$U^g_v$ and  $ U^e_t$ commute: $[U^g_v,U^e_t]=0$}.\\
(IV) {Given an edge  $t=\big(\sigma(t) \ra{t} \tau(t)\big) \in L^1$, including the case $\sigma(t)=\tau(t)$,   $g \in G$ and $e \in E$, we have:}
$$
  U^e_t U^g_{\sigma(t)}   =
     U^g_{\sigma(t)}  U^{g^{-1} \trr e}_t .
$$
(V) Moreover, any gauge operator can be obtained as a product of vertex and edge gauge spikes.
 \end{lemma}
\begin{proof}
This all follows immediately from the definition of the group of gauge
\op s as a semidirect product {in \S\ref{2gaugetrans}. Equations  \eqref{bullet} and \eqref{semidef} translate exacty to the formulae in the lemma.}
\end{proof}
\begin{remark}
  {Note that  (I--IV) also hold for the associated operators $\widehat{ U_v^g}, \widehat{ U_t^e}\colon \mathcal{H}(M,L,\mathcal{G})\to \mathcal{H}(M,L,\mathcal{G})$, since they are constructed from a representation of the group ${\cal  T}(M,L,{\Gc})$ on the Hilbert space  $\mathcal{H}(M,L,\mathcal{G})$.}
\end{remark}

Let us now unpack the construction in \S\ref{fgtb} and \S\ref{gaugeaction}, and give an explicit description of how vertex and edge gauge spikes act on $\mathcal{H}(M,L,\mathcal{G})  = \CC \Theta(M,L,\Gc)$. 
Let $\F=(\F^2\colon L^2 \to E\,\,,\,\,\F^1\colon L^1 \to G )$ be a fake-flat 2-gauge configuration; Def. \ref{fakeflatconf}. If $m\in \CC\TcM$, recall {$\widehat m\colon \mathcal{H}(M,L,\mathcal{G}) \to  \mathcal{H}(M,L,\mathcal{G})$ is the  operator, $\F \mapsto m \trrr \F$.} Put $\widehat m(\F)=\big(\widehat m ( \F^2)\colon L^2 \to E\,,\, \widehat m ( \F^1)\colon L^1 \to G\big)$. {Recall \eqref{conventions}:}

\noindent {\bf Action of vertex gauge spikes:} Let $v \in L^0$ and $g \in G$. Let {$\big(\s(t) \ra{t} \tau(t) \big)\in L^1$.} Then:
 $$ \big(\widehat{ U^g_v }(\F^1)\big) ( \s(t) \ra{t} \tau(t) )= \begin{cases}
                                                       g \,\F^1 \big(  \s(t) \ra{t} \tau(t) \big), \textrm{ if } v= \s(t), {\textrm{ and } \s(t) \neq \t(t)};\\
                                                        \F^1 \big(  \t(t) \ra{t} \tau(t)\,  \big) g^{-1}, \textrm{ if } v= \t(t), {\textrm{ and } \s(t) \neq \t(t);}\\
                                                       {g \F^1 \big(  \t(t) \ra{t} \tau(t)\,  \big) g^{-1}, \textrm{ if } v= \t(t)=\sigma(t);}\\
                                                         \F^1 \big(  \s(t) \ra{t} \tau(t)\,  \big), {\textrm{ if } v \neq \sigma(t) \textrm{ and } v \neq \tau(t).} 
                                                      \end{cases}
  $$
Let $P\in L^2$. Let {$x_P\in \bound(P)$} be the base-point of $P$. Then:
$$ \big(\widehat{U^g_v} (\F^2)\big)( P )= \begin{cases}
                                                       g \trr \F^2 \big( P \big), \textrm{ if } v= x_P;\\
                                                         \F^2 \big( P \big), \textrm{ if } v \textrm{ is not the base-point } x_P \textrm{ of }   P.
                                                      \end{cases}$$
\noindent {\bf Action of edge gauge spikes:} Let $(\s(\gamma)\ra{\gamma}  \t(\gamma)) \in L^1$ and $e \in E$. Let $(\s(t) \ra{t} \tau(t)) \in L^1$, then, {cf. \eqref{conventions}:}
 $$ \big(\widehat{ U_\gamma^e} (\F^1)\big) ( \s(t) \ra{t} \tau(t) )= \begin{cases}
                                                       \d(e) \,\F^1 \big(  \s(t) \ra{t} \tau(t) \big), \textrm{ if the edges } (\s(t) \ra{t} \tau(t)) \textrm{ and }  (\s(\gamma)\ra{\gamma}  \t(\gamma)) \textrm{ are the same}; \\\
                                                        \F^1 \big(  \s(t) \ra{t} \tau(t)\,  \big), \textrm{ if }  (\s(t) \ra{t} \tau(t)) \neq (  \s(\gamma)\ra{\gamma}  \t(
                                                        \gamma) ).
                                                      \end{cases}
$$

Let $P\in L^2$. Let $x_P$ be the base-point of $P$.  Some bits of notation are necessary to describe  $\big(\widehat{ U_\gamma^e} (\F^1)\big) ( P).$
Let $\dL^Q(P)=\big(x_P \ra{t_1^{\theta_1}t_2^{\theta_2}...t_n^{\theta_n} }x_P\big)$ be the quantised boundary of $P$; Def.  \ref{qbp}. Put $t_i=(x_i \ra{t_i} y_i)$,  $i=1,\dots,n$.  Given that the attaching map $\psi_P^2\colon S^1 \to M^1$ of the 2-cell $\overline{c_P^2}=P$ is either constant or an embedding, an arbitrary edge $\gamma=(\s(\gamma)\ra{\gamma}  \t(\gamma))\in L^1$ can appear in the list $(x_i \ra{t_i} y_i),$ where $t \in \{1,\dots,n\}$, at most once.

\begin{definition}Let  $\dL^Q(P)=\big(x_P \ra{t_1^{\theta_1}t_2^{\theta_2}...t_n^{\theta_n} }x_P\big)$.
We say that  the edge $\gamma=(\s(\gamma)\ra{\gamma}  \t(\gamma))\in L^1$:
\begin{align*} \textrm{ is not incident to } {\bound(P)}, &\textrm{ if } (\s(\gamma) \ra{\gamma }\t(\gamma)) \textrm { is not in the list } (x_i \ra{t_i} y_i), i \in \{1,\dots,n\};\\
 \textrm{ is positively incident to } {\bound(P)}, &\textrm{ if }  \big( \s(\gamma) \ra{\gamma }\t(\gamma)\big)= (x_i \ra{t_i} y_i)  \textrm{, for some } i \in \{1,\dots ,n\}, \textrm{ and } \theta_i=1;\\
 \textrm{ is negatively incident to } {\bound(P)}, &\textrm{ if }  \big( \s(\gamma) \ra{\gamma }\t(\gamma)\big)= (x_i \ra{t_i} y_i)  \textrm{, for some } i \in \{1,\dots ,n\}, \textrm{ and } \theta_i=-1.
 \end{align*}
 \end{definition}
 Suppose $\psi_P^2\colon S^1 \to M^1$ is an embedding.  Given a vertex $x \in \bound(P)$, with $x \neq x_P$,  there exists a unique quantised path $x_P \ra{p^+(x)} x$ made from the  $(x_i \ra{t_i} y_i)^{\theta^i}$ and contouring $\bound(P)$ in the positive direction (where orientation is provided by the attaching map of $P$).  There is another quantised path  $x_P \ra{p^-(x)} x$ contouring $\bound(P)$ in the negative direction.  We also put $p^+({x_P})=\emptyset_{x_P}$ and $p^{-}({x_P})=\emptyset_{x_P}$.

 We let $g_{p^+(x)}$ and $g_{p^-(x)}$  be the 1D holonomy of $\F^1$ along $p^+(x)$ and $p^{-}(x)$; see Def. \ref{n1}.
And then:
$$\big(\widehat{ U^e_\gamma (\F^2)}\big)( P )= \begin{cases}  \F^2 \big( P \big), \textrm{ if } P \textrm{ attaches along a constant map};\\ 
                                                        \F^2 \big( P \big), \textrm{ if } \s(\gamma) \ra{\gamma }\t(\gamma) \textrm { is not incident to } {\bound(P)};\\
                                                        { \big(g_{p^+(\sigma(\gamma))} \trr e\big)\F^2 \big( P \big), \textrm{ if }  \g \textrm { is positively incident to } {\bound(P)};} \\
                                                        {\F^2 \big( P \big)\,\,\big(g_{p^-(\sigma(\gamma))} \trr e^{-1}\big), \textrm{ if } \g \textrm { is negatively incident to } {\bound(P)}.}
                                                      \end{cases}$$

%{See \cite[III-A]{BCKMM} for a description of the action of vertex and edge gauge spikes in the particular case when the 2-lattice decomposition of $M$ is a triangulation.}
\begin{example}Consider a 2-lattice decomposition $L$ that close to a plaquette $P$ looks like figure \ref{exampleaction}, below:
\begin{figure}[H]
\centerline{\relabelbox
\epsfysize 4cm
\epsfbox{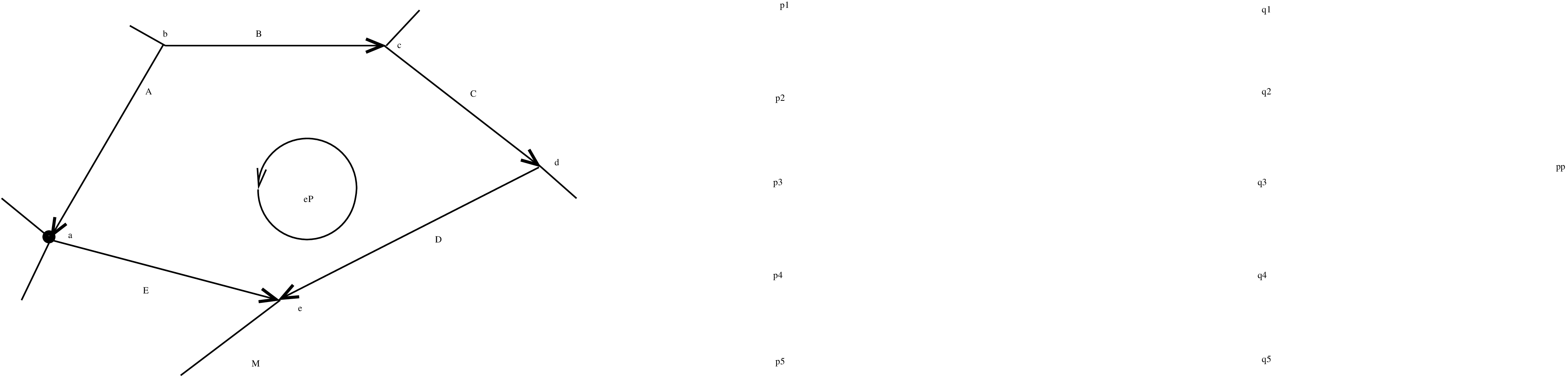}
\relabel{a}{$x_P$}
\relabel{b}{$x_1$}
\relabel{c}{$x_2$}
\relabel{d}{$x_3$}
\relabel{e}{$x_4$}
\relabel{A}{${t_1}$}
\relabel{B}{${t_2}$}
\relabel{C}{${t_3}$}
\relabel{D}{${t_4}$}
\relabel{E}{${t_5}$}
\relabel{eP}{$P$}
\relabel{p1}{$p^+(x_P)=\emptyset_{x_P}$}
\relabel{q1}{$p^-(x_P)=\emptyset_{x_P}$}
\relabel{p2}{$p^+(x_1)=t_5t_4^{-1}t_3^{-1}t_2^{-1}$}
\relabel{q2}{$p^-(x_1)=t_1^{-1}$}
\relabel{p3}{$p^+(x_2)=t_5t_4^{-1}t_3^{-1} $}
\relabel{q3}{$p^-(x_2)=t_1^{-1}t_2$}
\relabel{p4}{$p^+(x_3)=t_5t_4^{-1}$}
\relabel{q4}{$p^-(x_3)=t_1^{-1}t_2t_3$}
\relabel{p5}{$p^+(x_4)=t_5$}
\relabel{q5}{$p^-(x_4)=t_1^{-1}t_2t_3t_4$}
\relabel{M}{$\partial_L^Q(P)=t_5 t_4^{-1} t_3^{-1} t_2^{-1} t_1$}
\endrelabelbox}
\caption{A plaquette $P$ of a 2-lattice decomposition $L$. As indicated, the plaquette attaches counterclockwise. The base-point of $P$ is $x_P$. We also show the quantised paths $p^\pm(x_P),p^\pm(x_1),p^\pm(x_2),p^\pm,(x_3),p^\pm(x_4)$. \label{exampleaction}}
\end{figure}
\noindent The only edge and vertex gauge spikes which have a non-trivial action on the restriction of a fake-flat 2-gauge configuration $\F$ to $P$ are the vertex gauge spikes $\widehat{U_{x_P}^g},\widehat{U_{x_1}^g},\widehat{U_{x_2}^g},\widehat{U_{x_3}^g},\widehat{U_{x_4}^g}$, and the edge gauge spikes  $\widehat{U_{t_1}^e},\widehat{U_{t_2}^e},\widehat{U_{t_3}^e},\widehat{U_{t_4}^e},\widehat{U_{t_5}^e}$, {where $g\in G$ and $e \in E$ are arbitary. The actions} of some of these on a fake-flat 2-gauge configuration at $P$ are  below. (We can also see that fake-flatness at $P$ is preserved in these  examples.)
%%%%%%%%%%%%%%%%%%%% VP
$$
\centerline{\relabelbox
\epsfysize 4.6cm
\epsfbox{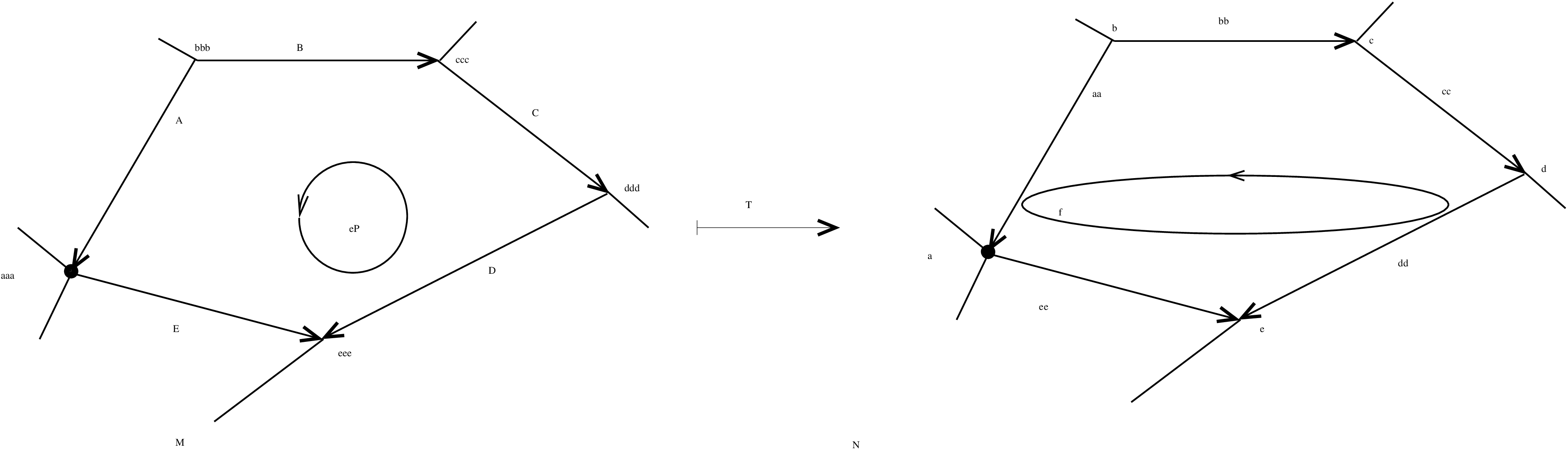}
 \relabel{a}{$x_P$}
  \relabel{aaa}{$x_P$}
%  \relabel{b}{$x_1$}
%  \relabel{c}{$x_2$}
%  \relabel{d}{$x_3$}
%  \relabel{e}{$x_4$}
 \relabel{A}{${g_1}$}
 \relabel{B}{${g_2}$}
 \relabel{C}{${g_3}$}
 \relabel{D}{${g_4}$}
 \relabel{E}{${g_5}$}
 \relabel{eP}{$e_P$}
 \relabel{M}{$\partial_{\cal G} (e_P)=g_5 g_4^{-1} g_3^{-1} g_2^{-1} g_1$}
 \relabel{aa}{${g_1g^{-1}}$}
 \relabel{bb}{${g_2}$}
 \relabel{cc}{${g_3}$}
 \relabel{dd}{${g_4}$}
 \relabel{ee}{${gg_5}$}
 \relabel{f}{$\quad \quad g \triangleright e_P$}
 \relabel{N}{$\partial_{\cal G} (g \trr e_P)=g g_5 g_4^{-1} g_3^{-1} g_2^{-1} g_1 g^{-1}$}
 \relabel{T}{$\widehat{U_{x_P}^g}$}
% \relabel{p1}{$p^+(x_P)=\emptyset_{x_P}$}
% \relabel{q1}{$p^-(x_P)=\emptyset_{x_P}$}
% \relabel{p2}{$p^+(x_1)=t_5t_4^{-1}t_3^{-1}t_2^{-1}$}
% \relabel{q2}{$p^-(x_1)=t_1^{-1}$}
% \relabel{p3}{$p^+(x_2)=t_5t_4^{-1}t_3^{-1} $}
% \relabel{q3}{$p^-(x_2)=t_1^{-1}t_2$}
% \relabel{p4}{$p^+(x_3)=t_5t_4^{-1}$}
% \relabel{q4}{$p^-(x_3)=t_1^{-1}t_2t_3$}
% \relabel{p5}{$p^+(x_4)=t_5$}
% \relabel{q5}{$p^-(x_4)=t_1^{-1}t_2t_3t_4$}
\endrelabelbox}
$$

%%%%%%%%%%%%%%%%%%%% Vx1

$$
\centerline{\relabelbox
\epsfysize 4.6cm
\epsfbox{ex2.eps}
 \relabel{a}{$x_P$}
  \relabel{aaa}{$x_P$}
  \relabel{b}{$x_1$}
  \relabel{c}{$x_2$}
  \relabel{d}{$x_3$}
  \relabel{e}{$x_4$}
 \relabel{aaa}{$x_P$}
  \relabel{bbb}{$x_1$}
  \relabel{ccc}{$x_2$}
  \relabel{ddd}{$x_3$}
  \relabel{eee}{$x_4$}
  \relabel{A}{${g_1}$}
 \relabel{B}{${g_2}$}
 \relabel{C}{${g_3}$}
 \relabel{D}{${g_4}$}
 \relabel{E}{${g_5}$}
 \relabel{eP}{$e_P$}
 \relabel{M}{$\partial_{\cal G} (e_P)=g_5 g_4^{-1} g_3^{-1} g_2^{-1} g_1$}
 \relabel{aa}{${g g_1}$}
 \relabel{bb}{${g g_2}$}
 \relabel{cc}{${g_3}$}
 \relabel{dd}{${g_4}$}
 \relabel{ee}{${g_5}$}
 \relabel{f}{$ \quad \quad e_P$}
 \relabel{N}{$\partial_{\cal G} ( e_P)= g_5 g_4^{-1} g_3^{-1} (gg_2)^{-1} (gg_1)$}
 \relabel{T}{$\widehat{U_{x_1}^g}$}
% \relabel{p1}{$p^+(x_P)=\emptyset_{x_P}$}
% \relabel{q1}{$p^-(x_P)=\emptyset_{x_P}$}
% \relabel{p2}{$p^+(x_1)=t_5t_4^{-1}t_3^{-1}t_2^{-1}$}
% \relabel{q2}{$p^-(x_1)=t_1^{-1}$}
% \relabel{p3}{$p^+(x_2)=t_5t_4^{-1}t_3^{-1} $}
% \relabel{q3}{$p^-(x_2)=t_1^{-1}t_2$}
% \relabel{p4}{$p^+(x_3)=t_5t_4^{-1}$}
% \relabel{q4}{$p^-(x_3)=t_1^{-1}t_2t_3$}
% \relabel{p5}{$p^+(x_4)=t_5$}
% \relabel{q5}{$p^-(x_4)=t_1^{-1}t_2t_3t_4$}
\endrelabelbox}
$$

%%%%%%%%%%%%%%%%% t1

$$
\centerline{\relabelbox
\epsfysize 4.6cm
\epsfbox{ex2.eps}
 \relabel{a}{$x_P$}
  \relabel{aaa}{$x_P$}
%  \relabel{b}{$x_1$}
%  \relabel{c}{$x_2$}
%  \relabel{d}{$x_3$}
%  \relabel{e}{$x_4$}
 \relabel{A}{${g_1}$}
 \relabel{B}{${g_2}$}
 \relabel{C}{${g_3}$}
 \relabel{D}{${g_4}$}
 \relabel{E}{${g_5}$}
 \relabel{eP}{$e_P$}
 \relabel{M}{$\partial_{\cal G} (e_P)=g_5 g_4^{-1} g_3^{-1} g_2^{-1} g_1$}
 \relabel{aa}{${\d(e)g_1}$}
 \relabel{bb}{${g_2}$}
 \relabel{cc}{${g_3}$}
 \relabel{dd}{${g_4}$}
 \relabel{ee}{${g_5}$}
 \relabel{f}{$\,\,\,\, \big ( g_5g_4^{-1}g_3^{-1}g_2^{-1}\trr e\big)\,  e_P$}
 \relabel{N}{$\partial_{\cal G} ( \big ( g_5g_4^{-1}g_3^{-1}g_2^{-1}\trr e\big)\,\,  e_P)=g_5 g_4^{-1} g_3^{-1} g_2^{-1} \d(e) g_1 $}
 \relabel{T}{$\widehat{U_{t_1}^e}$}
% \relabel{p1}{$p^+(x_P)=\emptyset_{x_P}$}
% \relabel{q1}{$p^-(x_P)=\emptyset_{x_P}$}
% \relabel{p2}{$p^+(x_1)=t_5t_4^{-1}t_3^{-1}t_2^{-1}$}
% \relabel{q2}{$p^-(x_1)=t_1^{-1}$}
% \relabel{p3}{$p^+(x_2)=t_5t_4^{-1}t_3^{-1} $}
% \relabel{q3}{$p^-(x_2)=t_1^{-1}t_2$}
% \relabel{p4}{$p^+(x_3)=t_5t_4^{-1}$}
% \relabel{q4}{$p^-(x_3)=t_1^{-1}t_2t_3$}
% \relabel{p5}{$p^+(x_4)=t_5$}
% \relabel{q5}{$p^-(x_4)=t_1^{-1}t_2t_3t_4$}
\endrelabelbox}
$$

%%%%% t2

$$
\centerline{\relabelbox
\epsfysize 4.6cm
\epsfbox{ex2.eps}
 \relabel{a}{$x_P$}
  \relabel{aaa}{$x_P$}
 \relabel{b}{$x_1$}
 \relabel{c}{$x_2$}
 \relabel{d}{$x_3$}
 \relabel{e}{$x_4$}
 \relabel{A}{${g_1}$}
 \relabel{B}{${g_2}$}
 \relabel{C}{${g_3}$}
 \relabel{D}{${g_4}$}
 \relabel{E}{${g_5}$}
 \relabel{eP}{$e_P$}
 \relabel{M}{$\partial_{\cal G} (e_P)=g_5 g_4^{-1} g_3^{-1} g_2^{-1} g_1$}
 \relabel{aa}{${g_1}$}
 \relabel{bb}{${\d(e) g_2}$}
 \relabel{cc}{${g_3}$}
 \relabel{dd}{${g_4}$}
 \relabel{ee}{${g_5}$}
 \relabel{f}{$ \quad \quad e_P\,g_1^{-1} \trr e^{-1}$}
 \relabel{N}{$\partial_{\cal G} (e_P\,\,g_1^{-1} \trr e^{-1})=g_5 g_4^{-1} g_3^{-1} g_2^{-1} \d(e)^{-1}g_1 $}
 \relabel{T}{$\widehat{U_{t_2}^e}$}
% \relabel{p1}{$p^+(x_P)=\emptyset_{x_P}$}
% \relabel{q1}{$p^-(x_P)=\emptyset_{x_P}$}
% \relabel{p2}{$p^+(x_1)=t_5t_4^{-1}t_3^{-1}t_2^{-1}$}
% \relabel{q2}{$p^-(x_1)=t_1^{-1}$}
% \relabel{p3}{$p^+(x_2)=t_5t_4^{-1}t_3^{-1} $}
% \relabel{q3}{$p^-(x_2)=t_1^{-1}t_2$}
% \relabel{p4}{$p^+(x_3)=t_5t_4^{-1}$}
% \relabel{q4}{$p^-(x_3)=t_1^{-1}t_2t_3$}
% \relabel{p5}{$p^+(x_4)=t_5$}
% \relabel{q5}{$p^-(x_4)=t_1^{-1}t_2t_3t_4$}
\endrelabelbox}
$$

%%%%%%%%%%%%%%%%%% t5

$$
\centerline{\relabelbox
\epsfysize 4.6cm
\epsfbox{ex2.eps}
%  \relabel{a}{$x_P$}
%  \relabel{b}{$x_1$}
%  \relabel{c}{$x_2$}
%  \relabel{d}{$x_3$}
%  \relabel{e}{$x_4$}
 \relabel{a}{$x_P$}
  \relabel{aaa}{$x_P$}
 \relabel{A}{${g_1}$}
 \relabel{B}{${g_2}$}
 \relabel{C}{${g_3}$}
 \relabel{D}{${g_4}$}
 \relabel{E}{${g_5}$}
 \relabel{eP}{$e_P$}
 \relabel{M}{$\partial_{\cal G} (e_P)=g_5 g_4^{-1} g_3^{-1} g_2^{-1} g_1$}
 \relabel{aa}{${g_1}$}
 \relabel{bb}{${g_2}$}
 \relabel{cc}{${g_3}$}
 \relabel{dd}{${g_4}$}
 \relabel{ee}{${\d(e)g_5}$}
 \relabel{f}{$\quad \quad e \, e_P$}
 \relabel{N}{$\partial_{\cal G} (e\, e_P)=\d(e) g_5 g_4^{-1} g_3^{-1} g_2^{-1} g_1$}
 \relabel{T}{$\widehat{U_{t_5}^e}$}
% \relabel{p1}{$p^+(x_P)=\emptyset_{x_P}$}
% \relabel{q1}{$p^-(x_P)=\emptyset_{x_P}$}
% \relabel{p2}{$p^+(x_1)=t_5t_4^{-1}t_3^{-1}t_2^{-1}$}
% \relabel{q2}{$p^-(x_1)=t_1^{-1}$}
% \relabel{p3}{$p^+(x_2)=t_5t_4^{-1}t_3^{-1} $}
% \relabel{q3}{$p^-(x_2)=t_1^{-1}t_2$}
% \relabel{p4}{$p^+(x_3)=t_5t_4^{-1}$}
% \relabel{q4}{$p^-(x_3)=t_1^{-1}t_2t_3$}
% \relabel{p5}{$p^+(x_4)=t_5$}
% \relabel{q5}{$p^-(x_4)=t_1^{-1}t_2t_3t_4$}
\endrelabelbox}
$$
\end{example}

\begin{remark}{{We note that if $L$ is a triangulation of $M$, then the vertex and edge gauge spikes here defined coincide with the vertex and edge gauge transformations appearing in \cite[III-A \& III-B]{BCKMM}.} }\end{remark}
\subsubsection{Vertex operators, edge operators and blob operators}\label{veb}
{Given a set $X$, we put $\# X$ to denote the cardinality of $X$.}
\begin{definition}[Vertex operators $\widehat{ A_v}$ and edge operators $\widehat{B_t}$]
Let $v\in L^0$ be a vertex and $t\in L^1$ be an edge, { of $(M,L)$}. The elements below of  the group algebra $\CC \TcM$  are called vertex and edge operators:
\begin{align} \label{de:AvBt}
   A_v   &\doteq \frac{1}{\# G} \sum_{g \in G}  U^g_v  ,
\qquad\hspace{.65in}
   B_t   \doteq \frac{1}{\# E} \sum_{e \in E}  U^e_t.   
\end{align}
The corresponding operators $\widehat{ A_v},\widehat{B_t}  \colon \mathcal{H}(M,L,\mathcal{G}) \to \mathcal{H}(M,L,\mathcal{G})$ are also called vertex and edge operators.
\end{definition}
Note that the operators $\widehat{ A_v},\widehat{B_t}  \colon \mathcal{H}(M,L,\mathcal{G}) \to \mathcal{H}(M,L,\mathcal{G})$ {are all self-adjoint. This is because, if $v \in L^0$:}
\begin{align*}
  \widehat{A_v}^\dagger=\frac{1}{\# G} \sum_{g \in G}  \widehat{U^g_v}^{\dagger} =\frac{1}{\# G} \sum_{g \in G}  \widehat{U^{g^{-1}}_v}=\frac{1}{\# G} \sum_{g \in G}  \widehat{U^{g}_v}=\widehat{A_v},\end{align*}
and analogously for  $\widehat{B_t}$.

From Lem. \ref{comm1} we have the following.
\begin{lemma}\label{comm2}
Given arbitrary vertices $u,v \in L^0$ and edges $s,t\in L^1$ we have:
\begin{align*}
&A_v A_v=A_v, & B_t B_t = B_t,\\
&[A_v,A_u]=0, & [A_v,B_t]=0, && [B_t,B_s]=0.
\end{align*}
These  relations also hold for $\widehat{ A_v}$ and $\widehat{B_t}$, since  $\CC \TcM$ acts on $\mathcal{H}(M,L,\mathcal{G})$.
\qed
\end{lemma}

\newcommand{\End}{{\mathrm End}}   % endomorphisms
{Let $b\in L^3$ be a 3-cell (i.e a blob; Rem. \ref{blob}).  Cf. \S \ref{2flatn},  by definition of 2-lattices, $\bound(b)$ is a subcomplex of $(M,L)$ homeomorphic to $S^2$, with a base-point $x_b=\base(b)$, which is a 0-cell, and an orientation. We can thus consider the 2D holonomy {${\rm Hol}^2_{\base(b)}(\bound(b),\F,L)\in \ker(\dG)\subset E$}, of  $\F \in \PhiM$ along $\bound(b)\cong S^2$.}
\begin{definition}[Blob operator] Let  $b\in L^3$. Let also {$a \in \ker(\dG) \subset E$.}
The diagonal idempotent blob operator
{$C_b^a \colon \mathcal{H}(M,L,\mathcal{G}) \to \mathcal{H}(M,L,\mathcal{G})$}
is given by the formula below, for each basis element $\F \in \PhiM$:
$$
{C_b^a(\F)=\delta\Big({\rm Hol}_{\base(b)}^2\big(\F,\bound(b),L\big),a\Big)\F.}
$$
Here if $e,e' \in E$ we put $\delta(e',e)=1$, if $e=e'$ and $\delta(e,e')=0$, if $e\neq e'.$
\end{definition}
\noindent{{See \cite[III-B]{BCKMM} and Ex. \ref{tet}, for the  explicit form of blob operators when our lattice $L$ is a triangulation of $M$.}}

 Since fake-flat 2-gauge configurations form an orthonormal basis of $\mathcal{H}(M,L,\mathcal{G})$, it is easy to see that each {$C_b^a \colon \mathcal{H}(M,L,\mathcal{G}) \to \mathcal{H}(M,L,\mathcal{G})$} is a self-adjoint operator.

As an immediate application of Thm. \ref{main5}, follows:
\begin{lemma}\label{comm3}
{Let $v\in L^0$, $t \in L^1$ and $b,b'\in L^3$ {(it may be that $b=b'$)}. Let $g\in G$ and $e \in E$. Let $a,a' \in \ker(\dG) \subset E$.} We have:

\begin{align*}
 &[\widehat{U_t^e},C_b^a]=0, && [C_{b'}^{a'},C_b^a]=0, & C_{b}^{a}C_{b}^{a'}=\delta(a,a')C_{b}^{a},\\
 & v \neq \base{(b)} \implies {[\widehat{U_v^g},C_b^a]=0},  && v =\base{(b)} \implies C_b^a\,\widehat{ U_v^g} =\widehat{U_v^g}  \, C_b^{g^{-1} \trr a}.
\end{align*} 
Hence edge gauge-spikes  $\widehat{U_t^e}$ always commute with blob operators $C_b^a$, regardeless of $t$ being an edge in $b$, or not. {A vertex gauge-spike $\widehat{U_v^g}$  commutes with a blob operator $C_b^a$, unless $v$ is the base point $\base(b)$ of $b$.}
\end{lemma}
\subsubsection{{The local operator algebra of higher lattice gauge theory}}\label{loca}
The algebra ${\cal OP}(M,L,\Gc)$, which underpins the construction of the higher Kitaev model in \ref{thkm}, is our proposal for the local operator algebra of higher lattice gauge theory.
\begin{definition}[{Local operator algebra for  higher lattice gauge theory}] \label{loa}{Let $(M,L)$ be a 2-lattice. Let $\Gc=(\dG\colon E \to G,\trr)$ be a crossed module of finite groups. We define the ${\mathbb{C}}$-algebra ${\cal OP}(M,L,\Gc)$ as formally generated by the}
\begin{align*}
 \widehat{U_v^g}, \quad &v \in L^0,\quad g \in G; & \\
 \widehat{U_t^e}, \quad &t =(\sigma(t) \ra{t} \tau(t)\big) \in L^1, \quad  e \in E; \\
C_b^a, \quad &b \in L^3, \quad  a \in \ker(\dG);
\end{align*}
{imposing} the relations appearing in Lem. \ref{comm1} and \ref{comm3}. 
\end{definition}
 Note that  ${\cal OP}(M,L,\Gc)$ is a *-algebra, where: 
\begin{align*}
&\widehat{U_v^g}^\dagger=\widehat{U_v^{g^{-1}}}, && \widehat{U_t^e}^{\dagger}= \widehat{U_t^{e-1}},&& (C_b^a)^{\dagger}=C_b^a.
\end{align*} 
{Given the discussion in \ref{exp} and \ref{veb}, we hence  have a unitary representation of  ${\cal OP}(M,L,\Gc)$ on the Hilbert space $\mathcal{H}(M,L,\mathcal{G})$ of higher lattice gauge theory.}

\subsubsection{{The higher Kitaev model for (3+1)-dimensional topological phases }}\label{thkm}
We now propose a  higher gauge theory version (the ``higher Kitaev
model'') of Kitaev quantum-double model for (2+1)-dimensional
topological phases of matter
\cite{kitaev2003fault,hu2013twisted}. This higher Kitaev model for
(3+1)-dimensional topological phases is formulated for manifolds
$M$, of any dimension, with a 2-lattice decomposition $L$; see
Def. \ref{2-lattice}.
{For a description of higher Kitaev model in
the particular case of triangulated manifolds we refer the reader
to \cite{BCKMM},
{and to \cite{williamson2016hamiltonian}, in a more general context}. Topological phases protected by higher gauge symmetry are also proposed in \cite{C}.} 
\begin{definition}[{Higher} Kitaev model]\label{hkm}
{(Cf. the notation in \ref{veb}).} {Let ${\cal G}=(\dG\colon E \to G,\trr)$ be a finite crossed module of groups. {Let $M$ be a compact topological manifold, of any dimension, with a 2-lattice decomposition $L$.}
Our proposal for a totally solvable (the sum of mutually commuting projection operators) higher lattice gauge theory
Hamiltonian, {which we call the ``higher Kitaev model''}:}
$$
H_L\colon\H(M,L,\Gc)\to \H(M,L,\Gc)
$$
(where $\H(M,L,\Gc)$ is as in Def. \ref{hs}), with respect to the 2-lattice {$(M,L)$} is {(where $1_E$ is the identity of $E$)}:
\begin{equation}\label{hamiltonian}
\begin{split}
 H_L&=\sum_{v \in L^0} \left ({\rm id} -\widehat{ A_v} \right) +\sum_{t \in L^1} \left ( {\rm id}-\widehat{B_t} \right) +  \sum_{b \in L^3} \left ({\rm id} -C_b^{1_E}  \right)\\
&=\sum_{v \in L^0}  { \A_v} +\sum_{t \in L^1} {\B_t} +  \sum_{b \in L^3}  \C_b   \\
&={\cal A} + {\cal B} + \C.
 \end{split}
\end{equation}
 The commutation relations {of} Lem. \ref{comm1}, \ref{comm2} and \ref{comm3}, ensure that, if $u,v \in L^0$, $t,s \in L^1$ and $b,c \in L^3$:
\begin{align}
&\A_v \A_v=\A_v, & \B_t \B_t = \B_t, && \C_b \C_b = \C_b,\nonumber\\
&[\A_v,\A_u]=0,  &[\B_t,\B_s]=0, &&[\C_b,\C_c]=0,  \label{a1}\\ 
& [\A_v,\B_t]=0, & [\A_v,\C_b]=0, && [\B_t,\C_b]=0, \nonumber
\end{align}
{(Observe that these relations also hold for the  $\widehat{ A_v}$, $\widehat{B_t}$ and $C_b^{1_E}$.)}
And moreover we have that:
\begin{align}
& [{\cal A},{\cal B}]=0,  & [{\cal A},\C]=0,&& [{\cal B},\C]=0. \nonumber\\
&{\cal A}^2={\cal A}, &{\cal B}^2={\cal B} &&{\cal C}^2={\cal C}.
\end{align}
Note that by construction each term {${\cal A}_v,{\cal B}_t, {\cal C}_b\colon \H(M,L,\Gc)\to \H(M,L,\Gc)$} is Hermitian, hence so is $H_L$.
\end{definition}
\noindent {Typically $M$ will be a 3-dimensional manifold, and the higher Kitaev model should be considered to be a model for (3+1)-dimensional topological phases \cite{WalkerWang,williamson2016hamiltonian,wan2015twisted,BCKMM,Simon,kong2014braided,SimonPhase,C,D}. The higher Kitaev model also makes sense if $M$ is a surface, but in this case  blob operators $C^{1_E}_b$ will not appear in the model.}

\medskip
{Note that vertex and edge operators, which implement gauge
invariance at vertices and edges of a 2-lattice, and the blob operators, which enforce 2-flatness at a blob, are very different in nature.}

\begin{lemma}\label{defofzeta}
Let $(M,L)$ be a 2-lattice. Let {$\Gc=(\dG \colon E \to G,\trr)$ be a finite crossed module of groups}.
Let  $\F\in \PhiM$.   Let
$$
\zeta_L(\F)=\#\{ b \in L^3: {\rm  Hol}_{\base(b)}^2\big(\F,\bound(b),L\big)\neq 1_E\}\in \mathbb{Z}^+_0.
$$
 Then $\zeta(\F)$ is invariant under full gauge transformations and in
 particular it is invariant under {the action of} vertex and edge gauge
 {operators}. Moreover $\C (\F)=\zeta(\F) \F$.
\end{lemma}
\begin{proof}
 The first bits follow from Thm. \ref{main5}: given $b\in L^3$, then ${\rm  Hol}_{\base(b)}^2\big(\F,\bound(b),L\big)$ is invariant under the action of gauge operators, up to acting by an element of $G$, which acts on $E$ by automorphisms. On the other hand, the fact that $\C (\F)=\zeta(\F) \F$ follows from the definition of {$\C\colon\H(M,L,\Gc)\to \H(M,L,\Gc)$}.\end{proof}

\subsubsection{{Example: higher gauge theory in the 3-sphere}}

Let us give an explicit description of the higher Kitaev model, if the underlying manifold is $S^3$. We  consider two different 2-lattice decompositions of $S^3.$ {Let ${\cal G}=(\d \colon E \to G,\trr)$ be a finite crossed module of groups.}

\medskip

 \noindent {\bf First case: $(S^3,L_0)$}\\ Cf. Ex. \ref{LandL0} and \ref{L0}. Consider the 3-sphere $S^3$ with the lattice decomposition $L_0$, with a unique $0$-cell, no 1-cells, one 2-cell (thus the 2-skeleton is $S^2$) and two blobs, attaching on each side of the 2-sphere. 
 By Ex. \ref{L0}, the Hilbert space $\H(S^3,L_0,\Gc)$ is thus isomorphic to $\mathbb{C}\ker(\partial)$, the  vector space generated by the orthonormal basis $\ker(\partial)\subset E$. The 2D holonomy along the 2-sphere of a fake-flat 2-gauge  configuration associated with $m \in \ker(\partial)$ is $m$ itself:  ${\rm Hol}^2_v(m,S^2,L_0)= m$.

 For this 2-lattice we have no edge operators. The higher Kitaev Hamiltonian $H_{L_0}\colon \H(S^3,L_0,\Gc) \to \H(S^3,L_0,\Gc)$ therefore has the form: 
$H_{L_0}=\A_{L_0} +  \C_{L_0}$, where:
\begin{align}
\A_{L_0} m&
    =m-\frac{1}{\#G}\sum_{a \in G} ( a \trr m), \\
\C_{L_0} m&=m- \delta(m,1_E)\, m. 
\end{align}
\medskip

\noindent {\bf Second case: $(S^3,L_{\mathfrak{g}})$}\\
For a more substantial example, let us give $S^3$ the 
2-lattice decomposition:
$L_{\mathfrak{g}}=(\{ v\}, \{ t \}, \{ P,P'\}, \{b,b' \})$ 
of Ex. \ref{LandL0}. A  2-gauge configuration is given by a $g=g_t \in G$ and a pair $(e=e_P,f=e_{P'})\in E \times E$. The fake flateness condition enforces that $\d(e)=\d(f)=g$.
Therefore, we have:
$$
\H(S^3,L_{\mathfrak{g}},\Gc)
   \; =  \;\; \CC  %-{\rm span}
\{(g,e,f) \in G \times E \times E \;
               \colon \; \d(e)=g \textrm{ and } \d(f)=g\}.
$$
The 2D holonomy of
a configuration
$\F = (g,e,f)$, along the
2-sphere $S^2 \subset S^3$, based at $v$, is:  
$$
{\rm Hol}^2_v(\F,S^2,L_{\mathfrak{g}})= e^{-1}f\in \ker(\partial).
$$
The vertex and edge gauge spikes on $v$ and $t$, and the blob operators along $b$ and $b'$, have the form:
\begin{align}
 \widehat{U_v^a}(g,e,f)&=\left ( aga^{-1}, a \trr e, a \trr f \right),\\
 \widehat{U_t^k}(g,e,f)&=\left ( \d(k)g, k  e, k f \right),\\
{C^{k'}_b}(g,e,f)&=\delta(k',e^{-1}f)(g,e,f),\\
{C^{k'}_{b'}}(g,e,f) &=\delta(k', e^{-1}f)(g,e,f).
 \end{align} 
{(Cf. \S \ref{exp}.)}
Here $a \in G$, $k\in E$ and $k' \in\ker(\partial) \subset E$. Note that $C^{k'}_b=C^{k'}_{b'}$, for each $k' \in \ker(\partial)$.

The commutation relations of Lem. \ref{comm1}
and \ref{comm3}  %. These latter relations
here boil down to:
\begin{align*}
\widehat{U^a_v} \widehat{U^{a'}_v}&=\widehat{U^{aa'}_v}, &\textrm{ where }\;\; a,a'\in G; &&
\\ 
\widehat{U^k_t} \widehat{ U^{l}_t}&=\widehat{U^{kl}_t}, &\textrm{ where }\;\; k,l \in E; &&
\\ 
\widehat{U^k_t}\widehat{ U^a_{v}} &=\widehat{U^a_{v}}  \widehat{ U^{a^{-1} \trr e}_t}, & \textrm{ where } \;\; k \in E, \; a \in G; &&
\\
C_b^{k'} \widehat{U_v^a} &=\widehat{U_v^a}  C_b^{a^{-1} \trr k'}, & \textrm{ where }  \;\; k' \in \ker(\partial), \; a \in G; &&
\\
C_b^{k'} \widehat{U_t^l} &=\widehat{U_t^l}  C_b^{k'}, & \textrm{ where } \;\; l \in E,  \; k' \in \ker(\partial); &&\\
C_b^{k'} C_b^{k''}&=\delta(k',k'')C_b^{k'},&  \textrm{ where } \; k',k'' \in \ker(\partial). &&
\end{align*}
{This gives the local operator algebra ${\cal OP}(S^3,L_{\mathfrak{g}},\Gc)$; see Def. \ref{loa}.}

The  higher Kitaev hamiltonian $H_{L_{\mathfrak{g}}}\colon \H(S^3,L_{\mathfrak{g}},\Gc) \to \H(S^3,L_{\mathfrak{g}},\Gc)$ has the form 
$H_{L_{\mathfrak{g}}}=\A + \B+ \C$, where:
\begin{align*}
\A (g,e,f)&
    =(g,e,f)-\frac{1}{\#G}\sum_{a \in G} (aga^{-1}, a \trr e, a \trr f), 
%\label{ABCA} 
\\
\B (g,e,f)&=(g,e,f)-\frac{1}{\#E}\sum_{k \in E} (\d(k) g, k e, k f), 
%\label{ABCB} 
\\
\C (g,e,f)&=(g,e,f)- \delta(e^{-1}f,1_E)\, ( g,e, f). 
%\label{ABC}
\end{align*}
\subsubsection{{Comparison with the Kitaev model}}\label{comparison}
Though constructed in a similar way, the {higher} Kitaev model and the Kitaev model \cite{kitaev2003fault} {(also known as Kitaev quantum double model)} are subtly different constructions. In the following, we will demonstrate that a subspace of the Kitaev model is equivalent to a class of {higher} Kitaev models, while differing in the whole Hilbert space. 

In the language of this paper, the Kitaev model takes as input: a 2-lattice $(M,L)$ and a finite group $G$, realising a lattice model with local operator algebra, {which is at  the base point of each plaquette isomorphic to $\mathcal{D}(G)$,} the quantum double of $G$ \cite{kitaev2003fault}. The Hilbert space $\H_K\big(M,L,G\big)$ of the Kitaev model is the free vector space on the set of gauge configurations $\F^1\colon L^1 \to G$; {see Def. \ref{gconfig}}. Considering the group $G$ as the crossed module, $(1 \to G)$ (see Ex. \ref{gtxm}), it follows:
$$\H_K\big(M,L,G\big)=\H\big (M^1,L,(1 \to G)\big).$$
{Here $M^1$ is the 1-skeleton of $(M,L)$}. Note the use of $M^{1}$ as opposed to $M$, {so that the fake-flatness condition becomes void}. The Kitaev model is defined by the Hamiltonian
$$
H^{K}_{L}=\A+\mathcal{D}:\mathcal{H}_{K}\big(M,L,G\big)\rightarrow \mathcal{H}_{K}\big(M,L,G\big).
$$
Here the operator $\A=\sum_{v \in L^0} (\id - \widehat{A_v})$ is as in \ref{veb}, defined for the crossed module $(1\to G)$, with action on $\mathcal{H}_{K}\big(M,L,G\big)$ {given by} the action on $\H\big (M^1,L,(1 \to G)\big)$. Whereas, $\mathcal{D}=\sum_{P \in L^2}(\id-D^{1_{E}}_{P})$ is {defined from  a new type of self-adjoint operators $D^g_{P}$, which act} on gauge configurations $\F^{1}:L^{1}\rightarrow G$ as follows {(for notation see Def. \ref{holalongcircle}):}
$${D^g_{P}(\F^1)=\delta\Big({\rm Hol}_{\base(P)}^1\big(\F^1,\bound(P),L\big),g\Big)\F^1, \textrm{ where } \base(P) \textrm{ is the basepoint of } P.}$$
\begin{lemma}
	Given $v,v'\in L^{0}$ and $P,P'\in L^{2}$ the following relations hold:
	\begin{align}
	[\A_{v},\A_{v'}]=&0,\qquad \A_{v}\A_{v}=\A_{v};\nonumber\\
	[D^{1_{E}}_{P},D^{1_{E}}_{P'}]=&0,\qquad D^{1_{E}}_{P}D^{1_{E}}_{P}=D^{1_{E}}_{P};\nonumber\\
	[\A_{v},D^{1_{E}}_{P}]=&0.\nonumber
	\end{align}
	Hence $H^{K}_{L}$ is a sum of mutually commuting projection operators.
\end{lemma}

We now compare the Kitaev model with group $G$ to the higher Kitaev model with crossed module $(1 \to G)$, with fixed 2-lattice $(M,L)$. We begin by defining the flat sub-Hilbert space of the Kitaev model $\mathcal{H}^{\text{flat}}_{K}(M,L,G)\subsetneq \mathcal{H}_{K}(M,L,G)$:
$$
\mathcal{H}^{\text{flat}}_{K}(M,L,G)=\{
\F^{1}\in \mathcal{H}_{K}(M,L,G)\quad|\quad \prod_{P\in L^{2}}D^{1_{E}}_{P}(\F^{1})=\F^{1}
\}.
$$
It is straightforward to show
$$
\mathcal{H}^{\text{flat}}_{K}(M,L,G)=\mathcal{H}(M,L,(1\to G))\subsetneq \mathcal{H}_{K}(M,L,G).
$$
This is due to the requirement of fake-flat 2-gauge configurations of $(1\to G)$ on $(M,L)$ being an equivalent condition to requiring $\prod_{P\in L^{2}}D^{1_{E}}_{P}(\F^1)=\F^1$.

For the crossed module $(1\to G)$, both the blob and edge operators act as the identity. In this way the {higher} Kitaev Hamiltonian reduces to:
\begin{align*}
H_{L}=\A \colon &\H\big (M,L,(1 \to G)\big) \to  \H\big (M,L,(1 \to G)\big).
\end{align*}
This model is equivalent to the Kitaev model defined on the flat sub-Hilbert space $\H^{\text{flat}}_{K}\big (M,L,G\big)$:
\begin{align*}
H_{K}=\A+\D\colon&\H^{\text{flat}}_{K}\big (M,L,G\big) \to  \H^{\text{flat}}_{K}\big (M,L,G\big).
\end{align*}
This is because, by definition, the operator $\D$ has trivial action on $\H^{\text{flat}}_{K}\big (M,L,G\big)=\H\big (M,L,(1 \to G)\big)$, while the $\A$ operator has the same action on both Hilbert spaces. In this way we can identify the higher Kitaev model for $(1\to G)$ with the Kitaev model defined on the sub-Hilbert space $\H^{\text{flat}}_{K}\big (M,L,G\big)$. However, the Kitaev model diverges from the higher Kitaev model for $(1\to G)$ outside of the sub-space $\H^{\text{flat}}_{K}\big (M,L,G\big)$ due to the presence of non-fake-flat configurations in $H_{K}\big(M,L,G \big)$.

\subsection{{Ground state degeneracy}}\label{gsd}
 Let $M$ be a compact manifold and {$\Gc=(\d \colon E \to G,\trr)$ be a finite group crossed module}. Let $L$ be a 2-lattice decomposition of $M$. Hence $\H(M,L,\Gc)$ is a finite dimensional Hilbert space, which explicitly depends on the  2-lattice decomposition $L$ of $M$.
In this subsection, we prove that the dimension of the ground
state space $GS(M,L,\Gc)$ of the higher Kitaev model $H_L\colon \H(M,L,\Gc) \to \H(M,L,\Gc)$ in \ref{thkm}  is a topological
{invariant} of $M$, meaning that $\dim GS(M,L,\Gc)$ depends only on $M$ alone. Specifically, we  will show that $GS(M,L,\Gc)$ has a basis in canonical one-to-one correspondence  with the set of homotopy classes of maps from $M$ to the
classifying space $B_\Gc$ of the  crossed module $\Gc$. (Classifying spaces of crossed modules are defined in \cite[\S 2.4]{brown_higgins_sivera} and \cite{brown_classifying,brown_hha,martins_porter,martins_cw_complex}.)
It therefore
follows that the dimension of the ground state {space} $GS(M,L,\Gc)$ is a homotopy
invariant of manifolds, as expected given the relation \cite{BCKMM} of our model
to Yetter's invariant of manifolds \cite{yetter_tqft,porter_tqft}.
Yetter's invariant was proven in \cite{martins_porter,martins_cw_complex} to be a
homotopy invariant of manifolds.

 {This subsection is less self-contained than the remainder of the paper. Cf. Rems. \ref{bh2} and \ref{bh3}. Let $X$ be a CW-complex.
 We use  deep results of Brown and Higgins on the description of the weak homotopy type of the function space $TOP(X,B_{\cal X})$, where ${\cal X}$ is a crossed complex (a generalisation of crossed modules) and  $B_{\cal X}$ is its classifying space; the bits we need can be found in \cite[Thm. A]{brown_classifying} and \cite[Thm. 11.4.19]{brown_higgins_sivera}. Let  $\Pi(X)$ denote the fundamental crossed complex of $X$. The main tool we use is the fact that the weak homotopy type of $TOP(X,B_{\cal X})$, with the $k$-iffication of the compact-open topology {on the space of continuous maps $X \to B_{{\cal X}}$}, is represented by the crossed complex ${\rm CRS}(\Pi(X),{\cal X})$; an explanation of this is in \cite[\S 2.6.1]{martins_porter}. The crossed complex ${\rm CRS}(\Pi(X),{\cal X})$ is made of crossed complex maps $\Pi(X) \to {\cal X}$ and  $n$-fold homotopies, $n\in\N$.}

{For a crossed module $\Gc$,  the underlying groupoid of the crossed complex ${\rm CRS}(\Pi(X),\Gc)$ is the groupoid   ${\rm CRS}_1(\Pi_2(X,X^1,X^0),\Gc)$ of crossed module maps $\Pi_2(X,X^1,X^0) \to \Gc$} and their homotopies, referred to in Rem \ref{bh}, \ref{bh2} and \ref{bh6}; see \cite[\S 7.1.vii and \S 9.3.i]{brown_higgins_sivera}. Combining with Thm \ref{2-flat-conf-theorem}, it hence follows that if $(M,L)$ is a 2-lattice, then the underlying groupoid of ${\rm CRS}(\Pi(X),\Gc)$ is isomorphic to the groupoid  $\Theta^\#_{{\rm flat}}(M,L,\Gc)$ of 2-flat 2-gauge configurations and full gauge transformations between them: see \ref{2flatgroup}.

 {Given a 2-lattice $(M,L)$, the results of \cite{brown_classifying} and \cite[\S 11.4]{brown_higgins_sivera} hence tell us that path-connected components of the function space $TOP(M,B_\Gc)$ -- i.e. homotopy classes of maps $M \to B_\Gc$ -- are in canonical one-to-one correspondence with connected components of the groupoid $\Theta^\#_{{\rm flat}}(M,L,\Gc)$ of {2-flat} 2-gauge configurations and full gauge transformations between them;  see \S\ref{2flatgroup}. {This  will be the main tool used in this subsection.}}

\medskip

{Let us now connect the discussion in the previous paragraphs with the ground state degeneracy of the higher Kitaev model.} We start by looking at the  expression \eqref{hamiltonian} for $H_L\colon \H(M,L,\Gc) \to \H(M,L,\Gc)$:
$$
H_L=\sum_{v \in L^0} \left ({\rm id} -\widehat{ A_v} \right) +\sum_{t \in L^1}
\left ( {\rm id}-\widehat{B_t} \right) +  \sum_{b \in L^3} \left ({\rm id}
-C_b^{1_E}  \right).
$$
 Each of the operators  ${\rm id} -\widehat{ A_v}$, where $v \in L^0$ is a
 vertex;
${\rm id} - \widehat{B_t}$, where $t \in L^1$ is an edge; and
${\rm id}-C_b^{1_E}$, where $b \in L^3$ is a blob, is a Hermitian projector.
All of those projectors commute. We can choose an
 eigenspace decomposition $\H(M,L,\Gc)=\sum_{i}^\perp \H_i(M,L,\Gc)$ {with respect to which}  all
 those projectors are diagonal.

If we apply $\widehat{ A_v}$ or $\widehat{B_t}$ to  $\F\in \Theta(M,L,\Gc)$, we get a
non-zero linear combination of fake-flat 2-gauge configurations with non-negative coefficients.
Let us write $\Hc^+$ for the subset of non-zero
$\RR_{\geq 0}$-linear combinations in {$\Hc=\H(M,L,\Gc) $.}
{Recall the naive vacuum
$\vac$ is given by $\F(x) = 1$ for all cells $x$; see \S\ref{ff2gc} and Ex. \ref{nv}.}
Since $\widehat{ A_v}$ and $\widehat{B_t}$ both take an $\F$ to an element of $\Hc^+$,
and indeed take   an element of $\Hc^+$ to  an element of $\Hc^+$,
we have that
$\Psi_{0} = \prod_v \widehat{ A_v} \prod_t \widehat{B_t} \vac \in \Hc$ is non-zero.  Since $C_b^{1_E} \vac = \vac$, for each $b \in L^3$, it follows that there is an $H_L$-eigenspace
of $\Hc$  (containing $\Psi_{0}$) with eigenvalue 0. {(Note that $\Psi_0$ is in the kernel of all of the ${\rm id} -\widehat{ A_v}$,  ${\rm id} - \widehat{B_t}$ and
 ${\rm  id}-C_b^{1_E}$, where $v\in L^0$, $t \in L^1$ and $b \in L^3$, given the commutation relations in \eqref{a1}.)}

Now projectors have eigenvalues $0$ or $1$,
thus the ground state has energy
zero, meaning: $$GS(M,L,\Gc)=\{\Psi\in \PhiM: H_L \Psi=0\}.$$ And, furthermore,  a vector
 belongs to the ground state $GS(M,L,\Gc)$ if, and only if, it is in the kernel of
 all of the projectors ${\rm id} -\widehat{ A_v}$,  ${\rm id} - \widehat{B_t}$ and
 ${\rm  id}-C_b^{1_E}$, where $v\in L^0$, $t \in L^1$ and $b \in L^3$.

\begin{lemma}\label{tech}
 A state {$\displaystyle\Psi=\sum_{\F \in \Theta(M,L,\Gc)} \lambda_\F
   \F\in \H(M,L,\Gc)$,}
 where $\lambda_\F\in \mathbb{C}$, is in $GS(M,L,\Gc)$  if and only if:
 \begin{enumerate}[(i)]
  \item unless $\F$ is 2-flat then $\lambda_\F=0$; see Def. \ref{2flatconf} for the definition of a 2-flat configuration;
  \item given any $g \in G$, any vertex $v \in L^0$ and any
$\F \in \Theta(M,L,\Gc)$, it holds $\l_{\F}=\l_{\widehat{U_v^g}( \F)}$;
  \item given any $e \in E$, any edge $t \in L^1$ and any
$\F \in \Theta(M,L,\Gc)$, it holds $\l_{\F}=\l_{\widehat{U_t^e}(\F)}$.

  \end{enumerate}
\end{lemma}
\begin{proof}
  First the `only if' part. Cf. the discussion just before the Lemma. In order that $\Psi\in GS(M,L,\Gc)$ it must be that (i) $C_b^{1_E}(\Psi)=\Psi, \forall b \in L^3$;
that  (ii) $\widehat{ A_v} \Psi =  \Psi, \forall v \in L^0$; and that (iii) $\widehat{B_t} \Psi= \Psi, \forall t \in L^1$.
\begin{enumerate}[(i)]
\item
 For each  $b \in L^3$,
$C_b^{1_E}(\Psi)=
  \displaystyle\sum_{\F \in \Theta(M,L,\Gc)}
  \lambda_\F \delta\Big({\rm Hol}_v^2\big(\F,\bound(b)\big),1_E\Big)\F $.
In order that $C_b^{1_E}(\Psi)=\Psi, \forall b \in L^3$, it must be that whenever $\l_\F\neq 0$:
  $\delta\Big({\rm Hol}_v^2\big(\F,\bound(b)\big),1_E\Big)=1, \forall b \in L^3$. Hence $\l_\F\neq 0 \implies \F$ is 2-flat.

\item Suppose that
  $\widehat{ A_v}(\Psi)=\Psi$, for all $v\in L^0$. Then:
  \begin{align*}
    \Psi&=\sum_{\F' \in \Theta(M,L,\Gc)} \lambda_{\F'} \F'
    =\widehat{ A_v}(\Psi) \;
    =\frac{1}{\#G} \sum_{\F' \in \Theta(M,L,\Gc)} \Big (
    \sum_{h \in G} \l_{\F'}  \widehat{U_v^{h}} (\F') \Big).
  \end{align*}
  Now apply $\langle \F, - \rangle$.
  And since $\widehat{U_v^{gh}}(\F)=\widehat{U_v^{g}}  \big(\widehat{ U_v^{h}} (\F)\big)$,
  we have for each $v \in L^0$ and any $g\in G$:
  \begin{align*}
    \l_\F=\frac{1}{\#G}\sum_{h \in G} \l_{\widehat{U_v^{h^{-1}}}(\F)}
    =  \frac{1}{\#G} \sum_{h \in G} \l_{\widehat{U_v^{h^{-1}}}(\widehat{U_v^g}( \F))}=\l_{\widehat{U_v^g}( \F)}.  \end{align*}

\item Analogously, it order that
$\Psi$ be in the kernel of all $(\id-\widehat{B_t})$,
then $\l_{\F}=\l_{\widehat{U_t^e}(\F)}$.
\end{enumerate}
Conversely, if $\Psi$ satisfies (i),(ii) and (iii), then $\Psi$ will be in the kernel of all  operators  ${\rm id} -\widehat{ A_v}$, ${\rm id} - \widehat{B_t}$  and
${\rm id}-{C_b^{1_E}}$, where $v\in L^0,t\in L^1$ and $b \in L^3$. As such $\Psi$ will be in the ground state $GS(M,L,\Gc)$.
 \end{proof}

Let $\fPhi(M,L,\Gc)$ be the set of 2-flat 2-gauge configurations {in $(M,L)$}; Def. \ref{2flatconf}.
We say that $\F$ and $\F'$ in $\fPhi(M,L,\Gc)$ are equivalent ($\F\cong \F'$) if we can go from $\F$ to $\F'$ through acting by  a
sequence of vertex and edge gauge spikes; {see \ref{exp}.}
 {By Lem. \ref{tech}: we hence have:}
 $${GS(M,L,\Gc)=\Big\{\sum_{\F \in \fPhi(M,L,\Gc)} \lambda_\F \F \in \H(M,L,\Gc) \,\,\,\Big{|}\,\,\, \forall \F,\F' \in \fPhi(M,L,\Gc): \F\cong \F' \implies \l_{\F}=\l_{\F'}\Big \}.}$$

The composite of gauge spikes is a full gauge transformation and,
Lem. \ref{comm1} (V),
any full gauge transformation is the composition of gauge spikes; as such to say that $\F$ and $\F'$ %, flat and fake-flat gauge configurations,
in $\fPhi(M,L,\Gc)$ are
equivalent is to say that they are connected  by a full gauge transformation. In other words $\F$ and $\F'$ are equivalent if, and only if, they can be connected by {a morphism in the groupoid {$\fPhi^{\sharp}(M,L,\Gc)$},}
of  2-flat {2-gauge} configurations and full gauge transformations between them; see \S\ref{2flatgroup}.

The set of connected components $[\F']$ of the groupoid {$\fPhi^{\sharp}(M,L,\Gc)$} is denoted by
$
\pi_0\big(\fPhi^\sharp(M,L,\Gc)\big).
$ I.e. $\pi_0\big(\fPhi^\sharp(M,L,\Gc)\big)$ is the set of equivalence classes of objects of
$\fPhi^\sharp(M,L,\Gc) $,
where two 2-flat 2-gauge configurations are equivalent
if a full gauge transformation connects the two.
We  have a basis ${\mathsf B}_0(M,L,\Gc)$ for the ground space $GS(M,L,\Gc)$, in one-to-one correspondence with  {$\pi_0\big(\fPhi^\sharp(M,L,\Gc)\big)$, namely:}
\begin{equation} \label{BGSdef}  
{{{\mathsf B}_0(M,L,\Gc) =
\Big\{ \sum_{\F \in \fPhi(M,L,\Gc) \textrm{ such that }  \F \in [\F']} \F\quad  \Big | [\F']\in
  \pi_0\big(\fPhi^\sharp(M,L,\Gc)\big)\Big\}.}}
\end{equation}  %$$

 As explained in
{\cite[Thm. A]{brown_classifying}
or \cite[Thm. 11.4.19]{brown_higgins_sivera}}
(in the general case of crossed complexes),
there is a natural bijection between elements of
$\pi_0\big(\fPhi^\sharp(M,L,\Gc)\big)=\pi_0({\rm CRS}_1\big(\Pi_2(M,M^1,M^0),\Gc)\big)$
and homotopy classes of maps
$M \to B_\Gc$; cf. Rems. \ref{bh2} and \ref{bh6}.
(For an explanation {of this}  in the crossed module case see
\cite{martins_porter,martins_cw_complex}.)
In particular the cardinality of {the set $\pi_0\big(\fPhi^\sharp(M,L,\Gc)\big)$}
does not depend on  $L$. 

\begin{Theorem}{
 Let $M$ be a compact manifold. Let $L$ be a {2-lattice} decomposition
 $M$. Let {$\Gc$ be a finite crossed module}. Consider the {higher Kitaev} Hamiltonian
 $H_L \colon \H(M,L,\Gc) \to \H(M,L,\Gc)$ of {Def. \ref{hkm}.} Then the
 ground state $GS(M,L,\Gc)$ of $H_L$ 
 has a basis in canonical one-to-one correspondence with the set of homotopy
 classes of maps $f\colon M \to B_\Gc$, where $B_\Gc$ is the
 classifying space of the crossed module $\Gc$. Hence
 $\dim GS(M,L,\Gc)$ depends only on the topology of $M$ and not on the chosen 2-lattice decomposition $L$ of $M$.}
\end{Theorem}
\proof
Compare the preceding observation with (\ref{BGSdef}).
\qed

\begin{remark}
 It was proven in \cite{BCKMM} that the dimension of the ground state
 $GS(M,L,\Gc)$  coincides with Yetter invariant \cite{yetter_tqft,porter_tqft} $Y(M\times S^1)$ of $M \times S^1$.
We note that the Yetter invariant of $M\times S^1$ is
 not quite the same as $\dim GS(M\times S^1,L,\Gc)$ since $Y(M\times S^1)$ uses finer
 information on the space of functions $M\times S^1\to B_\Gc$ than its
 number of connected components; see \cite{martins_porter}.
\end{remark}
\begin{remark}
We can write $\dim GS(M,L,\Gc)$ as $\dim GS(M,\Gc)$, since it does not
depend on $L$, {and only on $M$. Indeed $\dim GS(M,\Gc)$ depends only on the
homotopy type of $M$ and the weak homotopy type of the crossed module
$\Gc$ \cite{martins_porter}, since the homotopy type  of $B_\Gc$  depends only on the weak homotopy type of $\Gc$.}
\end{remark}

%\subsection{{Higher gauge theory in the 3-sphere}}\label{hgt3s}
%\input{hgt3sF}
\subsubsection*{Acknowledgement}
We would like to thank a referee for a previous version of this paper for useful comments. {Alex Bullivant, Zoltan K\'ad\'ar, and Paul Martin thank
  EPSRC for funding under Grant EP/I038683/1.

\bibliographystyle{plain}

\bibliography{HGT2.bib}

\end{document}